\documentclass[
 prx,
 reprint,
 superscriptaddress,
 longbibliography,
 dvipsnames
]{revtex4-2}

\pdfoutput=1

\usepackage[utf8]{inputenc}
\usepackage{amsmath}
\usepackage{amssymb}
\usepackage{amsthm}
\usepackage{physics}
\usepackage{graphicx}
\usepackage{xcolor}
\usepackage{hyperref}
\usepackage{xspace}

\usepackage{xstring} 
\usepackage{mathtools} 
\usepackage{bm}
\usepackage{dsfont} 
\usepackage{textgreek} 
\usepackage[group-separator={,}, group-minimum-digits=3]{siunitx} 
\usepackage{booktabs} 
\usepackage{multirow}

\usepackage{header}

\begin{document}

\title{Low-overhead magic state distillation with color codes}

\author{Seok-Hyung Lee}
\email{seokhyung.lee@sydney.edu.au}
\affiliation{Centre for Engineered Quantum Systems, School of Physics, The University of Sydney, Sydney, New South Wales 2006, Australia}

\author{Felix Thomsen}
\affiliation{Centre for Engineered Quantum Systems, School of Physics, The University of Sydney, Sydney, New South Wales 2006, Australia}

\author{Nicholas Fazio}
\affiliation{Centre for Engineered Quantum Systems, School of Physics, The University of Sydney, Sydney, New South Wales 2006, Australia}

\author{Benjamin J. Brown}
\affiliation{IBM Quantum, T. J. Watson Research Center, Yorktown Heights, New York 10598, USA}
\affiliation{IBM Denmark, Sundkrogsgade 11, 2100 Copenhagen, Denmark}

\author{Stephen D. Bartlett}
\email{stephen.bartlett@sydney.edu.au}
\affiliation{Centre for Engineered Quantum Systems, School of Physics, The University of Sydney, Sydney, New South Wales 2006, Australia}

\begin{abstract}
Fault-tolerant implementation of non-Clifford gates is a major challenge for achieving universal fault-tolerant quantum computing with quantum error-correcting codes.
Magic state distillation is the most well-studied method for this but requires significant resources.
Hence, it is crucial to tailor and optimize magic state distillation for specific codes from both logical- and physical-level perspectives.
In this work, we perform such optimization for two-dimensional color codes, which are promising due to their higher encoding rates compared to surface codes, transversal implementation of Clifford gates, and efficient lattice surgery.
We propose two carefully designed distillation schemes based on the 15-to-1 distillation circuit and lattice surgery, differing in their methods for handling faulty rotations.
Our first scheme employs faulty T-measurement, achieving infidelities of $O(p^3)$ for physical noise strength $p$.
To achieve lower infidelities, our second scheme integrates distillation with `cultivation' (a distillation-free approach to fault-tolerantly prepare magic states through transversal Clifford measurements).
Our second scheme achieves significantly lower infidelities (e.g., $\sim 2 \times 10^{-16}$ at $p = 10^{-3}$), surpassing the capabilities of both cultivation and single-level distillation.
Notably, to reach a given target infidelity, our schemes require approximately two orders of magnitude fewer resources than the previous best magic state distillation schemes for color codes.
\end{abstract}

\maketitle

\section{Introduction \label{sec:introduction}}

A large-scale quantum computer should possess two properties: \textit{universality} and \textit{fault tolerance}. 
Universality refers to the capability of applying any unitary operations on qubits with arbitrary accuracy. 
This is achievable if the computer can implement a universal set of gates, which commonly includes the Hadamard gate ($H$), phase gate ($S = Z^{1/2}$), controlled-\nnot (\cnot) gate, and T gate ($T = Z^{1/4}$) \cite{nielsen2010quantum}. 
While there is flexibility in this choice of gate set, for universality it must contain at least one non-Clifford gate, such as the T gate in this example. 
Fault tolerance is essential due to the noisy physical environments that may corrupt quantum information stored in qubits. 
Stabilizer quantum error-correcting (QEC) codes can be employed for fault-tolerant quantum memories (i.e., idling operations), provided that the physical noise strength is sufficiently low. 
However, implementing every gate in a universal set of gates fault-tolerantly remains a far more challenging task.

Two-dimensional (2D) color codes \cite{bombin2006topological,bombin2013topological} are a family of stabilizer QEC codes defined on a trivalent and 3-colorable lattice of qubits, which can be constructed using only local interactions. 
Compared with surface codes \cite{bravyi1998quantum,dennis2002topological}, color codes use fewer physical qubits per logical qubit at a given code distance \cite{landahl2011faulttolerant} and have richer topological structures \cite{kesselring2022anyon}. 
This allows an arbitrary pair of multi-qubit Pauli operators to be measured in parallel via lattice surgery \cite{thomsen2024low}. 
Additionally, the 7-qubit Steane code, a small example of a color code with distance~3, has been recently demonstrated experimentally \cite{postler2022demonstration,ryananderson2022implementing,bluvstein2024logical,postler2024demonstration}, exhibiting not only memory experiments but also the implementation of nontrivial gates.
Furthermore, a recent experiment \cite{lacroix2024scaling} successfully demonstrated that the distance-5 color code provides better logical error suppression than the distance-3 color code.

With color codes, logical Clifford gates can be fault-tolerantly implemented efficiently using transversal gates \cite{bombin2006topological,kubica2015universal} or lattice surgery \cite{landahl2014quantum,kesselring2022anyon,thomsen2024low}. 
However, the fault-tolerant implementation of logical non-Clifford gates is rather challenging. 
One well-studied method to achieve this is by employing magic state distillation (MSD) \cite{bravyi2005universal,reichardt2005quantum,meier2012magic,fowler2012surface,bravyi2012magic,ogorman2017quantum,beverland2021cost}, a scheme that distills high-quality logical magic states (specific non-stabilizer states) from multiple faulty magic states. 
The distilled magic states are then consumed to implement desired non-Clifford gates on logical qubits. 
For instance, the state $\ketlgc{A} \coloneqq \ketlgc{0} + e^{i\pi/4}\ketlgc{1}$, where $\ketlgc{0}$ and $\ketlgc{1}$ are the logical basis states, can be used to execute the logical T~gate or, more generally, any $\pi/8$-rotation gate $\rot{\lgc{P}}{\pi/8} \coloneqq e^{-i(\pi/8)\lgc{P}}$, where $\lgc{P}$ is a logical Pauli operator on multiple logical qubits.
(Throughout the paper, we omit normalization factors of quantum states unless necessary and use overlines to denote logical states or operators.)

MSD is generally resource-intensive due to its demand for a significant number of logical Clifford gates between multiple logical qubits. 
Therefore, when applying it to a specific quantum error-correcting code, the procedure should be tailored and optimized for that code at both the logical and physical levels, in order to minimize its resource cost.
Such adaptations have been well studied for surface codes, notably by Litinski \cite{litinski2019magic}, reducing its cost by orders of magnitude compared to previous proposals.
The scheme in Ref.~\cite{litinski2019magic} is based on two key ideas: 
(i) Distillation circuits can be designed in a way that many faulty $\pi/8$-rotations are executed on a small number of logical qubits (e.g., 15 $\pi/8$-rotations on five logical qubits in the 15-to-1 MSD protocol), and (ii) ancillary logical qubits for the scheme can have smaller code distances than the code distance $\dout$ of the logical qubit outputting the distilled magic state. 
More specifically, the code distance $\dz$ with respect to $\lgz$ errors on ancillary logical qubits can be chosen to be smaller than $\dout$, since such $\lgz$ errors are detectable by the final $\lgx$ measurements of the circuit. 

In this work, we propose two resource-efficient MSD schemes for 2D color codes based on the 15-to-1 MSD protocol, by adapting and developing Litinski's ideas applied to surface codes.
Our first scheme employs lattice surgery and faulty T-measurement as its main components. 
The distillation circuit consists of 15 $\pi/8$-rotations on five logical qubits (initialized to $\ketlgc{+}^{\otimes 5}$), followed by $\lgx$ measurements on four of them (referred to as \textit{validation qubits}). 
Provided that all the $\lgx$ measurements yield $+1$, the distilled magic state is obtained from the unmeasured logical qubit. 
To execute each $\pi/8$-rotation, we use an auxiliary logical qubit, which undergoes lattice surgery with other logical qubits, followed by a non-fault-tolerant measurement in the basis of ${\ketlgc{0} \pm e^{-i\pi/4}\ketlgc{1}}$, referred to as a faulty T-measurement. 
We carefully design a layout for the scheme from both macroscopic and microscopic perspectives and verify that it is fault-tolerant, meaning it preserves code distances during lattice surgery.

We optimize the resource cost of the scheme using various approaches. 
Notably, we show that, not only are $\lgz$ errors on validation qubits tolerable (c.f.~Ref.~\cite{litinski2019magic}), but any single-location $\lgx$ error on a validation qubit is also tolerable, although it may incur three or more rotation errors. 
Consequently, both of the distances $\dx$ and $\dz$ (for $\lgx$ and $\lgz$ errors, respectively) can be set lower than $\dout$.
Additionally, we leverage the rich structure of color codes for optimization. 
For instance, color codes allow the measurement of an arbitrary pair of commuting Pauli operators in parallel \cite{thomsen2024low}, making it sufficient to use a single ancillary region surrounded by all the logical patches, including two auxiliary patches. 
Furthermore, domain walls for lattice surgery can be shortened by using rectangular patches instead of standard triangular patches. 

\label{para:ourfirstscheme}
Our first scheme using faulty T-measurement has an inherent limitation, as its output logical error rate cannot be lower than approximately $35p_\mr{FT}^3$, where $p_\mr{FT} = O(p)$ denotes the error rate of the faulty T-measurement and $p$ is the physical error rate.
To address this issue, we develop our second MSD scheme that employs a two-level approach. 
In this scheme, we integrate our first scheme with recent distillation-free magic state preparation protocols, such as those described in Refs.~\cite{chamberland2020very,itogawa2024even,gidney2024magic}, referred to as \textit{magic state cultivation} \cite{gidney2024magic}.
These protocols leverage the transversality of logical Clifford gates in color codes for directly preparing magic states fault-tolerantly.
They are highly resource-efficient since they avoid logical operations on multiple logical qubits.
However, they have a fundamental scalability issue due to their heavy reliance on post-selection, which imposes a practical lower bound on achievable infidelity.
To reach lower infidelities, our approach is to use magic states from cultivation as inputs for MSD, with appropriate optimization. 
The growing operation on cultivated magic states can be a bottleneck that significantly hinders fault tolerance, but we demonstrate that this issue can be mitigated by incorporating post-selection during decoding with the concatenated minimum-weight perfect matching (MWPM) decoder \cite{lee2025color} (the circuit-level matching-based decoder for color codes with the best known sub-threshold scaling).


We assess the performance of our schemes based on their output infidelities and resource costs, assuming that errors are corrected using the recently proposed concatenated MWPM decoder.
Compared to previous MSD schemes for color codes, such as the one in Ref.~\cite{beverland2021cost}, our schemes achieve a spacetime cost that is approximately two orders of magnitude lower for a given target infidelity.
For instance, at $p = 10^{-3}$, a magic state with an infidelity of $10^{-9}$ can be prepared using $\sim 15,000$ qubits over $\sim 1400$ time steps, with a failure rate of $\sim 0.2\%$, resulting in an effective spacetime cost of $\sim 2.1 \times 10^7$ (in contrast to $\sim 6.6 \times 10^9$ in Ref.~\cite{beverland2021cost}).
Compared to the latest cultivation scheme \cite{gidney2024magic}, although our schemes do not surpass its high resource efficiency, they can achieve significantly lower infidelities (e.g., $\sim 2 \times 10^{-16}$ at $p = 10^{-3}$) than those attainable with cultivation ($\gtrapprox 10^{-9}$).

This paper is structured as follows: 
In Sec.~\ref{sec:color_code}, we present preliminaries on 2D color codes, including their definitions, error-detecting properties, anyon model (with descriptions of domain walls), and methods for executing logical operations. 
In Sec.~\ref{sec:msd_scheme}, we describe our MSD scheme using faulty T-measurement and lattice surgery, along with our methodology for optimizing its resource cost. 
In Sec.~\ref{sec:higher_quality_magic_states}, we modify the scheme to achieve lower infidelities by integrating cultivation.
In Sec.~\ref{sec:performance_analysis}, we analyze the performance of our schemes based on their output infidelities and resource costs, and compare them with other approaches.
We conclude with final remarks in Sec.~\ref{sec:remarks}.

\section{Two-dimensional color codes \label{sec:color_code}}

In this section, we briefly review the basics of the 2D color codes including definitions, stabilizer structures, and anyon model.

\subsection{Definition and stabilizer structure \label{subsec:color_code_definition}}

\begin{figure}[!t]
	\centering
	\includegraphics[width=\linewidth]{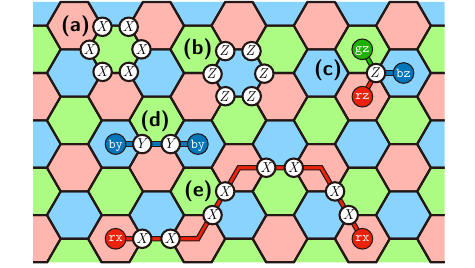}
	\caption{
        Hexagonal color code lattice with $X$- and $Z$-type checks on each face, exemplified in \subfig{a} and \subfig{b}, respectively.
        \subfig{c} A single-qubit Pauli error flips surrounding three checks, creating a triplet of bosons with different colors.
        \subfig{d} A pair of Pauli errors with the same type on an edge flips two checks on the faces that it connects, creating a pair of bosons with the same type.
        \subfig{d} A string operator of a specific type (`red Pauli-$X$' in the example) flips a pair of checks located at its ends, creating a pair of bosons with that type (\rx).
    }
	\label{fig:color_codes}
\end{figure}

A 2D color code lattice \cite{bombin2006topological} is defined as a trivalent and 3-colorable lattice; that is, each vertex of the lattice is connected with three edges such that one of the three colors (red, green, or blue) can be assigned to each face in a way that adjacent faces do not have the same color.
Each edge of the lattice can be given the same color as that of the two faces that are connected by the edge.
The hexagonal (6-6-6) lattice in Fig.~\ref{fig:color_codes} is one of its representative examples and we will consider it throughout this paper.

For a given 2D color code lattice, a 2D color code is defined with a qubit placed on each vertex of the lattice. 
It has two types of checks (or stabilizer generators) $\sgx{f}$ and $\sgz{f}$ for each face $f$, defined as
\begin{align*}
    \sgx{f} \coloneqq \prod_{v \in f} \Xv{v}, \qquad \sgz{f} \coloneqq \prod_{v \in f} \Zv{v},
\end{align*}
where $\Xv{v}$ and $\Zv{v}$ are the Pauli-$X$ and $Z$ operators, respectively, on the qubit placed at $v$, and `$v \in f$' denotes that the vertex $v$ is included in $f$. 
In other words, the code space is the common $+1$ eigenspace of these operators. 
We refer to $\sgx{f}$, $\sgz{f}$, and $\sgx{f}\sgz{f}$ as the \textit{$X$-, $Z$-, and $Y$-type checks} on face $f$, respectively. 
Examples of $X$- and $Z$-type checks are shown in Figs.~\ref{fig:color_codes}(a) and~(b), respectively. 

\subsection{Detection of Pauli errors \label{subsec:error_detection}}

Pauli errors flip the eigenvalues of checks.
For instance, a Pauli error on a qubit flips the checks anticommuting with the error on the surrounding three faces (with different colors), as exemplified in Fig.~\ref{fig:color_codes}(c).
If two qubits connected by an edge simultaneously undergo Pauli errors of the same type, the corresponding checks on the two faces (with the same color) connected by the edge are flipped, as shown in Fig.~\ref{fig:color_codes}(d).
Errors on multiple consecutive edges of the same color can form a string operator, which flips the two checks located at its two endpoints.
For example, as shown in Fig.~\ref{fig:color_codes}(e), a red Pauli-$X$ string operator flips two $Z$-type checks on the red faces at which the string operator terminates.
A general string-net operator is composed of a combination of string operators (with different types) connected by single-qubit errors, which may flip multiple checks located at its ends.
Note that Pauli-$Y$ string operators can be interpreted as pairs of overlapped Pauli-$X$ and $Z$ string operators; thus, a string-net operator may also contain string operators of different Pauli types as well as different colors.

Flips of checks are commonly described using the language of anyon theory. 
The color code phase contains nine nontrivial types (termed \textit{charge labels}) of bosons: \rx, \ry, \rz, \gx, \gy, \gz, \bx, \by, and \bz. 
Each boson has a color label (\rbs, \gbs, or \bbs) and a Pauli label (\xbs, \ybs, or \zbs). 
These nine bosons generate the entire group of anyons. 
A set of flipped checks on a face is associated with a boson that has a color label corresponding to the color of the face and a Pauli label corresponding to the operator needed to flip the checks. 
For example, if only the $X$-type check on a red face is flipped, an \rz boson is created. 
If both the $X$- and $Z$-type checks on a green face are flipped, a \gy boson is created.
Figure~\ref{fig:color_codes} visualizes some examples of bosons created by Pauli errors.
For instance, a red Pauli-$X$ string operator creates two \rx bosons at its ends, or equivalently, it transfers an \rx boson from one end to the other. 
Throughout the paper, we refer to a red Pauli-$X$ string operator simply as an \rx-string operator, and similarly for the other types of string operators. 
The interactions between bosons via string(-net) operators can be expressed in terms of the fusion rules:
\begin{align*}
&\rx \times \rx = \ry \times \ry = \cdots = \idn~~\text{(holds for all bosons)}, \\
&\rx \times \gx \times \bx = \ry \times \gy \times \by = \rz \times \gz \times \bz = \idn, \\
&\rx \times \ry \times \rz = \gx \times \gy \times \gz = \bx \times \by \times \bz = \idn.
\end{align*}
By using these properties, errors can be predicted from check measurement outcomes.

The above discussion is limited to the case when there are no errors during syndrome extraction.
If such errors are considered, we need to measure checks repeatedly and the product of two consecutive check measurement outcomes (called \textit{detectors}) replace the roles of checks.
That is, an $X$ or $Z$ error on a data qubit flips three spatially adjacent detectors and a measurement error of a check flips two temporarily adjacent detectors.
The anyon theory can be natively extended to this spacetime picture; for example, consecutive measurement errors on a red $Z$-type check moves an \rx boson in the temporal direction.

See Ref.~\cite{kesselring2022anyon} for more details on the anyon theory of color codes.

\subsection{Logical qubits encoded in patches \label{subsec:logical_qubit_definitions}}

\begin{figure}[!t]
    \centering
    \includegraphics[width=\columnwidth]{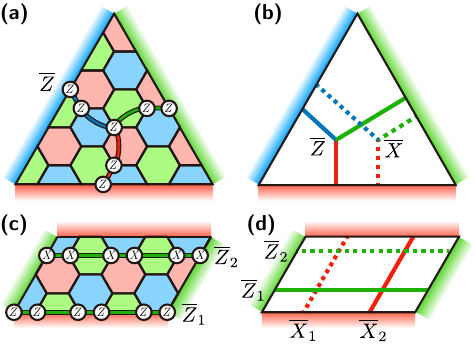}
    \caption{
        \subfig{a} Triangular logical patch with distance 7, which encodes a single logical qubit.
        Its logical Pauli-$Z$ ($X$) operator $\lgz$ ($\lgx$) is the Pauli-$Z$ ($X$) string-net operator terminating at the three boundaries.
        \subfig{b} Schematic diagram of a general triangular color code. Solid (dotted) lines represent Pauli-$Z$ ($X$) string-net operators.
        \subfig{c} Rectangular logical patch with code distances of $d_X=4$ and $d_Z=6$ (respectively for logical Pauli-$X$ and $Z$ errors), which encodes two logical qubits.
        Its two logical Pauli-$Z$ ($X$) operators $\lgz_1$, $\lgz_2$ ($\lgx_1$, $\lgx_2$) are respectively \gz- and \gx-string (\rx- and \rz-string) operators terminating at the two opposite green (red) boundaries.
        \subfig{d} Schematic diagram of a general rectangular color code with red and green boundaries.
        Similarly, solid (dotted) lines represent Pauli-$Z$ ($X$) string operators.
    } 
    \label{fig:color_code_patches}
\end{figure}

Logical qubits can be encoded in a color code in various ways \cite{bombin2006topological,fowler2011twodimensional,brown2017poking} and one representative method is to use a `patch' surrounded by boundaries \cite{bombin2006topological}.
A boundary of a specific color (red, green, or blue) is created by cutting the lattice, ensuring that the boundary is incident to only edges of that color (namely, it touches only faces of the other two colors).
By definition, a string operator of a specific color can terminate at a boundary of that color, or equivalently, in terms of the anyon theory, bosons of that color label can be absorbed by the boundary (this is also known as `condensation' \cite{kesselring2022anyon}).
We consider two of the simplest types of logical patches, triangular and rectangular patches, and refer to the corresponding codes as triangular and rectangular color codes, respectively.

A triangular logical patch is surrounded by three boundaries of different colors, as exemplified in Fig.~\ref{fig:color_code_patches}(a).
It encodes a single logical qubit, whose logical Pauli-$X$ and $Z$ operators $\lgx$, $\lgz$ are respectively Pauli-$X$ and $Z$ string-net operators terminating at the three boundaries.
These logical Pauli operators are also schematically drawn in Fig.~\ref{fig:color_code_patches}(b) for a general triangular logical patch.
Here, solid and dotted lines represent Pauli-$Z$ and $X$ string-net operators, respectively, and we will use these notations throughout this work.
Note that the string-net operators can be moved to one of the three boundaries, meaning that we can define $\lgx$ and $\lgz$ to be supported on qubits placed along the boundary.
The code distance $d$ of the code (i.e., minimum weight of $\lgx$ and $\lgz$) is odd since $\lgx$ and $\lgz$ may have the same support and should anticommute with each other.

A rectangular logical patch is surrounded by two pairs of parallel single-colored boundaries, where the two pairs have different colors, as exemplified in Fig.~\ref{fig:color_code_patches}(c).
It encodes two logical qubits with logical Pauli operators $\lgx_1$, $\lgz_1$ and $\lgx_2$, $\lgz_2$.
$\lgz_1$ and $\lgz_2$ are respectively defined as Pauli-$Z$ and $X$ string operators terminating at a specific pair of opposite boundaries (green in the example), while $\lgx_1$ and $\lgx_2$ are Pauli-$X$ and $Z$ string operators terminating at the other pair of boundaries (red in the example), which are drawn schematically in Fig.~\ref{fig:color_code_patches}(d).
The two logical qubits may have different code distances $d_X$ and $d_Z$ (which are even) for $\lgx$ and $\lgz$ errors, determined by the lengths of the boundaries.
We emphasize that $\lgx_2$ and $\lgz_2$ are defined with opposite labels from their physical Pauli types (i.e., $Z$ for $\lgx_2$ and $X$ for $\lgz_2$), which is to ensure that the two logical qubits have the same pair of code distances $(d_X, d_Z)$ for logical $X$ and $Z$ errors.
If this is not desirable (due to biased physical noise, for example), one can define the patch using Pauli boundaries \cite{kesselring2018boundaries,thomsen2024low} instead of color boundaries.
We will not consider this in our schemes, but modifying them to use rectangular patches with Pauli boundaries would not be difficult and would not significantly change their resource costs or performance.

While triangular patches are sufficient in many cases, rectangular patches can be particularly useful when logical noise is biased (i.e., a specific type of logical Pauli error is dominant or more harmful), such as in MSD (see Sec.~\ref{subsec:distillation_circuit}) or magic state injection \cite{fazio2024logical}, where independent adjustment of $d_X$ and $d_Z$ is beneficial. 
Moreover, using rectangular patches may help reduce resource costs. 
When performing lattice surgery, the region between a logical patch and the ancillary region for lattice surgery can be reduced by half by using a rectangular patch instead of a triangular patch, though this comes at the cost of sacrificing the ability to perform any Pauli measurements. 
See Sec.~\ref{subsec:lattice_surgery} for more details.



\subsection{Logical operations \label{subsec:logical_operations}}

We first consider initializing or measuring a logical qubit fault-tolerantly in a Pauli basis.
To initialize a triangular color code qubit to $\ketlgc{0}$, we just need to initialize every physical qubit to $\ket{0}$ and measure the checks.
Assuming no errors, $Z$-type checks deterministically give $+1$, while the values of $X$-type checks are randomly determined, which does not matter since products of two consecutive check outcomes are used for decoding.
To measure the logical qubit in the $\lgz$ basis, we need to measure every physical qubit in the $Z$ basis, which gives the measurement outcome as the product of the outcomes on qubits along one of the three boundaries.
For each face, the product of the $Z$-measurement outcomes of qubits belonging can be used to infer the value of the final $Z$-type check outcome. In the case that measurements are noisy, the single qubit measurements form a detector together with the previous $Z$-type check outcome.
Note that the above processes correspond to placing $Z$-type temporal boundaries, which condense \rz, \gz, and \bz bosons \cite{kesselring2022anyon}.
Initialization and measurement to other Pauli bases can be done analogously.

For the rectangular color code qubit in Fig.~\ref{fig:color_code_patches}(c), we can measure $\lgz_1$ ($\lgz_2$) by measuring $\lgc{Z} \otimes \lgc{Z}$ ($\lgc{X} \otimes \lgc{X}$) on every green edge.
Operators $\lgz_1$ and $\lgz_2$ can be measured at the same time by measuring every green edge in the Bell basis, which corresponds to placing a green temporal boundary condensing \gx, \gy, and \gz bosons.
Note that we can obtain the values of red and blue checks that commute with these Bell bases.
$\lgx_1$ and $\lgx_2$ can be measured similarly by measuring red edges.

Next, we consider logical Clifford operations.
For triangular color codes, the logical Hadamard, phase, and \cnot gates (which generate the Clifford group) can be implemented transversally.
For example, the logical Hadamard gate $\lgc{H}$ can be done just by applying physical Hadamard gates on all the physical qubits in the patch.
The logical phase gate $\lgc{S}$ is similar but a little more tricky since some qubits undergo $S$ while the other qubits undergo $S^\dagger$.
For the logical \cnot gate, we just need to apply a physical \cnot gate on every pair of corresponding physical qubits in the two logical patches, although the two patches need to be stacked in three-dimensional space if only local interactions are available.
See Ref.~\cite{kubica2015universal} for more details.

The above methods (except for the \cnot gate) may not generally work for other types of logical patches including rectangular ones.
Alternatively, we can use non-destructive multi-qubit Pauli measurements instead of Clifford gates as ingredients for universal quantum computing.
In other words, given a quantum circuit composed of Clifford and T gates, we can commute all the Clifford gates to the end of the circuit (while transforming T gates appropriately) and merge them with final measurements, which yields a series of Pauli measurements.
In color code logical patches, Pauli measurements can be performed by using lattice surgery \cite{landahl2011faulttolerant,kesselring2022anyon,thomsen2024low}, which is a process to merge the boundaries of multiple checks with additional check operators.
We will elaborate on detailed schemes for this in Sec.~\ref{subsec:lattice_surgery}.

\begin{figure}[!t]
	\centering
	\includegraphics[width=\columnwidth]{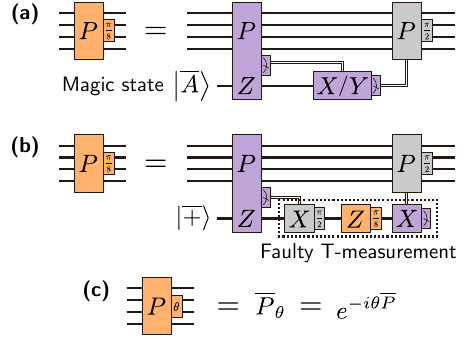}
	\caption{
        Two circuits for implementing the $\pi/8$-rotation gate $\rot{\lgc{P}}{\pi/8} \coloneqq e^{-i(\pi/8)\lgc{P}}$ for a logical multi-qubit Pauli operator $\lgc{P}$, respectively employing \subfig{a} a logical magic state $\ketlgc{A} \coloneqq \ketlgc{0} + e^{i\pi/4}\ketlgc{1}$ and \subfig{b} faulty T-measurement, which are expressed using Pauli measurements (purple boxes) and Pauli rotation gates (orange/gray boxes) defined in \subfig{c}.
        In \subfig{a}, $\ketlgc{A}$ is first prepared on an auxiliary logical qubit.
        $\lgc{P} \otimes \lgz$ is then measured jointly on the input and auxiliary qubits, which determines the basis ($X$ if $+1$ and $Y$ if $-1$) for the following measurement of the auxiliary qubit.
        If this measurement returns $-1$, a Pauli correction of $\lgc{P}$ should be applied.
        In \subfig{b}, $\ketlgc{+} \coloneqq \ketlgc{0} + \ketlgc{1}$ is prepared on an auxiliary logical qubit, followed by a $\lgc{P} \otimes \lgz$ measurement on the input and auxiliary qubits, which returns $\lambda = \pm 1$.
        The auxiliary logical qubit is then measured in the basis of $\qty{\ketlgc{0} \pm e^{-i\lambda\pi/4}\ketlgc{1}}$ via a faulty T-measurement.
        If this measurement returns $-1$, a Pauli correction of $\lgc{P}$ should be applied as well.
    }
	\label{fig:rotation_gate_circuit}
\end{figure}

We lastly need non-Clifford gates.
In particular, the $\pi/8$-rotation gate $\rot{\lgc{P}}{\pi/8}$ should be executable for any logical Pauli operator $\lgc{P}$.
(Throughout this work, we denote $\rot{P}{\theta} \coloneqq e^{-i\theta P}$ for a Pauli operator $P$ and a real number $\theta$.)
Two circuits for this operation \cite{litinski2019game,litinski2019magic} are presented in Figs.~\ref{fig:rotation_gate_circuit}(a) and~(b), which consist of Pauli measurements (represented by purple boxes) and Pauli rotation gates (represented by orange/gray boxes) defined in Fig.~\ref{fig:rotation_gate_circuit}(c).
Here, $\pi/2$-rotation gates, which are simply Pauli gates with a global phase change, are shown as gray boxes to emphasize that they do not need to be applied explicitly but only change the Pauli frame \cite{knill2005quantum}.
The circuit in Fig.~\ref{fig:rotation_gate_circuit}(a) consumes a magic state
\begin{align}
    \ketlgc{A} \coloneqq \ketlgc{0} + e^{i\pi/4} \ketlgc{1} \label{eq:magic_state}
\end{align}
with a joint Pauli measurement $\lgc{P} \otimes \lgz$ (returning $\lambda$), followed by a measurement in the basis of $\lgx$ (for $\lambda = 1$) or $\lgy$ (for $\lambda = -1$).
If the second measurement returns $-1$, a Pauli correction $\lgc{P}$ should be applied to the remaining qubits.
By reversing the time order of this circuit, we obtain the second circuit in Fig.~\ref{fig:rotation_gate_circuit}(b).
Here, the auxiliary qubit is prepared to $\ketlgc{+}$, undergoes a $\lgc{P} \otimes \lgz$ measurement jointly with the input qubits (which returns $\lambda$), and then measured in the basis of non-stabilizer states $\qty{\ketlgc{0} \pm e^{-i\lambda\pi/4}\ketlgc{1}}$, making a Pauli correction on the remaining qubits.
Note that the circuit in Fig.~\ref{fig:rotation_gate_circuit}(a) can be rewritten in a way that the measurement basis for the auxiliary qubit is fixed to $\lgx$ and a Clifford correction is applied to the other qubits.
However, we avoid using this construction since a Clifford correction may flip the angles of the following rotation gates in the MSD circuit, which is not desirable.

These two circuits take an approach of concentrating the `non-Clifford part' into the initialization or measurement of the auxiliary logical qubit, which is thus the core of the circuits that determine their fault tolerance.
A logical magic state can be prepared from physical magic states by using state injection.
In a color code logical patch, this can be done by preparing the physical state on a specific physical qubit in the patch and measuring other qubits appropriately, followed by ordinary check measurements \cite{beverland2021cost}.
Faulty T-measurement is the reverse process: reducing a logical patch into a single physical qubit and then measuring it in the desired basis.
This is one of the key ingredients of our MSD scheme, thus we will describe this in detail in Sec.~\ref{subsec:faulty_T_measurement}.
Both state injection and faulty T-measurement are not fault-tolerant, namely, a single physical error during the protocols may cause a logical failure.
However, MSD protocols can be employed to extract high-quality magic states from multiple trials of non-fault-tolerant $\pi/8$-rotations \cite{bravyi2005universal}, which can be inputted in the circuit of Fig.~\ref{fig:rotation_gate_circuit}(a) for executing fault-tolerant $\pi/8$-rotations.
Alternatively, transversal logical Clifford gates on color codes can be utilized to construct distillation-free magic state preparation protocols \cite{goto2016minimizing, chamberland2019faulttolerant, chamberland2019faulttolerant,chamberland2020very, itogawa2024even}.
We will discuss these methods in more details in Sec.~\ref{sec:higher_quality_magic_states}

\subsection{Domain walls \label{subsec:domain_walls}}

We now briefly review domain walls, which are key ingredients for lattice surgery of color codes.
Domain walls are one-dimensional (1D) subregions along which two topological phases interface.
In simple terms, a domain wall of a color code is a thin region in which checks are deformed from ordinary color-code checks under a specific rule so that bosons approaching the domain wall are affected in a particular way.
Such effects are determined by the charge labels of the bosons ($\rx, \ry, \cdots, \bz$) and the directions in which they approach.
Bosons located in one of the two regions separated by the domain wall can be classified as follows:
\begin{enumerate}
    \item \textbf{Condensed}: The boson is condensed on the domain wall; that is, the domain wall acts as a boundary of the corresponding type for the boson (such as a red boundary for an \rx boson) so that the boson can be absorbed at the domain wall.
    \item \textbf{Deconfined}: The boson can freely move across the domain wall to the other side, but its charge label can be changed.
    \item \textbf{Confined}: The boson is confined to the region; that is, it cannot pass through or be condensed on the domain wall.
\end{enumerate}
Note that bosons can be classified differently for either side of the domain wall.
Based on these features of domain walls, they can be categorized into three types: \textit{opaque}, \textit{transparent}, and \textit{semi-transparent} domain walls \cite{kesselring2022anyon}.

\subsubsection{Opaque domain wall}
An opaque domain wall is a trivial one that does not have any deconfined bosons.
In other words, it is just an empty region between two color codes with boundaries.
Each side of the domain wall is a color or Pauli boundary, which condenses bosons that have the corresponding color or Pauli label (e.g., if it is a red boundary, it condenses \rx, \ry, and \rz bosons).

\subsubsection{Transparent domain wall}
A transparent (or invertible) domain wall \cite{kitaev2012models} allows every boson to be deconfined.
The charge labels of bosons are permuted according to certain rules when they move across the domain wall.
The rule is always symmetric; namely, if a boson \ttt{a} is transformed into $\ttt{a}'$ in one direction, $\ttt{a}'$ is transforms into $\ttt{a}$ in the opposite direction.
Hence, the charge-changing rule can be characterized by a permutation of nine charge labels of bosons.
Note that only $72$ permutations are allowed considering the fusion and braiding rules of bosons \cite{kesselring2018boundaries}.
Each of them can be expressed as a combination of a color permutation (between \rbs, \gbs, and \bbs), a Pauli permutation (between \xbs, \ybs, and \zbs), and an exchange between the color and Pauli labels ($\rbs \leftrightarrow \xbs$, $\gbs \leftrightarrow \ybs$, $\bbs \leftrightarrow \zbs$), which leads to total $6 \times 6 \times 2 = 72$ permutations.


\subsubsection{Semi-transparent domain wall \label{subsubsec:STDW}}
A semi-transparent domain wall allows only some bosons to be deconfined.
Each side of the domain wall condenses bosons with one particular charge label, which can be different for the two sides.
Bosons that have the same color (Pauli) label as the condensed boson are deconfined in the region and referred to as \textit{electric (magnetic)} charges for that side of the domain wall.
Importantly, the charge-changing rule of deconfined bosons applied when they move across the domain wall is either `\ttt{em}-preserving' or `\ttt{em}-exchanging'; namely, if it is \ttt{em}-preserving (\ttt{em}-exchanging), an electric charge of one side is transformed into an electric (magnetic) charge of the other side, and vice versa for a magnetic charge.
Note that it is not a one-to-one correspondence and the transformed charge can be any of two types of electric (magnetic) charges.

For example, let us suppose that a semi-transparent domain wall condenses \rx on side A and \gz on side B and its charge-changing rule is \ttt{em}-exchanging.
Then \ry, \rz, \gx, and \bx (\gx, \gy, \rz, and \bz) bosons are deconfined in side A (B), where the first two are electric charges and the latter two are magnetic charges.
Since it is \ttt{em}-exchanging, an electric charge (\ry or \rz) moving across the domain wall from side A is transformed into any of magnetic charges (\rz and \bz) in the other side, and vice versa for a magnetic charge in side A.
The other bosons \gy, \gz, \by, and \bz (\rx, \ry, \bx, and \by) are confined to the region that side A (B) belongs to.

\begin{figure*}[!t]
	\centering
	\includegraphics[width=\textwidth]{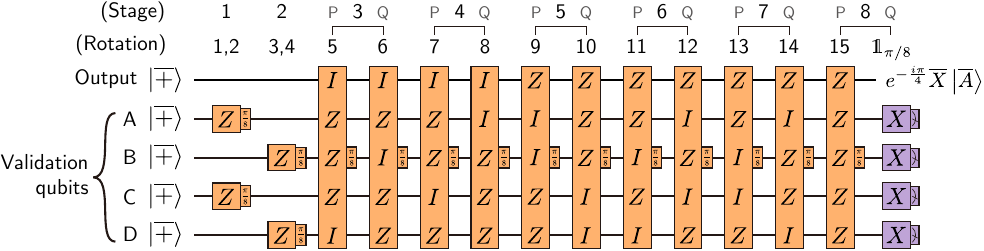}
	\caption{
        15-to-1 magic state distillation circuit.
        Five logical qubits including one output qubit and four validation qubits (A--D) are initialized to $\ketlgc{+} \coloneqq \ketlgc{0} + \ketlgc{1}$ and undergo (non-fault-tolerant) 15 $\pi/8$-rotations via the circuit in Fig.~\ref{fig:rotation_gate_circuit}(b).
        All the validation qubits are then measured in the $\lgx$ basis.
        If all of these measurements give the outcome of $+1$, the distillation succeeds, and the output qubit has a distilled magic state, which should be $e^{-i\pi/4} \lgx \ketlgc{A} = \ketlgc{0} + e^{-i\pi/4} \ketlgc{1}$ when there are no errors.
        In our default scheme, the 15 rotations, labeled by 1--15, are paired up as $\qty{\qty(\Pstage{k}, \Qstage{k})}_{k=1}^{8}$ in order (including a placeholder rotation $\rot{\identity}{\pi/8}$) and executed respectively in stages 1--8.
 }
	\label{fig:msd_circuit}
\end{figure*}

The effects of a semi-transparent domain wall on bosons can be intuitively described by using a boson table \cite{kesselring2018boundaries}
\begin{align*}
    \bosontable,
\end{align*}
where the nine charge labels of bosons are aligned in a $3 \times 3$ table.
On this table, we mark condensed bosons with a bullet ($\condensed$), confined bosons with a cross ($\confined$), electric charges with a square ($\electric$), and magnetic charges with a triangle ($\magnetic$).
With these notations, the domain wall in the above example can be expressed as a pair of boson tables:
\begin{align*}
    \bosontableST[\mr{A}\quad\mr{B}]{r}{x}{g}{z}{1}{2},
\end{align*}
where the colors (orange or sky blue) of the squares and triangles indicate the charge-changing rule of the domain wall; namely, bosons with the same symbol color on both sides can be mapped.

\section{Single-level 15-to-1 magic state distillation scheme \label{sec:msd_scheme}}

In this section, we introduce our 15-to-1 MSD scheme for color codes, which use lattice surgery and faulty T-measurement as its basic ingredients.
We first present the distillation circuit based on the 15-qubit Reed-Muller code \cite{bravyi2005universal,fowler2012surface}, which is modified to contain 15 $\pi/8$-rotations on five logical qubits \cite{litinski2019game,litinski2019magic}.
We then describe its color code implementation from both macroscopic and microscopic perspectives.

We consider the circuit-level noise model with strength $p$, which is defined as follows:
\begin{itemize}
    \item Every measurement outcome is flipped with probability $p$.
    \item Every preparation of a qubit produces an orthogonal state with probability $p$.
    \item Every single- or two-qubit unitary gate (including the idle gate $I$) is followed by a single- or two-qubit depolarizing noise channel of strength $p$. We here regard that, for every time step of the circuit, idle gates $I$ are acted on all the qubits that are not involved in any non-trivial unitary gates or measurements.
\end{itemize}
Here, the single- and two-qubit depolarizing channels of strength $p$ are respectively defined as
\begin{align*}
    &\depchannelone{p}:\; \rho^{(1)} \mapsto (1 - p) \rho^{(1)} + \frac{p}{3} \sum_{P \in \qty{X, Y, Z}} P \rho^{(1)} P, \\ 
    &\depchanneltwo{p}:\; \rho^{(2)} \mapsto (1 - p) \rho^{(2)} + \frac{p}{15} \\ 
    &\qquad\quad \times \sum_{\substack{P_1, P_2 \in \qty{I, X, Y, Z} \\ P_1 \otimes P_2 \neq I \otimes I}} \qty(P_1 \otimes P_2) \rho^{(2)} \qty(P_1 \otimes P_2),
\end{align*}
where $\rho^{(1)}$ and $\rho^{(2)}$ are arbitrary single- and two-qubit density matrices, respectively.

\subsection{Distillation circuit \label{subsec:distillation_circuit}}

We consider the 15-to-1 MSD circuit in Fig.~\ref{fig:msd_circuit} that distills one high-fidelity logical magic state $\ketlgc{0} + e^{-i\pi/4}\ketlgc{1} = e^{-i\pi/4}\lgx\ketlgc{A}$ from 15 faulty $\pi/8$-rotations.
Five logical qubits, which consist of one \textit{output qubit} and four \textit{validation qubits} (respectively A--D), are initialized to $\ketlgc{+}$ and undergo 15 faulty $\pi/8$-rotations.
Each of the rotations is performed via one of the two circuits in Figs.~\ref{fig:rotation_gate_circuit}(a) and~(b).
Although this choice does not make significant difference, we select using the circuit in Fig.~\ref{fig:rotation_gate_circuit}(b) that contains a faulty T-measurement, as the $\lgy$ measurement in Fig.~\ref{fig:rotation_gate_circuit}(a) makes decoding slightly difficult.
($\lgy$ measurement outcomes involve both decoding graphs respectively for $X$ and $Z$ errors, thus the two graphs need to be connected.)
Note that all the $\pi/8$-rotations contain only $\lgz$ in their rotation bases, thus they all commute with each other.
After that, the four validation qubits are measured in the $\lgx$ basis.
If all of these outcomes are $+1$, we conclude that the protocol succeeds and use the marginal state on the output qubit as a distilled magic state $e^{-i\pi/4}\lgx\ketlgc{A}$.
This is an input to the circuit of Fig.~\ref{fig:rotation_gate_circuit}(a) for implementing a $\pi/8$-rotation gate, but the $\ov{P} \otimes \lgz$ measurement outcome should be interpreted reversely due to the extra $\lgx$.
If the protocol fails, we discard the output state and retry the protocol.

The MSD circuit works since the combination of the 15 $\pi/8$-rotations is mathematically identical to the unitary gate $(\rot{\lgz}{-\pi/8})_\mr{out} \otimes \identity_\mr{ABCD}$ in the ideal case \cite{litinski2019magic}.
Not only that, various types of errors in the circuit can be detected by the final $\lgx$ measurements on the validation qubits.
Let us explore this fault tolerance property in more detail.
Errors in the circuit can be categorized into two groups: (i) errors from faulty T-measurements and (ii) memory errors of the five logical qubits.
Note that other types of errors are equivalent to errors in one of these groups.
For example, the $\lgc{P} \otimes \lgz$ measurement in the circuit of Fig.~\ref{fig:rotation_gate_circuit}(b) may give a flipped outcome, which is equivalent to a $\lgx$ error just before the faulty T-measurement.
To spoil the conclusion first, the MSD circuit can tolerate up to two rotation errors (from faulty T-measurements), any $\lgz$ errors on validation qubits, and up to one $\lgx$ error on a validation qubit.

For the first group of errors, if we model the faulty T-measurement as random Pauli noise followed by the perfect T-measurement, the corresponding $\lgx$, $\lgy$, and $\lgz$ errors are respectively converted to $\rot{\lgc{P}}{-\pi/4}$, $\rot{\lgc{P}}{\pi/4}$, and $\rot{\lgc{P}}{\pi/2} (=\lgc{P})$ errors after the circuit is executed \cite{litinski2019magic}.
Importantly, every combination of at most two $\rot{\lgc{P}}{\pi/2}$ errors can be detected by the final $\lgx$ measurements, while some combinations of three $\rot{\lgc{P}}{\pi/2}$ errors are not detectable, such as rotations~5, 7, and~14 in Fig.~\ref{fig:msd_circuit}.
$\rot{\lgc{P}}{\pm\pi/4}$ errors are even less detrimental since $\rot{\lgc{P}}{\pm\pi/4} = (\identity \mp i\lgc{P})/\sqrt{2}$.
Therefore, the output error rate $\infidMSD$ scales like $\infidMSD \sim p_\mr{T}^3$, where $p_\mr{T}$ is the noise strength of that random Pauli noise.

Let us now consider the second group of errors.
A $\lgz$~error on any part of the MSD circuit can be commuted to end of the circuit without changing anything.
If it is on the output qubit, it incurs a logical error on the distilled magic state, whereas it is always detectable if it is on one of the validation qubits.
In other words, the MSD circuit is tolerant to $\lgz$ errors on the validation qubits.
For $\lgx$ errors, situations are more complicated because they may anticommute with some rotation bases.
An $\lgx$~error on a qubit $q$ after the $i$-th rotation can be commuted to the beginning of the circuit and absorbed into the initial $\ketlgc{+}$ state, which changes the angles of several rotations (that are placed before the $(i + 1)$-th rotation and nontrivially involves $q$) to $-\pi/8$.
This can be regarded as correlated $(-\pi/4)$-rotation errors after all the $\pi/8$-rotations are perfectly executed.
Since more than two rotation errors can be correlated, one might expect the errors to be detrimental.
However, this is not the case when $q$ is one of the validation qubits; namely, the errors cannot damage the output state without being detected.
See Appendix~\ref{app:MSD_X_error_tolerance_proof} for the proof.
(To sketch the proof, supposing that the errors make an output logical error, at least one subset $\widetilde{\mathcal{P}}$ of the set of the correlated rotation errors affects the output qubit an odd number of times and affects each validation qubit an even number of times.
Then the weight sum of the elements of $\widetilde{\mathcal{P}}$ is odd, implying that $\abs{\widetilde{\mathcal{P}}}$ is also odd.
However, since each element of $\widetilde{\mathcal{P}}$ involves $q$, $\widetilde{\mathcal{P}}$ affects $q$ an odd number of times, which is a contradiction.)
To summarize, the MSD circuit is tolerant to a single-location $\lgx$ error on one of the validation qubits, while two or more $\lgx$ errors may not be tolerable.
On the other hand, a $\lgx$ error on the output qubit is always detrimental.

To implement the MSD circuit with color codes, we will investigate the following questions throughout the next three subsections: how to perform the faulty T-measurement (Sec.~\ref{subsec:faulty_T_measurement}), how to arrange logical patches encoding the output qubit, four validation qubits, and auxiliary qubits (Sec.~\ref{subsec:layout}), and how to perform lattice surgery for measuring $\lgc{P} \otimes \lgz$ in the circuit (Sec.~\ref{subsec:lattice_surgery}).
We will then describe our scheme comprehensively in Sec.~\ref{subsec:scheme} and calculate its resource costs in Sec.~\ref{subsec:msd_resource_cost}.

\subsection{Faulty $T$-measurement \label{subsec:faulty_T_measurement}}

Faulty T-measurement is a process to measure a logical qubit non-fault-tolerantly in the basis of $\qty{\rot{\lgz}{-\lambda\pi/8} \ketlgc{\pm}} = \qty{\ketlgc{0} \pm e^{-i\lambda\pi/4} \ketlgc{1}}$ for $\lambda \in \qty{0, 1}$, as illustrated in Fig.~\ref{fig:rotation_gate_circuit}(b).
For a triangular logical patch, it is implemented by shrinking the patch into a single physical qubit and measuring it in the corresponding basis of the physical qubit.
In detail, we first measure $\qty{X \otimes X, Z \otimes Z}$ on every red edge of the patch, leaving one physical qubit $q_\mr{corner}$ (located at the corner where the blue and green boundaries meet) unmeasured, as exemplified in Fig.~\ref{fig:faulty_T_measurement} for $d=7$.
Let $m_{XX}^{(e)}, m_{ZZ}^{(e)} \in \qty{\pm 1}$ denote the measurement outcomes of $X \otimes X$ and $Z \otimes Z$, respectively, on a red edge $e$.
We also define
\begin{align}
\begin{split}
    \lambda_{X} &\coloneqq \prod_{e \in \qty{\text{green bdry}}} m_{XX}^{(e)}, \\
    \lambda_{Z} &\coloneqq \prod_{e \in \qty{\text{green bdry}}} m_{ZZ}^{(e)},
\end{split}
    \label{eq:faulty_T_measurement_outcomes}
\end{align}
where $\qty{\text{green bdry}}$ is the set of red edges located along the green boundary (e.g., $\qty{e_1, e_2, e_3}$ in Fig.~\ref{fig:faulty_T_measurement}).
We then measure $q_\mr{corner}$ in the basis of $\qty{\rot{Z}{-\lambda \lambda_Z \pi/8}\ket{\pm}}$.
Given that the measurement outcome corresponds to $\rot{Z}{-\lambda \lambda_Z \pi/8}\ket{\pm}$, the final outcome of the faulty T-measurement is $\pm \lambda_X$.
The scheme works since the qubits on the green boundary support $\lgx$ and $\lgz$, implying that the state of $q_\mr{corner}$ after measuring the red edges is
\begin{align*}
    X^{(1 - \lambda_Z)/2} Z^{(1 - \lambda_X)/2} \qty(\ketbra{0}{\lgc{0}} + \ketbra{1}{\lgc{1}})\ketlgc{\psi}
\end{align*}
for an initial logical state $\ketlgc{\psi}$.

\begin{figure}[!t]
	\centering
	\includegraphics[width=\columnwidth]{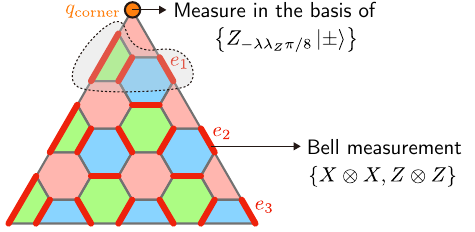}
	\caption{
        Implementation of a faulty T-measurement on a triangular logical patch for measuring the logical qubit in the basis of $\qty{\ketlgc{0} \pm e^{-i\lambda\pi/4} \ketlgc{1}}$ for $\lambda \in \qty{0, 1}$.
        Each red edge ($e$) is measured in the Bell basis, outputting two outcomes $m_{XX}^{(e)}$ and $m_{ZZ}^{(e)}$ for $X \otimes X$ and $Z \otimes Z$, respectively.
        The remaining single qubit $q_\mr{corner}$ (orange dot) is then measured in the basis of $\qty{Z_{-\lambda\lambda_Z \pi/8} \ket{\pm}}$, where $\lambda_Z \coloneqq m_{ZZ}^{(e_1)} m_{ZZ}^{(e_2)} m_{ZZ}^{(e_3)}$.
        The corresponding outcome of the faulty T-measurement is $\pm \lambda_X$, where $\lambda_X \coloneqq m_{XX}^{(e_1)} m_{XX}^{(e_2)} m_{XX}^{(e_3)}$.
        Red edges involved in the least-weight undetectable nontrivial combination of errors are highlighted as an area surrounded by a dashed line.
 }
	\label{fig:faulty_T_measurement}
\end{figure}

\label{para:aforementioned_faulty}
The aforementioned faulty T-measurement scheme has a logical error rate of $O(p)$, originated from physical-level errors on $q_\mr{corner}$, which are equivalent to logical errors on the patch before applying the faulty T-measurement.
We suppose that $q_\mr{corner}$ idles during two time steps for Bell measurements and then is measured in a single time step.
Then, the leading-order terms of $\lgx$, $\lgy$, and $\lgz$ error rates are $(2/3)p$, $(2/3)p$, and $(5/3)p$, respectively, under the circuit-level noise model.
(Here, the final measurement only contributes to the $\lgz$ error rate, as our noise model assumes a probabilistic flip of the measurement outcome. This assumption is stricter than modeling the measurement noise as a depolarizing channel, as $\lgz$ errors are more harmful than $\lgx$ and $\lgy$ errors.)
Note that Bell measurements may cause logical errors as well, but they can be suppressed by using $X$- and $Z$-type checks on green and blue faces (whose values are obtained from the Bell measurement outcomes on red edges).
They can be decoded straightforwardly by MWPM, as each error during a Bell measurement affects at most two checks for each Pauli type.
Since undetectable nontrivial combinations of such errors have weights of at least three as illustrated in Fig.~\ref{fig:faulty_T_measurement}, their contribution on the logical error rates of the faulty T-measurement is $O(p^2)$, which we will ignore in our analysis.

\subsection{Layout \label{subsec:layout}}

We now discuss how to design a layout for implementing the 15-to-1 MSD circuit in Fig.~\ref{fig:msd_circuit}.
We encode the output qubit in a triangular patch (denoted as \patchout) and the four validation qubits in two rectangular patches (denoted as \patchab and \patchcd) that respectively encode qubits~A,~B and qubits~C,~D.
We suppose that the $\lgz$ operators of qubits~A and~C consist of physical $Z$ operators, while those of qubits~B and~D consist of physical $X$ operators; see Fig.~\ref{fig:color_code_patches}(d).
In addition, we need auxiliary logical qubits for faulty T-measurements.
These qubits are involved in the $\lgc{P} \otimes \lgz$ measurement of the circuit in Fig.~\ref{fig:rotation_gate_circuit}(a) for each rotation $\rot{\lgc{P}}{\pi/8}$, which is executed via lattice surgery.
Importantly, a single lattice surgery process can measure two commuting Pauli operators at the same time \cite{thomsen2024low}, thus we employ two triangular patches \patchalpha and \patchbeta that respectively encode auxiliary qubits~\talpha and~\tbeta.
Lastly, lattice surgery requires an ancillary region with single-color boundaries, which is surrounded by the patches.
(See Sec.~\ref{subsec:lattice_surgery} for more details on lattice surgery.)
In summary, the layout consists of three triangular patches (\patchout, \patchalpha, \patchbeta), two rectangular patches (\patchab, \patchcd), and an ancillary region surrounded by the patches.

We use different code distances for the five patches.
Namely, \patchout has a code distance of $\dout$ and both \patchalpha and \patchbeta have a code distance of $\dm$.
Both \patchab and \patchcd have code distances of $\dx$ and $\dz$, which are the smallest weights of undetectable Pauli errors equivalent to $\lgx$ and $\lgz$, respectively.
In addition, the temporal code distance is set to be $\dm$, which is the number of syndrome extraction rounds required for each merging operation of lattice surgery.
We suppose $\dm, \dz < \dout$ throughout the discussion.

We have two reasons for using rectangular patches to encode the validation qubits.
First, $\lgx$ and $\lgz$ errors on the validation qubits have asymmetric effects:
$\lgz$ errors are always detectable, while $\lgx$ errors can be detrimental if two of them occur at the same time, as discussed in Sec.~\ref{subsec:distillation_circuit}.
Note that it is not essential to set $\dx \approx \dout$ unlike in Ref.~\cite{litinski2019magic} since a single $\lgx$ error on a validation qubit is tolerable.
Secondly, considering that all the $\pi/8$-rotations involve only $\lgz$ operators, using rectangular patches can contribute to lowering the resource cost.
Namely, the ancillary region only needs to be adjacent to the boundaries supporting $\lgz$ operators of the validation qubits, thus it is more efficient to use a rectangular patch where two logical qubits share the same boundary for their respective $\lgz$ operators.

\begin{figure*}[!t]
	\centering
	\includegraphics[width=\linewidth]{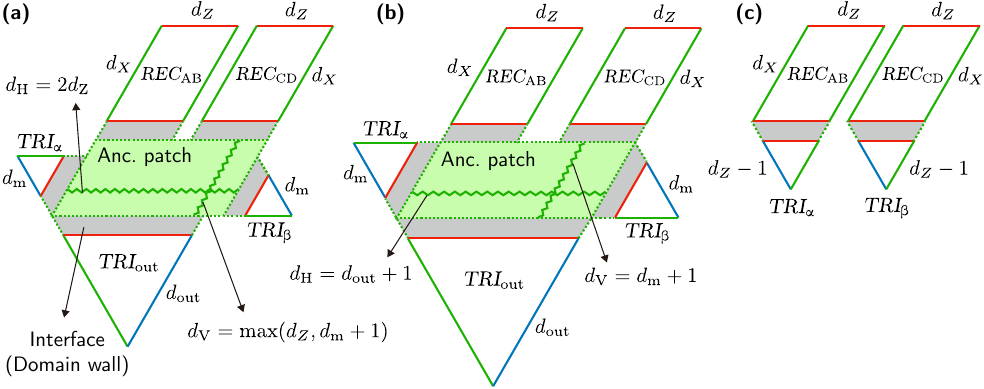}
	\caption{
        Macroscopic layouts of our 15-to-1 MSD scheme for \subfig{a} $2\dz > \dout$ and \subfig{b} $2\dz < \dout$, which requires $\dout - \dz \leq \dm < 2\dz$ to be fault-tolerant.
        Exceptionally, single-qubit rotations (1--4 in Fig.~\ref{fig:msd_circuit}) are executed using the layout in \subfig{c}.
        The layouts in~\subfig{a} and~\subfig{b} consist of five patches (\patchout, \patchab, \patchcd, \patchalpha, and \patchbeta), an ancillary patch with green boundaries, and interface regions (gray areas) that contain domain walls for lattice surgery.
        The ancillary patch and interface regions are collectively called the ancillary region.
        The colors of the solid lines indicate the colors of the corresponding boundaries.
        The horizontal and vertical dimensions of the ancillary patch (in terms of the weights of the shortest green string operators terminating at the respective pairs of opposite boundaries) are respectively $\dH \coloneqq \max(\dout+1, 2\dz)$ and $\dV \coloneqq \max(\dz, \dm+1)$ for both~\subfig{a} and~\subfig{b}.
 }
	\label{fig:msd_layout}
\end{figure*}

We denote the seven logical qubits as $\lgq{out}$, $\lgq{A}$, $\lgq{B}$, $\lgq{C}$, $\lgq{D}$, $\lgq{\talpha}$, and $\lgq{\tbeta}$, respectively, and the corresponding Pauli operators by using the same subscripts (e.g., $\lgzA$ for the $\lgz$ operator of $\lgq{A}$).
Additionally, we use the shorthand notation for a tensor product of these operators such as $\lgztp{out,A,D,\talpha} = \lgztp{OAD\talpha} \coloneqq \lgzout \otimes \lgzA \otimes \lgiB \otimes \lgiC \otimes \lgzD \otimes \lgzalpha \otimes \lgibeta$.

In Figs.~\ref{fig:msd_layout}(a) and~(b), we present the macroscopic pictures of possible layouts respectively for $2\dz > \dout$ and $2\dz < \dout$, where \patchout, \patchalpha, \patchab, \patchcd, and \patchbeta are placed clockwise in order and surrounds an ancillary region with green boundaries.
The $\lgx$ ($\lgz$) operators of \patchab and \patchcd terminate at their red (green) boundaries.
The ancillary region consists of the ancillary patch, which is a rectangular patch with only green boundaries (that does not encode logical qubits), and the interface regions between the ancillary patch and the logical patches, which contain domain walls for lattice surgery.
The ancillary patch is fixed throughout the scheme, whereas the interface regions vary depending on the operators we measure.
For example, interface regions corresponding to patches that are not involved in the measurement are turned off.
The horizontal and vertical dimensions of the ancillary patch are $\dH \coloneqq \max\qty(\dout+1, 2\dz)$ and $\dV \coloneqq \max\qty(\dm+1, \dz)$, respectively, in terms of the weights of the shortest green string operators terminating at the corresponding pairs of boundaries.
We require $\dout - \dz \leq \dm < 2\dz$ for the layout to be distance-preserving; see Condition~\ref{cond:dist_rotation_grouping} in Sec.~\ref{subsec:scheme} for more details.

Exceptionally, single-qubit $\pi/8$-rotations (rotations 1--4 in Fig.~\ref{fig:msd_circuit}) can be executed more efficiently by attaching auxiliary patches (with distances $\dz-1$) directly to the rectangular patches through a thin ancillary region, as shown in Fig.~\ref{fig:msd_layout}(c).
These auxiliary patches can be placed inside the ancillary region of the regular layout in Fig.~\ref{fig:msd_layout}(a) or~(b), thus the space cost does not increase.

\subsection{Lattice surgery \label{subsec:lattice_surgery}}

We now describe lattice surgery for measuring Pauli operators on the logical patches.
Color codes allow to measure a pair of commuting Pauli operators at the same time \cite{thomsen2024low}.
Given a pair of Pauli operators $(\lgc{P}, \lgc{Q})$ to measure, the structure of the ancillary region between the patches is determined appropriately based on certain rules.
Hereafter we only consider the cases where $\lgc{P}$ and $\lgc{Q}$ contain only $\lgi$'s and $\lgz$'s (which are sufficient for performing the MSD circuit in Fig.~\ref{fig:msd_circuit}) and the layouts in Fig.~\ref{fig:msd_layout} are used.
See Appendix~\ref{app:general_lattice_surgery_scheme} for more general cases.

A lattice surgery operation consists of three steps: the initialization of the ancillary region, merging operation, and splitting operation.
Supposing that the red boundary of each patch is in contact with the ancillary region (as depicted in Fig.~\ref{fig:msd_layout}), we first initialize the ancillary region with a red temporal boundary; that is, we prepare the Bell state $\ket{\Phi_+} \coloneqq (\ket{00} + \ket{11})/\sqrt{2}$ on every red edge in the region.
We then merge the patches by measuring checks in the ancillary region, which is repeated $\dm$ times to correct temporal error chains.
Here, checks in the interface regions should be chosen appropriately, ensuring that each of $\lgc{P}$ and $\lgc{Q}$ can be expressed as $R_\rbs \prod_{S \in G_\mr{anc}} S$, where $R_\rbs$ is a \rx-string ($\rz$-string) operator for $\lgc{P}$ ($\lgc{Q}$) and $G_\mr{anc}$ is a certain set of checks in the ancillary region.
(Note that $R_\rbs = \identity$ in every case that we will consider for MSD, but it can be nontrivial in general lattice surgery.)
After that, we split the patches by measuring the ancillary region with a red temporal boundary (i.e., measuring every red edge in the Bell basis), thereby obtaining the measurement outcomes of $\lgc{P}$ and $\lgc{Q}$.
In addition, the Pauli frames of the logical qubits may need to be updated depending on the measurement outcomes of red edges.

\begin{figure}[!t]
	\centering
	\includegraphics[width=\linewidth]{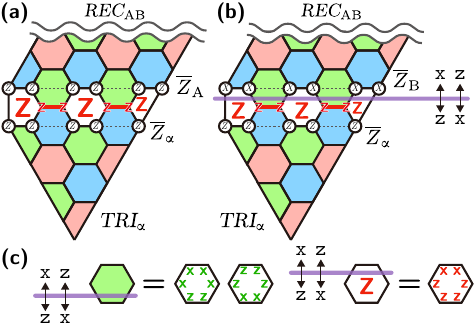}
	\caption{
        Lattice surgery schemes for measuring \subfig{a} $\lgztp{A\talpha}$ and \subfig{b} $\lgztp{B\talpha}$ for $\dz = 6$.
        Additional qubits and checks are placed between the two patches and serve as a semi-transparent domain wall, thereby measuring $\lgztp{A\talpha}$ or $\lgztp{B\talpha}$.
        Each face with a big `Z' written on it has only the $Z$-type check.
        Red edges indicated as thick red lines support two-body checks $Z \otimes Z$.
        In \subfig{b}, the purple thick line indicate a Pauli-permuting domain wall that swaps \xbs and \zbs, which deforms checks on the line as exemplified in \subfig{c}.
 }
	\label{fig:lattice_surgery_simple}
\end{figure}

We first consider a simple example of measuring $\lgztp{A\talpha}$ by using the layout in Fig.~\ref{fig:msd_layout}(c).
For this, in the interface region between \patchab and \patchalpha, we place an \ttt{em}-preserving semi-transparent domain wall that condenses \rz bosons on its both sides, which is expressed in terms of boson tables as
\begin{align*}
    \bosontableST[\text{\patchab}\quad\text{\patchalpha}]{r}{z}{r}{z}{2}{2}.
\end{align*}
Thereby, \gz and \bz ($\colorone{\magnetic}$) in \patchalpha are deconfined and can be mapped to \gz ($\colorone{\magnetic}$) in \patchab, while \rz is condensed in \patchalpha.
Therefore, a Pauli-$Z$ string-net operator in \patchalpha that represents $\lgzalpha$ can be transformed into a \gz-string operator in \patchab that represents $\lgzA$ by multiplying checks.
The microscopic structure of the domain wall is depicted in Fig.~\ref{fig:lattice_surgery_simple}(a) for $\dz = 6$.
In the interface region, the red faces only have $Z$-type checks and the red edges support two-body checks $Z \otimes Z$, together determining the value of $\lgztp{A\talpha}$.
Note that the domain wall contains $\dz-2$ additional data qubits.
Using weight-8 checks instead of weight-6 checks for the domain wall does not require additional data qubits \cite{kesselring2022anyon,thomsen2024low}, but it may significantly increase the time cost for syndrome extraction.

Similarly, $\lgztp{B\talpha}$ can be measured using an \ttt{em}-preserving domain wall that condenses \rx on the side of \patchab and \rz on the side of \patchalpha, which is expressed as
\begin{align*}
    \bosontableST[\text{\patchab}\quad\text{\patchalpha}]{r}{x}{r}{z}{2}{2}.
\end{align*}
The microscopic structure of the domain wall is depicted in Fig.~\ref{fig:lattice_surgery_simple}(b), where the purple line indicates a Pauli-permuting ($\xbs \leftrightarrow \zbs$) transparent domain wall, which deforms checks on the line as exemplified in Fig.~\ref{fig:lattice_surgery_simple}(c)

Let us now consider more complicated cases that use the full layout in Fig.~\ref{fig:msd_layout}(a) or~(b).
Our goal is to measure a pair of Pauli operators $(\lgc{P}, \lgc{Q})$ in parallel, which both consist of only $\lgz$ operators.
For this, we adapt the method in Ref.~\cite{thomsen2024low} to be applicable to rectangular patches as well as triangular patches.
That is, as displayed in Fig.~\ref{fig:lattice_surgery_full}(a) for the case of $\lgc{P} = \lgztp{OAD\talpha}$ and $\lgc{Q} = \lgztp{OBC\tbeta}$ (where solid and dotted lines are respectively Pauli-$Z$ and $X$ string operators), we place appropriate domain walls between the patches and the ancillary region, ensuring that each nontrivial factor of $\lgc{P}$ ($\lgc{Q}$) can be transformed into a \gx-string (\gz-string) operator `jumping over' the corresponding logical patch (i.e., connecting the two green boundaries of the ancillary region adjacent to the logical patch) by multiplying checks.
In addition, if a logical patch is not involved in $\lgc{P}$ ($\lgc{Q}$), we require any \gx-string (\gz-string) operator jumping over the patch to be trivial, i.e., to be a stabilizer.
By doing so, $\lgc{P}$ and $\lgc{Q}$ can be transformed into certain trivial green string operators in the ancillary patch, implying that their values can be determined from check measurement outcomes.

We now discuss the way to determine the types of the domain walls.
If a patch is not involved in both $\lgc{P}$ and $\lgc{Q}$, the corresponding domain wall is opaque; namely, the patch is completely separated from the ancillary patch that has a green boundary.
Other nontrivial cases are as follows:
First, the domain wall adjoining each triangular patch (\patchout, \patchalpha, or \patchbeta) is an \ttt{em}-exchanging semi-transparent one that condenses \rz and $\gbs\wbs$ on the sides of the patch and the ancillary region, respectively, where
\begin{align*}
    \wbs = \begin{cases}
        \xbs & \text{if the qubit is involved in $\lgc{Q}$ but not in $\lgc{P}$}, \\
        \zbs & \text{if the qubit is involved in $\lgc{P}$ but not in $\lgc{Q}$}, \\
        \ybs & \text{if the qubit is involved in both $\lgc{P}$ and $\lgc{Q}$}.
    \end{cases}
\end{align*}
For $\wbs = \xbs$, the domain wall is expressed in terms of boson tables as
\begin{align*}
    \bosontableST[q\quad\text{Anc. region}]{r}{z}{g}{x}{2}{1},
\end{align*}
thus a Pauli-$Z$ string-net operator in the patch (representing $\lgz_q$) can be transformed into a \gz-string operator in the ancillary region and any \gx-string operator jumping over the patch is trivial.
The case of $\wbs = \zbs$ can be interpreted analogously.
For $\wbs = \ybs$, $\lgz_q$ can be transformed into any of \gx- and \gz-string operators in the ancillary region, as the domain wall is expressed as
\begin{align}
    \bosontableST[q\quad\text{Anc. region}]{r}{z}{g}{y}{2}{1}.
    \label{eq:boson_tables_rz_gy}
\end{align}
Domain walls for rectangular patches are more diverse.
Let us consider \patchab as an example.
We denote the restrictions of $\lgc{P}$ and $\lgc{Q}$ on qubits A and B as $\lgc{P}_\mr{AB}, \lgc{Q}_\mr{AB} \in \qty{\identity, \lgzA \otimes \lgiB, \lgiA \otimes \lgzB, \lgzA \otimes \lgzB}$, respectively.
We also define a Pauli label 
\begin{align*}
    \pbs_\mr{AB} \coloneqq \begin{cases}
        \zbs & \text{if } \lgc{P}_\mr{AB} = \lgz_\mr{A} \otimes \lgi_\mr{B}, \\
        \xbs & \text{if } \lgc{P}_\mr{AB} = \lgi_\mr{A} \otimes \lgz_\mr{B}, \\
        \ybs & \text{otherwise}, \\
    \end{cases}
\end{align*}
and similarly $\qbs_\mr{AB}$ by replacing $\lgc{P}_\mr{AB}$ in the definition with $\lgc{Q}_\mr{AB}$.
If $\identity \neq \lgc{P}_\mr{AB} \neq \lgc{Q}_\mr{AB} \neq \identity$ (implying $\pbs_\mr{AB} \neq \qbs_\mr{AB}$), the domain wall is a Pauli-permuting transparent one that maps $\pbs_\mr{AB}$ and $\qbs_\mr{AB}$ in \patchab to $\xbs$ and $\zbs$ in the ancillary region, respectively.
Otherwise, it is an \ttt{em}-exchanging semi-transparent domain wall that condenses bosons $\ttt{a}$ and $\ttt{b}$ on the sides of \patchab and the ancillary region, respectively, where
\begin{align}
    \qty(\ttt{a}, \ttt{b}) = \begin{cases}
        \qty(\rbs\pbs_\mr{AB}, \gy) & \text{if } \lgc{P}_\mr{AB} = \lgc{Q}_\mr{AB} \neq \identity, \\
        \qty(\rbs\pbs_\mr{AB}, \gz) & \text{if } \lgc{P}_\mr{AB} \neq \lgc{Q}_\mr{AB} = \identity, \\
        \qty(\rbs\qbs_\mr{AB}, \gx) & \text{if } \lgc{Q}_\mr{AB} \neq \lgc{P}_\mr{AB} = \identity. \\
    \end{cases}
    \label{eq:rec_patch_domain_wall_condensed_bosons}
\end{align}
It is straightforward to show that the above configuration works properly from the fact that $\lgc{P}_\mr{AB}$ and $\lgc{Q}_\mr{AB}$ can be respectively represented by $\gbs\pbs_\mr{AB}$- and $\gbs\qbs_\mr{AB}$-string operators connecting the two green boundaries of \patchab.

\begin{figure*}[!t]
	\centering
	\includegraphics[width=\textwidth]{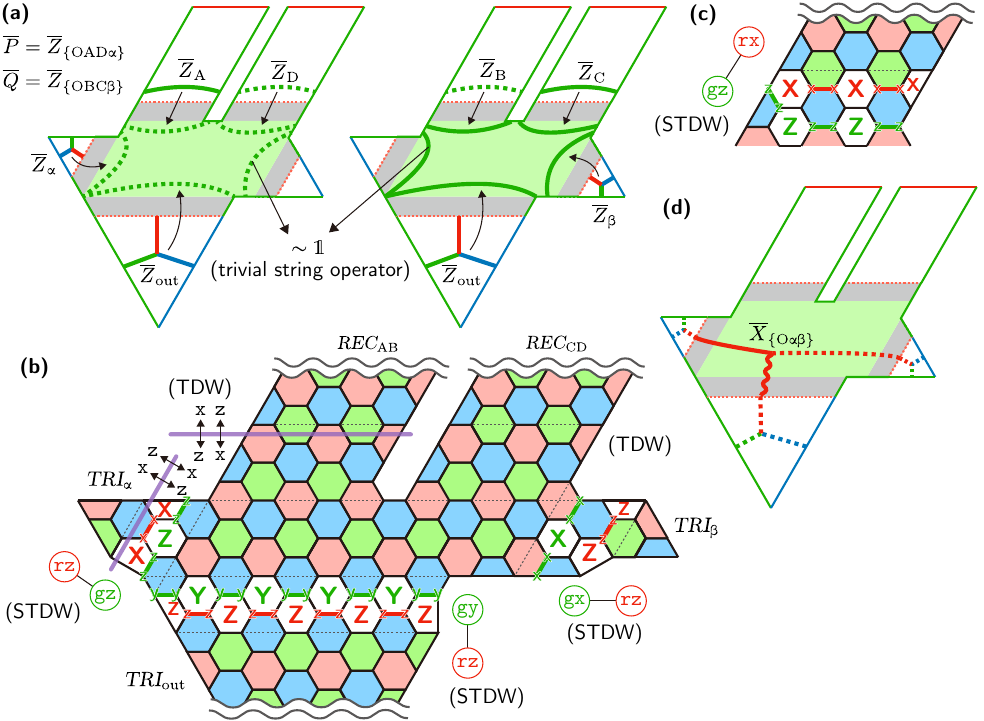}
	\caption{
        \subfig{a} Lattice surgery for measuring $\lgc{P} = \lgztp{OAD\talpha}$ and $\lgc{Q} = \lgztp{OBC\tbeta}$ in parallel.
        Pauli-$X$ ($Z$) string operators are visualized as thick dotted (solid) lines.
        Domain walls are placed between the logical patches and ancillary region so that $\lgc{P}$ and $\lgc{Q}$ are respectively equivalent to certain \gx- and \gz-string operators in the ancillary patch.
        Note that the domain wall adjoining \patchalpha (\patchbeta) condenses \gz (\gx) on the side of the ancillary region, implying that any \gz-string (\gx-string) operator connecting its both sides is trivial.
        \subfig{b} Microscopic structure of the layout of the above process for $\dout = 9$, $\dz=6$, and $\dm=3$.
        irregular checks in the domain walls are explicitly presented (see Fig.~\ref{fig:lattice_surgery_simple} for the notations).
        TDW and STDW stand for transparent and semi-transparent domain walls, respectively.
        For each STDW, bosons condensing on its two sides are specified by a pair of circles connected by a line; e.g., the STDW adjoining \patchout condenses \rz and \gy on the sides of \patchout and the ancillary region, respectively.
        Domain walls adjoining rectangular patches are transparent in this example, but other measurements may require STDWs for them as shown in \subfig{c}.
        \subfig{d} String-net operator representing a new logical operator $\lgxtp{O\textalpha\textbeta}$ that replaces $\lgxout$ during the merging operation, where the wavy line indicates a Pauli-$Y$ string operator.
        After the splitting operation, the value of its portion belonging to the ancillary region is determined. 
        If it is $-1$, a Pauli correction $\lgzout$ is applied.
 }
	\label{fig:lattice_surgery_full}
\end{figure*}

As an example, in Fig.~\ref{fig:lattice_surgery_full}(b), we describe the microscopic structure of the layout to measure $\lgc{P} = \lgztp{OAD\talpha}$ and $\lgc{Q} = \lgztp{OBC\tbeta}$ for $\dout = 9$, $\dz = 6$, and $\dm=3$.
Checks are deformed appropriately to form necessary domain walls described above \cite{kesselring2022anyon}.
The domain walls adjoining the rectangular patches are transparent in this example, thus in Fig.~\ref{fig:lattice_surgery_full}(c) we additionally show how they can be structured when they are semi-transparent.

Lastly, it is important to note that the Pauli frames of the logical qubits need to be updated appropriately after the lattice surgery finishes.
This is necessary only when using the full layout as Fig.~\ref{fig:lattice_surgery_full}, not when using the simple layout as Fig.~\ref{fig:lattice_surgery_simple}.
Let us again consider the example of measuring $\lgc{P} = \lgztp{OAD\talpha}$ and $\lgc{Q} = \lgztp{OBC\tbeta}$.
During the merging operation, original $\lgx$'s of the logical qubits are no longer logical operators as they anticommute with $\lgc{P}$ or $\lgc{Q}$, thus we replace them with new logical operators $\lgxtp{O\textalpha\textbeta}$, $\lgxtp{A\talpha}$, $\lgxtp{B\tbeta}$, $\lgxtp{C\tbeta}$, and $\lgxtp{D\talpha}$, which anticommute with $\lgztp{O}$, $\lgztp{A}$, $\lgztp{B}$, $\lgztp{C}$, and $\lgztp{D}$, respectively.
(Note that $\lgztp{\talpha}$ and $\lgztp{\tbeta}$ are no longer independent from other logical operators, thus considering these five pairs of logical operators is sufficient.)
To identify a Pauli correction applied on the output qubit, we consider $\lgxtp{O\textalpha\textbeta}$, which can be represented by a string-net operator connecting the original representations of $\lgxout$, $\lgxalpha$, and $\lgxbeta$, as visualized in Fig.~\ref{fig:lattice_surgery_full}(d).
Here, the red wavy line indicates an \ry-string operator.
Note that we should be careful about the behavior of the string-net near each domain wall.
For example, its \rx-string part in \patchout can be connected with its \ry-string part in the ancillary region through the domain wall; see Eq.~\eqref{eq:boson_tables_rz_gy}.
The value of its portion belonging to the ancillary region can be determined by measuring red edges fault-tolerantly in the splitting operation.
If it is $-1$, we apply a Pauli correction $\lgzout$.

In general, for each logical qubit $q$ (encoded in a patch $\labelname{PAT}_q$) that is not qubit~\talpha or~\tbeta, we determine a Pauli correction on $q$ by the following method:
If both $\lgc{P}$ and $\lgc{Q}$ involve $q$, we choose a red string-net operator $R$ belonging to the ancillary region that consists of \rz-, \rx-, and \ry-string parts, which are connected with \patchalpha, \patchbeta, and $\labelname{PAT}_q$, respectively.
If only $\lgc{P}$ ($\lgc{Q}$) involves $q$, we choose an \rz-string (\rx-string) operator $R$ belonging to the ancillary region that connects \patchalpha (\patchbeta) and $\labelname{PAT}_q$.
If the value of $R$ is $-1$ as a result of measuring the red edges, we apply a Pauli correction $\lgz$ on $q$.

\subsection{Scheme \label{subsec:scheme}}

We finally describe the overall procedure of our MSD scheme by combining the above ingredients.
The 15 $\pi/8$-rotations in the MSD circuit of Fig.~\ref{fig:msd_circuit} are executed pairwise through eight \emph{stages} by adding one placeholder rotation $\rot{\identity}{\pi/8} = \identity$.
We denote the pair of rotations for the $k$-th stage as $\qty(\rot{\Pstage{k}}{\pi/8}, \rot{\Qstage{k}}{\pi/8})$, where $\Pstage{k}$ and $\Qstage{k}$ are logical Pauli operators.
The first two stages are allocated for four single-qubit rotations on validation qubits A, B, C, and D, which are rotations 1--4 in Fig.~\ref{fig:msd_circuit}, namely, $\Pstage{1} = \lgztp{A}$, $\Qstage{1} = \lgztp{C}$, $\Pstage{2} = \lgztp{B}$, and $\Qstage{2} = \lgztp{D}$.
These are executed by using the simple layout in Fig.~\ref{fig:msd_layout}(c) instead of the regular one in Fig.~\ref{fig:msd_layout}(a) or~(b).
For example, to perform $\rot{\qty(\lgztp{A})}{\pi/8}$, we initialize qubit~$\mralpha$ to $\ket{\lgc{+}}$, measure $\lgztp{A\mralpha}$ via lattice surgery (which takes $\dm$ rounds), and perform a faulty T-measurement on qubit~$\mralpha$.
In the $k$-th stage where $k \geq 3$, we initialize qubits~\talpha and~\tbeta to $\ket{\ov{+}}$ and simultaneously measure 
\begin{align}
\begin{split}
    \PstageLS{k} &\coloneqq \Pstage{k}_{\mr{out,A,B,C,D}} \otimes \lgz_\mralpha \otimes \lgi_\mrbeta, \\ 
    \QstageLS{k} &\coloneqq \Qstage{k}_{\mr{out,A,B,C,D}} \otimes \lgi_\mralpha \otimes \lgz_\mrbeta
\end{split}
    \label{eq:dist_operators_to_measure}
\end{align}
via lattice surgery (which takes $\dm$ rounds), followed by faulty T-measurements on \patchalpha and \patchbeta.
We should carefully track the Pauli corrections made by the lattice surgery and faulty T-measurements.
The former may make corrections of $\lgz$ operators depending on the measurement outcomes of red edges in the ancillary region.
The latter may make corrections of $\Pstage{k}$, $\Qstage{k}$, or $\Pstage{k}\Qstage{k}$ depending on the faulty T-measurement outcomes.
After performing all eight stages, we check whether the $X$-measurement outcomes after applying the accumulated Pauli corrections are all equal to $+1$.
If then, we conclude that the distillation succeeds and a distilled logical magic state $\lgx\ket{\lgc{A}}$ (or $\lgy\ket{\lgc{A}}$ if there is an odd number of corrections on the output qubit) is produced from \patchout.
Otherwise, the entire process is retried.


We cannot arbitrarily determine the six pairs of 11 rotations $\qty{\qty(\rot{\Pstage{k}}{\pi/8}, \rot{\Qstage{k}}{\pi/8})}_{k=3}^8$, due to the following two factors:
(i) The number of errors that cause logical errors on the output state varies depending on this configuration.
(ii) Some configurations damage the fault tolerance of the layout in Fig.~\ref{fig:msd_layout}(a) or~(b); namely, they may result in the existence of undetectable string operators equivalent to a logical operator of a patch but shorter than the corresponding code distance.
To spoil the conclusion first, the configuration presented in Fig.~\ref{fig:msd_circuit} is one of the optimal ones considering these two.

Elaborating on the first factor, memory errors on the output qubit can be reduced by preparing the output qubit as late as possible.
This can be achieved by executing rotations 1--8 (that do not involve the output qubit) first, followed by the remaining rotations.
Note that this configuration has an additional advantage of allowing enough time to consume the output magic state even if the next MSD starts immediately.
Furthermore, among four possible locations of $\lgxout$ errors (right after stages 5, 6, 7, and 8), we can make two of them (right after stages~5 and~6) not harmful.
A $\lgxout$ error right after stage 5 is always unharmful since it makes only two correlated rotation errors, which are detectable.
A $\lgxout$ error right after stage 6 can be made unharmful by setting stages~5 and~6 to include four rotations in two among three pairs $(\lgztp{OAB}, \lgztp{OCD})$, $(\lgztp{OAC}, \lgztp{OBD})$, and $(\lgztp{OAD}, \lgztp{OBC})$, which are pairs of weight-3 rotations acting on disjoint sets of validation qubits.
By doing so, no subset of these four rotations acts on the output qubit an odd number of times and acts on each validation qubit an even number of times, meaning that a $\lgxout$ error right after stage 6 is unharmful (see Appendix~\ref{app:MSD_X_error_tolerance_proof} for details on why this argument works).

Let us now consider the second factor on the fault tolerance of the layout.
During lattice surgery, each logical Pauli operator of a patch that commutes with the operators to measure becomes equivalent (under stabilizer multiplication) with certain string operators in the ancillary region.
We demand these string operators to have weights not smaller than the corresponding code distance of the patch; namely, the layout should be distance-preserving.
The following conditions should be satisfied for this:
\begin{condition}
    (i) $\dout - \dz \leq \dm < 2\dz$, (ii) $\lgztp{OCD} \in \qty{\Pstage{k}}_{i=1}^6$, and (iii) if $\dz > 2\dm + 2$, $\rot{\qty(\lgz_\qty{s})}{\pi/8}$ and $\rot{\qty(\lgz_\qty{t})}{\pi/8}$ are not paired for each $(s,t)$ in (OAC, OBC), (OAD, OBD), (OAC, OAD), (OBC, OBD), (OCD, OABCD), and (OAB, OABCD).
\label{cond:dist_rotation_grouping}
\end{condition}
\noindent The first condition pertains to the layout itself, rather than the configuration of rotations, but we mention it here since it relates to this second factor.
In Appendix~\ref{app:layout_fault_tolerance}, we verify that the layout is distance-preserving if and only if these conditions are met.
For example, $\dout - \dz \leq \dm$ is required because, when measuring $\lgztp{OAB\talpha}$ or $\lgztp{OAB\tbeta}$, $\lgzout$ is equivalent to a green string operator in the ancillary region with weight $\dz + \dm$.
Other conditions can be derived similarly.

The configuration presented in Fig.~\ref{fig:msd_circuit} satisfies the above requirements, thus we set this as our default configuration.

\subsection{Resource costs \label{subsec:msd_resource_cost}}

We lastly evaluate the space and time costs of our MSD scheme.
The space cost of the scheme is quantified by the maximal number of physical qubits simultaneously required while executing the scheme.
Here, physical qubits include not only data qubits but also syndrome qubits, which are used to extract check measurement outcomes.
We suppose that each check has one syndrome qubit, which enables the simultaneous measurements of all checks during a single syndrome extraction round.
The time cost is quantified by the number of time steps required for each attempt of the scheme.
Note that a single syndrome extraction round can be done in eight time steps by selecting the entangling gate schedule appropriately \cite{beverland2021cost}.

A triangular patch with distance $d$ contains $(3d^2 + 1)/4$ data qubits and $(3d^2 -3)/4$ syndrome qubits, with a total of
\begin{align*}
    \ntri{d} \coloneqq \frac{3d^2 - 1}{2}
\end{align*}
qubits.
A rectangular patch with distances $d_1$ and $d_2$ contains $3d_1 d_2 / 2 - d_1 - d_2 + 2$ data qubits and $3d_1 d_2 / 2 - d_1 - d_2$ syndrome qubits, with a total of
\begin{align*}
    \nrec{d_1}{d_2} \coloneqq 3 d_1 d_2 - 2d_1 - 2d_2 + 2
\end{align*}
qubits.
The ancillary patch (surrounded by green boundaries) contains
\begin{align*}
    n_\mr{anc.patch} \coloneqq 3\dH \dV - 2\dH - 2\dV + 2
\end{align*}
qubits, where $\dH = \max(\dout + 1, 2\dz)$ and $\dV = \max(\dz, \dm + 1)$.
Lastly, the interface regions covering domain walls contain at most
\begin{align*}
    n_\mr{int} \coloneqq{}& 2 \cdot \qty[2\dout + 2\qty(3\dz - 2) + 2\dm + (3\dm + 1) ] \\ 
    ={}& 4\dout + 12\dz + 10\dm - 6
\end{align*}
qubits, excluding those that already belong to the ancillary patch.
Therefore, the space cost of the scheme is
\begin{align}
\begin{split}
    \spacecost{org}{\dout,\dx,\dz,\dm} \coloneqq \ntri{\dout} + 2\nrec{\dx}{\dz} \\
    + 2\ntri{\dm} + n_\mr{anc.patch} + n_\mr{int}
\end{split}
    \label{eq:org_space_cost}
\end{align}

To evaluate the time cost, we consider the following:
(i)~Each round consists of $\tround = 8$ time steps.
(ii)~For each stage, the merging operation takes $\dm$ rounds.
(iii)~After each merging operation, the measurement and reinitialization of the ancillary region together takes one round. Note that it actually takes four time steps (as we need two time steps to make Bell measurements), but we use a full round for it to synchronize syndrome extraction. If the merging operation is for the final stage, the reinitialization of the ancillary region is for the next MSD.
(iv)~Two time steps are required to initialize logical qubits, but it can be performed in parallel with initializing the ancillary region for the first stage. Similarly, logical qubits can be measured in parallel with measuring the ancillary region for the final stage.
(v) Each faulty T-measurement takes three time steps (i.e., two for Bell measurements and one for measuring $q_\mr{corner}$), but it can be performed in parallel during the extra single round in (iii) before starting the next stage.
Therefore, the time cost of the scheme is
\begin{align}
    \timecost{org}{\dm} \coloneqq 8\tround(\dm + 1),
    \label{eq:org_time_cost}
\end{align}
time steps.

\section{Production of higher-quality magic states \label{sec:higher_quality_magic_states}}

The MSD scheme in Sec.~\ref{sec:msd_scheme} cannot produce magic states with logical error rates lower than $\sim 35p_\mr{FT}^3$ even if there are no memory errors, where $p_\mr{FT} = O(p)$ is the error rate of the faulty T-measurement.
A conventional method to reach lower logical error rates is to concatenate the scheme with itself or another MSD scheme, namely, to input magic states produced from the scheme into MSD again.
However, here we take an alternative approach building on several recent results, combining distillation with recently-proposed distillation-free magic state preparation protocols, which we will refer to as `\textit{magic state cultivation}' following the terminology introduced in Ref.~\cite{gidney2024magic}.

Magic state cultivation is an alternative approach to distillation for generating logical magic states by leveraging the transversality of Clifford operations in color codes.
This method relies on the fact that the eigenstates of certain Clifford operators can serve as magic states; for example, $\ket{A}$ is a $+1$ eigenstate of $(X + Y)/\sqrt{2}$.
By measuring a logical qubit multiple times in the basis of such eigenstates (through noisy transversal controlled-Clifford gates) and post-selecting based on these measurement outcomes as well as the check outcomes, the final state is highly likely to be projected onto the desired magic state.
Notably, cultivation does not involve any multi-qubit logical operations, leading to its high resource efficiency compared to distillation.
For example, the scheme in Ref.~\cite{gidney2024magic} achieves an output infidelity of $10^{-9}$ at a spacetime cost of $\sim 10^6$ for $p=10^{-3}$ (see Fig.~\ref{fig:msd_performance_comparison} for more detailed comparison).

Several cultivation schemes have been proposed \cite{chamberland2020very,itogawa2024even,gidney2024magic}, differing in their detailed methods, such as circuits for measuring Clifford logical operators and color code checks.
The scheme in Ref.~\cite{chamberland2020very} employs flag qubits to detect all sets of faults at up to $(d-1)/2$ locations.
In this scheme, the logical Clifford operator and check measurement circuits are designed such that the flag and ancillary qubits switch roles depending on the operator to measure.
In Ref.~\cite{itogawa2024even}, the authors propose a cultivation scheme that uses only nearest-neighbor two-qubit gates on a square grid and enables efficient teleportation into surface codes.
The most recent work, Ref.~\cite{gidney2024magic}, refines these ideas by introducing a gradual increase in code size and advanced techniques for color codes (such as the superdense syndrome extraction circuit \cite{gidney2023new}), significantly improving resource efficiency and fault tolerance.
Importantly, unlike the two earlier works, the authors carefully address the process of growing the output magic state to a large code distance, showing that this step has a substantial impact on the fault tolerance and can be the most challenging part in the entire process.

Despite the high resource efficiency of cultivation, it has a fundamental limitation in scalability due to the extensive use of post-selection.
Specifically, the retry cost grows exponentially with the code distance, making it difficult to achieve output infidelities below a certain bound.
For instance, the scheme in Ref.~\cite{gidney2024magic} has been explored only for patches with $d \leq 5$, where the output infidelity is lower-bounded at $\sim 10^{-9}$ when $p=10^{-3}$ (noting that the discard rate reaches 99\% for $d=5$).
A promising approach to achieving lower output infidelities is combining cultivation and distillation.
That is, cultivation can be used to inject high-quality magic states into distillation, replacing non-fault-tolerant injection or the faulty T-measurement.
Notably, our MSD scheme, based on color codes, can be seamlessly integrated with cultivation by adjusting the single-level scheme described in Sec.~\ref{sec:msd_scheme}.
We elaborate on this cultivation-MSD scheme through the following subsections.

\subsection{Magic state cultivation \label{subsec:cultivation}}

Cultivation, the primary ingredient in our construction, is treated as a black box.
This process can be implemented using any of the schemes in Refs.~\cite{chamberland2020very,itogawa2024even,gidney2024magic} or with an improved scheme that may be proposed in the future.
It is performed on a triangular color code patch of distance $\dcult$, directly preparing a logical magic state $\ketlgc{A}$ with infidelity $q_\mr{cult}$ and a success rate of $q_\mr{cult}^\mr{succ}$.
(Note that the schemes in Refs.~\cite{chamberland2020very,itogawa2024even} prepare an eigenstate of the Hadamard gate rather than $\ketlgc{A}$; this difference can be addressed by adjusting the measurement bases in the rotation gate circuits in Fig.~\ref{fig:rotation_gate_circuit}.)
The space and time costs for a single successful cultivation attempt (without considering retrying) are denoted as $n_\mr{cult}$ and $t_\mr{cult}$, representing the number of physical qubits and the number of time steps, respectively.
For simplicity, we assume that the success of cultivation is determined immediately upon the completion of the process, although, in practice, it is more efficient to retry immediately after detecting a failure in the middle of the circuit.

\subsection{Growing operation \label{subsec:growing_operation}}

\begin{figure}[!t]
	\centering
	\includegraphics[width=\linewidth]{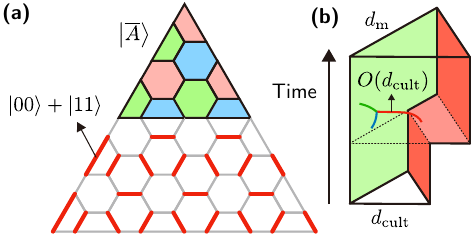}
	\caption{
        \subfig{a} Example of a growing operation of a triangular patch from $\dcult=5$ to $\dm=9$, which is done by preparing Bell states $\ket{00} + \ket{11}$ on additional red edges.
        \subfig{b} Spacetime picture of the operation, which can be interpreted as extending the red spatial boundary of the patch to a red temporal boundary.
        The blue spatial boundary (front face) is not colored for visibility.
        An example of a nontrivial string-net operator with weight $O(\dcult)$ is placed inside the diagram.
        }
	\label{fig:growing_operation}
\end{figure}

The distance $\dcult$ must be set small enough to minimize the retry cost; however, it is typically insufficient to store the generated magic state.
For instance, the scheme described in Ref.~\cite{gidney2024magic} can prepare a magic state with an infidelity of $6 \times 10^{-7}$ on a distance-3 patch when $p = 10^{-3}$, which corresponds to the error rate of a color code patch with a distance of around 19~\cite{lee2025color}.
Thus, it is crucial to grow the patch to a sufficiently large code distance, $d_\mr{m}$, immediately after completing cultivation.
As illustrated in Fig.~\ref{fig:growing_operation}(a), this can be performed by preparing Bell states, $\ket{00} + \ket{11}$, on additional red edges (simultaneously with the completion of cultivation) and performing regular check measurements thereafter.
In the spacetime picture, this operation can be interpreted as extending the red spatial boundary of the patch to a red temporal boundary, as shown in Fig.~\ref{fig:growing_operation}(b).

However, a critical obstacle is that the fault tolerance of the growing operation depends on $\dcult$, rather than $\dm$, due to error strings terminating at the red temporal boundary, as shown in Fig.~\ref{fig:growing_operation}(b).
Thus, even just a single round of the growing operation may significantly damage the cultivated magic state.
To address this issue, we employ post-selection during decoding for the growing operation. 
Specifically, we consider performing the growing operation and subsequent $\dm$ rounds of syndrome extraction, decoded jointly.
If the confidence of the prediction exceeds a preset threshold, we accept the final state; otherwise, we abort it and restart cultivation from the beginning.

One way to quantify confidence is by using the \textit{logical gap} (or \textit{complementary gap}) \cite{gidney2023yoked,bombin2024fault,smith2024mitigating}, defined as the minimum log-likelihood weight difference between the correction and an alternative correction in a different logical class. 
A larger logical gap indicates greater confidence in the prediction, thus we abort a trial when the logical gap is below a certain threshold $c_\mr{gap}$.
For surface codes, the logical gap can be computed by running the minimum-weight perfect matching (MWPM) decoder \cite{dennis2002topological} multiple times with different preassigned logical values \cite{gidney2023yoked,bombin2024fault,smith2024mitigating}, potentially raising the noise threshold to at most 50\% \cite{smith2024mitigating}.
For color codes, we calculate the logical gap using the \textit{concatenated MWPM decoder} \cite{lee2025color} (one of the best-performing circuit-level decoders for color codes in sub-threshold scaling) in a similar way by varying the preassigned logical values.
Note that, unlike surface codes, the concatenated MWPM decoder does not guarantee the least-weight correction, so the calculated logical gap is an approximation.

\begin{figure}[!t]
	\centering
	\includegraphics[width=\linewidth]{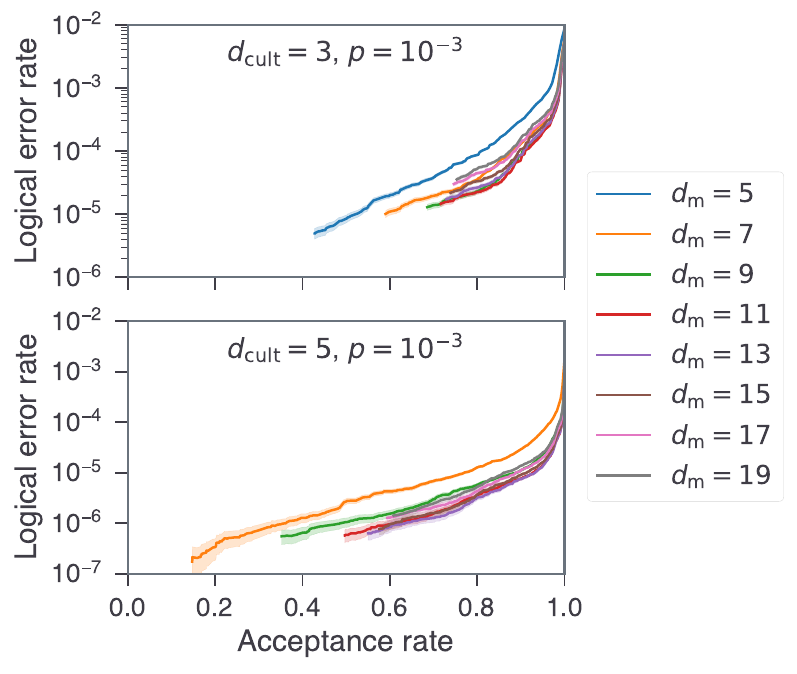}
	\caption{
        Simulations of the growing operation from distance $\dcult$ to $\dm$ with post-selection. 
        The summation of $\lgx$ and $\lgz$ failure rates are plotted against the acceptance rates for $\dcult \in \qty{3, 5}$, $\dm > \dcult$, and $p=10^{-3}$, obtained by varying the logical gap threshold $c_\mr{gap}$.
        The shaded regions indicate the 99\% confidence intervals.
    }
	\label{fig:growing_operation_plog_vs_pacc}
\end{figure}

To analyze the effect of post-selection, we run circuit-level simulations of the growing operation followed by $\dm$ rounds of syndrome extraction.
In Fig.~\ref{fig:growing_operation_plog_vs_pacc}, the logical failure rates (i.e., the summations of $\lgx$ and $\lgz$ failure rates) and the acceptance rates obtained by varying the logical gap threshold are presented for $p=10^{-3}$ and $\dcult \in \qty{3, 5}$.
(We say, e.g., $\lgx$ fails if a $\lgy$ or $\lgz$ error occurs.)
These results demonstrate that the post-selection method works very well with the concatenated MWPM decoder.
For instance, for $\dcult=3$ and $\dm \geq 7$, the logical error rate can be reduced from $\sim 10^{-2}$ to $\sim 10^{-4}$ ($3 \times 10^{-5}$) by aborting only 10\% (20\%) of trials.
The details of the simulation method and additional numerical results are provided in Appendix~\ref{app:growing_operation_simulations}.

It is worth noting that the logical gap approach does not seem to work well with the Möbius decoder \cite{sahay2022decoder,gidney2023new} (another matching-based color code decoder with performance comparable to the concatenated MWPM decoder), as investigated in Ref.~\cite{gidney2024magic}.
The reason for this discrepancy remains unclear, and it is uncertain whether this is a fundamental limitation of the Möbius decoder or an issue that could be overcome.

\subsection{Cultivation-MSD scheme \label{subsec:combined_scheme}}

\begin{figure}[!t]
	\centering
	\includegraphics[width=\linewidth]{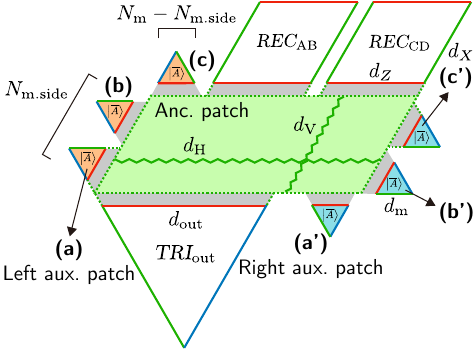}
	\caption{
        Layout for the cultivation-MSD scheme.
        The layout in Fig.~\ref{fig:msd_layout}(a) or~(b) is modified to contain $2\Nm$ triangular patches ($\Nm = 3$ in the figure), where magic states $\ketlgc{A}$ are prepared via cultivation and the growing operation.
        Among them, $N_\mr{m.side} \coloneqq \min(\Nm, \lceil(\dz / (\dm + 1))\rceil)$ patches are placed on each of the left and right sides of the ancillary region ($N_\mr{m.side} = 2$ in the figure), while the other patches are placed next to \patchout or \patchab.
        They are divided into two groups, left and right auxiliary patches, which are labeled as \subfig{a}, \subfig{b}, \subfig{c} and \subfig{a'}, \subfig{b'}, \subfig{c'}, respectively.
        The horizontal and vertical dimensions of the ancillary patch are respectively $\dH$ and $\dV$ defined in Eq.~\eqref{eq:combined_scheme_anc_region_dim}.
 }
	\label{fig:msd_layout_twolevel}
\end{figure}

After preparing magic states via cultivation (followed by the growing operation), we input them into MSD.
The layout for the cultivation-MSD scheme is presented in Fig.~\ref{fig:msd_layout_twolevel}, which is a modified version of the original layout in Fig.~\ref{fig:msd_layout}(a) or~(b) to include $2\Nm$ auxiliary triangular patches with distance $\dm$.
To minimize the ancillary region, we place
\begin{align*}
    N_\mr{m.side} \coloneqq \min\qty(\Nm, \ceil{\frac{\dz}{\dm + 1}})
\end{align*}
auxiliary patches on each of the left and right sides of the ancillary region, while the others are placed next to \patchab or \patchout.
Note that $\Nm = 3$ and $N_\mr{m.side} = 2$ in the example of Fig.~\ref{fig:msd_layout_twolevel}.
We group them into `left' and `right' auxiliary patches, which are respectively labeled as (a), (b), (c) and (a'), (b'), (c') in Fig.~\ref{fig:msd_layout_twolevel}. 
The horizontal and vertical dimensions of the ancillary patch are respectively
\begin{align}
\begin{split}
    \dH &\coloneqq \max\qty(2\dz, \dout + 1) \\ 
    &\quad + \qty(\Nm - N_\mr{m.side})\qty(\dm + 1), \\ 
    \dV &\coloneqq \max\qty[\dz, N_\mr{m.side}\qty(\dm + 1)].
\end{split}
\label{eq:combined_scheme_anc_region_dim}
\end{align}


Cultivation is executed in each auxiliary patch as frequently as possible, ensuring that successive executions within each of the left and right groups are separated by at least $T_\mr{m}$ rounds, where
\begin{align*}
    T_\mr{m} \coloneqq \ceil{\frac{t_\mr{cult}}{\tround\Nm}}.
\end{align*}
For instance, in Fig.~\ref{fig:msd_layout_twolevel}, after initiating cultivation in patch~(a), $T_\mr{m}$ rounds must pass before cultivation can begin in patch~(b) or~(c).
Once cultivation is completed successfully, the patch undergoes the growing operation, followed by $\dm$ rounds of syndrome extraction.
(Exceptionally, if $\dm = \dcult$, these additional steps are omitted.)
If the logical gap obtained from decoding exceeds a preset threshold, the created magic state can then be used for distillation.
If either the cultivation or the growing operation fails, the resulting state in the patch is discarded, and cultivation is retried (subject to the aforementioned condition on successive executions).
Additionally, if lattice surgery involving other auxiliary patches is in progress when the magic state is created, the state is also discarded instead of waiting until the lattice surgery ends.
(This is to prevent additional logical errors on the magic state caused by idling, though it is not strictly necessary.)

We wait until two magic states, one from each of the left and right groups, are prepared and grown successfully.
The earlier-generated magic state should idle while waiting for the second state. 
(If a new magic state is generated while another within the same group is idling, the earlier one is discarded.)
These magic states are then consumed to perform a pair of $\pi/8$-rotations using the circuit shown in Fig.~\ref{fig:rotation_gate_circuit}(a).
This process is repeated throughout all eight stages of the MSD.
Note that, unlike the single-level scheme, single-qubit rotations in the stages~1 and~2 of Fig.~\ref{fig:msd_circuit} are implemented by using the same layout as other rotations.

The space cost of the scheme is
\begin{multline}
    \spacecost{comb}{\dout, \dx, \dz, \dm, \dcult, \Nm} \\
    \coloneqq \ntri{\dout} + 2\nrec{\dx}{\dz} + n_\mr{anc.patch} + n_\mr{int} \\
    + 2\Nm\Bigl[ n_\mr{cult} + \qty(1 - \delta_{\dm, \dcult})\qty{\ntri{\dm} - \ntri{\dcult}}  \Bigr],
\label{eq:combined_space_cost}
\end{multline}
where $\delta_{i,j}$ is the Kronecker delta defined to be 1 if $i=j$ and 0 otherwise, and
\begin{align*}
    n_\mr{anc.patch} &\coloneqq 3\dH \dV - 2\dH - 2\dV + 2, \\
    n_\mr{int} &\coloneqq 4\dout + 12\dz + 10\Nm\dm + 2\Nm - 8
\end{align*}
are the space costs of the ancillary patch and the interface regions for domain walls, respectively.
Denoting the expected number of rounds between the initiation of successive stages as $T_\mr{intv}$, the expected time cost of the scheme is
\begin{align}
    \timecost{comb}{\dm, \Nm, p} \coloneqq 8 \tround T_\mr{intv}.
    \label{eq:combined_time_cost}
\end{align}
In our numerical analysis that will be described in the next section, we estimate $T_\mr{intv}$ by simulating the aforementioned procedure of the scheme for 1000 stages.
Note that, through the simulation, we additionally estimate the average number of rounds $T_\mr{idle}$ that the auxiliary patches idle before they are consumed, which is used for analyzing errors.
The code we used for this simulation is available on GitHub \cite{github:msd-magic-state-prep-cycle-simulation}.
See Sec.~\ref{subsec:analysis_results} and Table~\ref{table:msd_performance} for explicit examples on the estimated values of $T_\mr{intv}$ and $T_\mr{idle}$.

\section{Performance analysis \label{sec:performance_analysis}}

\label{para:inthissection}
In this section, we numerically analyze the performance of our MSD schemes, including both single-level and cultivation-MSD schemes, in terms of their output infidelities and success probabilities.
For cultivation, we assume that the state-of-the-art scheme in Ref.~\cite{gidney2024magic} is employed.

\label{para:insteadofdirectly}
Instead of directly simulating the entire MSD circuit at once, we first calculate the probabilities of logical errors in individual patches (including timelike error strings) via Monte Carlo simulations and then carefully track their effects on the final magic state. 
Although this approach may be less reliable than fully simulating MSD, we adopt this for the following reasons:
(i) The target output infidelity is extremely low (e.g.,  $<10^{-9}$), and the system size is very large, making direct Monte Carlo simulations computationally infeasible.
(ii) Existing color code decoders \cite{sarvepalli2012efficient,delfosse2014decoding,maskara2019advantages,chamberland2020triangular,beverland2021cost,delfosse2021almostlinear,sahay2022decoder,sabo2022trellis,kubica2023efficient,zhang2023facilitating,takada2023highly,berent2023decoding, lee2025color} are primarily designed for idling gates and may require modifications to handle logical operations (particularly irregular checks in domain walls).
While such modifications are essential for practical MSD implementation, this work focuses on estimating achievable performance rather than detailing decoder adjustments.
We assume that irregular checks for lattice surgery do not significantly degrade decoding performance (see Sec.~\ref{subsec:analysis_method} for its partial justification).
(iii)~This modular approach allows different decoders or cultivation protocols to be tested by simply adjusting a few parameters in the logical error tracking step.
For this, we provide analytical expressions for infidelities and success probabilities in terms of logical error rates, which may assist other researchers in evaluating their decoders or cultivation protocols.

\subsection{Method \label{subsec:analysis_method}}

Our analysis method is described in Appendix~\ref{app:numerical_analysis_method} in detail.
We outline this here as follows.

We first simulate $4d$ rounds of syndrome extraction of triangular and rectangular patches with code distances $d \leq 21$ under circuit-level noise of $p \in [10^{-4}, 10^{-3}]$, decoded via the concatenated MWPM decoder.
The syndrome extraction circuit is carefully chosen to minimize the logical failure rate (see Appendix~\ref{app:entangling_gate_schedule} for more details on the selection method and the resulting circuit).
For rectangular patches, we assume that they have $\lgz$ ($\lgx$) failure rates proportional to $\dz$ ($\dx$), which makes it sufficient to simulate only the cases of $\dx=\dz=d$.
From these simulations, we estimate per-round logical failure rates of logical patches.
In addition, we simulate stability experiments \cite{gidney2022stability} by performing $T (\leq 14)$ rounds of syndrome extraction of a patch encoding no logical qubits, which give `per-area' logical failure rates caused by timelike errors.

For generalizing the outcomes to other regimes of $p$ and $d$ (or $T$), the computed values for the per-round/area logical failure rates are fitted into the ansatz
\begin{align}
    p_\mr{fail}(p, d) = \alpha \qty(\frac{p}{p_\mr{th}})^{\beta d + \eta}\qty[1 + \epsilon \qty(\frac{p}{p_\mr{th}})^{\zeta d^\lambda}]
    \label{eq:ansatz}
\end{align}
with seven parameters $p_\mr{th}$, $\alpha$, $\beta$, $\eta$, $\epsilon$, $\zeta$, $\lambda$, where $p_\mr{th}$, $\alpha$, $\beta$, and $\zeta$ are positive.
(For stability experiments, $d$ is replaced with $T$.)
Note that, for a more precise prediction, the ansatz contains a sub-leading order term with respect to $p/p_\mr{th}$, selected from several candidates via cross-validation to prevent overfitting.

In Appendix~\ref{app:decoder_simulations}, we detail the simulation method and the ansatz selection process, with the numerical results and the corresponding parameter estimates.

We denote the logical Pauli error rates of triangular and rectangular patches as $\plogtri{P}$ and $\plogrec{P}$, respectively, for a Pauli operator $P$.
To determine $\plogtri{X/Y/Z}$ from logical failure rates, we assume that $\plogtri{X} = \plogtri{Z}$ and $\plogtri{Y} = \pyratio \plogtri{X}$, where $\pyratio$ is a small non-negative number quantifying the contribution of Pauli-$Y$ string errors.
Similarly, for a rectangular patch, we assume $\plogrec{X_1} = \plogrec{X_2}$, $\plogrec{X_1 X_2} = \pyratio \plogrec{X_1}$, $\plogrec{Z_1} = \plogrec{Z_2}$, and $\plogrec{Z_1 Z_2} = \pyratio \plogrec{Z_1}$
The same assumption apply to the growing operation ($\pgrow{X} = \pgrow{Z}$ and $\pgrow{Y} = \pyratio \pgrow{X}$) and to cultivation, where only the output infidelity is known ($\pcult{X} = \pcult{Z} = q_\mr{cult}/(2 + \pyratio)$ and $\pcult{Y} = \pyratio \pcult{X}$).
Although $\pyratio$ may appear to be an artificial coefficient, we will later show in that it has a negligible impact on the MSD performance.

The next step is to map every possible logical Pauli error (caused by an error string in a logical patch or the ancillary region) during MSD to an equivalent noise channel acted on the output and validation qubits immediately before the final $\lgx$ measurements of the validation qubits.
The noise channel is of the form
\begin{align*}
    \Lambda_{\lgc{U},p_\mr{err}}: \; \lgc{\rho} \mapsto (1 - p_\mr{err})\lgc{\rho} + p_\mr{err} \lgc{U} \lgc{\rho} \lgc{U}^\dagger,
\end{align*}
where $\lgc{\rho}$ is the logical state of the output and validation qubits, $p_\mr{err}$ is the logical error rate (expressed in terms of the above notations on logical Pauli error rates), and $\lgc{U}$ is a product of $\pi/2$- or $(\pm \pi/4)$-rotations.
See Appendix~\ref{app:determining_noise_channels} for a exhaustive list of the possible error sources and the corresponding noise channels.

Denoting the set of noise channels as $\{\Lambda_{\lgc{U}_i, p_i}\}_{i=1}^{m}$, the unnormalized output state after the final measurements is
\begin{align*}
    \lgc{\rho}_\mr{out} = (\identity_\mr{out} \otimes \bra{\lgc{+}\lgc{+}\lgc{+}\lgc{+}}_\mr{ABCD}) \Lambda_\mr{noise} \qty(\ketbra{\lgc{\psi}_\mr{init}}),
\end{align*}
where
\begin{align*}
    \Lambda_\mr{noise} &\coloneqq \Lambda_{\lgc{U}_1,p_1} \circ \cdots \circ \Lambda_{\lgc{U}_m,p_m}, \\
    \ket{\lgc{\psi}_\mr{init}} &\coloneqq \ket{\lgc{A}_-}_\mr{out} \otimes \ket{\lgc{+}\lgc{+}\lgc{+}\lgc{+}}_\mr{ABCD}, \\
    \ket{\lgc{A}_-} &\coloneqq \frac{1}{\sqrt{2}} \qty(\ketlgc{0} + e^{-i\pi/4}\ketlgc{1}).
\end{align*}
The success probability $q_\mr{succ}$ of the scheme and the output infidelity $\infidMSD$ are then respectively given as
\begin{align*}
    q_\mr{succ} = \Tr\qty(\lgc{\rho}_\mr{out}), \quad
    \infidMSD = 1 - \frac{1}{q_\mr{succ}} \bra{\lgc{A}_-} \lgc{\rho}_\mr{out} \ket{\lgc{A}_-}
\end{align*}
In Appendix~\ref{app:MSD_performance_expressions}, we present analytic expressions of $q_\mr{succ}$ and $\infidMSD$ as functions of the physical error rate $p$, the code distances, and the logical error rates of several patches.

\begin{figure*}[!t]
	\centering
	\includegraphics[width=\textwidth]{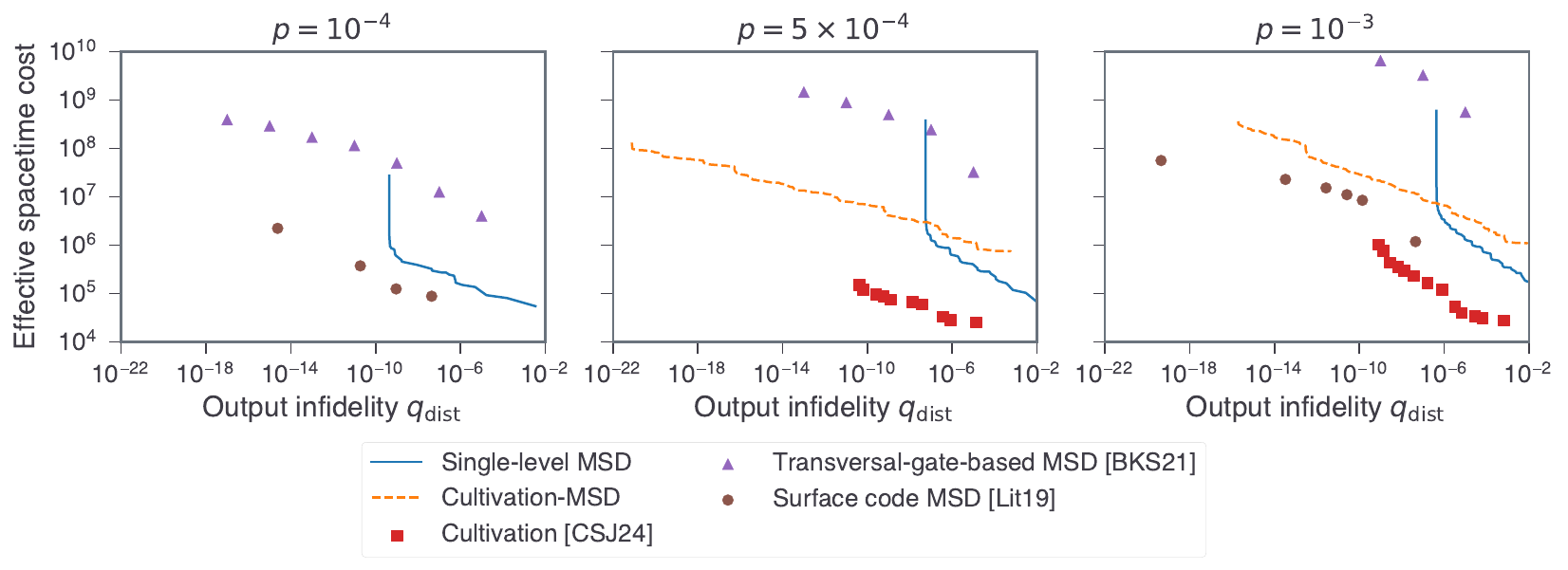}
	\caption{
        Cost comparison of various magic state preparation schemes. 
        The effective spacetime costs (i.e., spacetime costs divided by success probabilities $q_\mr{succ}$) are plotted against the output infidelities $\infidMSD$ for the single-level and cultivation-MSD schemes for various combinations of parameters, when $p \in \qty{10^{-4}, 5 \times 10^{-4}, 10^{-3}}$ and $\pyratio = 1/10$.
        For comparison, the costs of the following schemes are additionally displayed: (i) a variant of the cultivation-MSD scheme in which the growing operation is performed for a single round without post-selection, (ii) the cultivation protocol without MSD integration in Ref.~\cite{gidney2024magic}, (iii) the transversal-gate-based color code MSD scheme in Ref.~\cite{beverland2021cost}, and (iv) the surface code MSD schemes in Ref.~\cite{litinski2019magic}.
        We suppose that a single round of syndrome extraction takes 8 (6) time steps for color (surface) codes.
 }
	\label{fig:msd_performance_comparison}
\end{figure*}

\label{para:we_now_clarify}
We now specify several assumptions underlying the above method:
\begin{enumerate}
    \item We assume that the presence of irregular checks (e.g., two-body checks) within domain walls for lattice surgery has a negligible effect on the decoder’s performance (that is, we treat irregular checks as regular checks when estimating logical error rates).
    Although this has not been rigorously verified, we anticipate that this assumption may hold to a sufficient degree, as domain walls preserve the code distance and do not increase the required connectivity of data and ancillary qubits.
    In particular, irregular checks can reduce the number of entangling gates acting on each data qubit from six to five [see Fig.~\ref{fig:lattice_surgery_full}(b)], and since this number is a key factor in determining the `effective' error rate of the data qubit, fault tolerance might be even better within domain walls than in the bulk.
    Furthermore, a recent study~\cite{ramette2024fault} indicates that noise along a 1D region within a 2D code does not substantially affect performance, providing additional support for this assumption.
    \item We assume that, during lattice surgery, logical qubits that are being measured do not have logical errors anticommuting with any operators to measure.
    For example, when measuring $\Pstage{k} = \lgztp{OAB}$ and $\Qstage{k} = \lgztp{OCD}$, we ignore $\lgx$ and $\lgy$ errors on the output and validation qubits.
    This assumption is reasonable because error strings that lead to such logical errors have larger weights than the original code distances of the patches, as illustrated in Fig.~\ref{fig:lattice_surgery_full}(d).
    \item Due to the previous item, error strings within the ancillary region can incur only $\lgz_q$ errors on each logical qubit $q$ being measured.
    We estimate the corresponding $\lgz_q$ error rate by considering a hypothetical rectangular patch, with one boundary aligned along the boundary of the logical patch encoding $q$ and the other extending across the entire ancillary region (see Appendix~\ref{subsec:errors_ancillary_region} for more details).
    In this way, we can reasonably expect to avoid underestimating the $\lgz_q$ error rate, as the hypothetical patch always contains at least as many minimum-weight error strings equivalent to $\lgz_q$ than the ancillary region does, and their weights remain the same in both patches.
    \item For the cultivation-MSD scheme, we assume that stages begin at fixed intervals of $T_\mr{intv}$ rounds. 
    As defined in Sec.~\ref{subsec:combined_scheme}, $T_\mr{intv}$ is actually the expected value of these intervals.
\end{enumerate}

\subsection{Results \label{subsec:analysis_results}}

\begin{table*}[t!]
    \centering
    \begin{ruledtabular}
    \begin{tabular}{cccccccccc}
    Scheme & \begin{tabular}[c]{@{}c@{}}Output infidelity\\ $\infidMSD$\end{tabular} & \begin{tabular}[c]{@{}c@{}}Failure rate \\ $1 - \succprobMSD$\end{tabular} & \begin{tabular}[c]{@{}c@{}}Space cost $n$\\ (Qubits)\end{tabular} & \begin{tabular}[c]{@{}c@{}}Time cost $t$\\ (Time steps)\end{tabular} & \begin{tabular}[c]{@{}c@{}}Effective spacetime cost\\ ($nt/\succprobMSD$)\end{tabular} & $T_\mr{m}$ & $T_\mr{intv}$ & $T_\mr{idle}$ \\ \hline
    \hline \multicolumn{9}{c}{\subfig{a} $p = 5 \times 10^{-4}$} \\
    sng-$(11, 8, 6, 5)$ & $1.52 \times 10^{-5}$ & $4.71 \times 10^{-2}$ & $833$ & $384$ & $3.36 \times 10^{5}$ & \\
    sng-$(19, 10, 12, 7)$ & $1.02 \times 10^{-7}$ & $1.99 \times 10^{-2}$ & $2401$ & $512$ & $1.25 \times 10^{6}$ & \\
    cmb-$(23, 14, 16, 7, 3, 4, 5.03)$ & $1.13 \times 10^{-9}$ & $3.69 \times 10^{-3}$ & $5347$ & $759$ & $4.08 \times 10^{6}$ & 2 & 11.9 & 0.6 \\
    cmb-$(31, 18, 20, 11, 3, 3, 10.05)$ & $1.11 \times 10^{-12}$ & $1.00 \times 10^{-4}$ & $8825$ & $1298$ & $1.15 \times 10^{7}$ & 2 & 20.3 & 1.2 \\
    cmb-$(41, 22, 28, 13, 3, 4, 13.41)$ & $1.09 \times 10^{-15}$ & $3.02 \times 10^{-5}$ & $1.59 \times 10^{4}$ & $1348$ & $2.14 \times 10^{7}$ & 2 & 21.1 & 1.1 \\
    cmb-$(49, 30, 34, 15, 5, 4, 18.12)$ & $1.24 \times 10^{-18}$ & $4.09 \times 10^{-6}$ & $2.52 \times 10^{4}$ & $2272$ & $5.73 \times 10^{7}$ & 3 & 35.5 & 3.6 \\
    cmb-$(59, 34, 40, 19, 5, 6, 23.23)$ & $1.06 \times 10^{-21}$ & $4.19 \times 10^{-7}$ & $4.02 \times 10^{4}$ & $2378$ & $9.57 \times 10^{7}$ & 2 & 37.2 & 3.0 \\
    cmb-$(63, 40, 44, 19, 5, 8, 23.23)$ & $7.75 \times 10^{-22}$ & $4.21 \times 10^{-7}$ & $6.05 \times 10^{4}$ & $2073$ & $1.25 \times 10^{8}$ & 2 & 32.4 & 2.2 \\
    \hline \multicolumn{9}{c}{\subfig{b} $p = 10^{-3}$} \\
    sng-$(19, 8, 12, 7)$ & $1.21 \times 10^{-5}$ & $7.17 \times 10^{-2}$ & $2265$ & $512$ & $1.25 \times 10^{6}$ & \\
    sng-$(25, 12, 16, 11)$ & $1.03 \times 10^{-6}$ & $3.77 \times 10^{-2}$ & $4181$ & $768$ & $3.34 \times 10^{6}$ & \\
    cmb-$(29, 16, 20, 9, 3, 5, 6.09)$ & $1.04 \times 10^{-7}$ & $1.34 \times 10^{-2}$ & $9081$ & $925$ & $8.52 \times 10^{6}$ & 1 & 14.5 & 0.7 \\
    cmb-$(39, 22, 26, 13, 3, 4, 7.60)$ & $1.03 \times 10^{-9}$ & $1.96 \times 10^{-3}$ & $1.50 \times 10^{4}$ & $1391$ & $2.10 \times 10^{7}$ & 2 & 21.7 & 1.4 \\
    cmb-$(51, 28, 36, 15, 3, 4, 10.66)$ & $1.00 \times 10^{-11}$ & $8.37 \times 10^{-4}$ & $2.60 \times 10^{4}$ & $1595$ & $4.16 \times 10^{7}$ & 2 & 24.9 & 1.5 \\
    cmb-$(63, 36, 46, 17, 5, 6, 16.75)$ & $1.02 \times 10^{-13}$ & $1.95 \times 10^{-4}$ & $4.58 \times 10^{4}$ & $3020$ & $1.38 \times 10^{8}$ & 2 & 47.2 & 5.1 \\
    cmb-$(71, 36, 48, 23, 5, 8, 20.62)$ & $1.01 \times 10^{-15}$ & $2.60 \times 10^{-5}$ & $6.70 \times 10^{4}$ & $3513$ & $2.35 \times 10^{8}$ & 2 & 54.9 & 5.4 \\
    cmb-$(81, 44, 58, 23, 5, 8, 20.62)$ & $2.00 \times 10^{-16}$ & $2.65 \times 10^{-5}$ & $9.07 \times 10^{4}$ & $3513$ & $3.19 \times 10^{8}$ & 2 & 54.9 & 5.4
    \end{tabular}
    \end{ruledtabular}
    \caption{
        Output infidelities, failure rates, and resource costs of the single-level and cultivation-MSD schemes for various combinations of parameters at \subfig{a} $p=5 \times 10^{-4}$, and \subfig{b} $p = 10^{-3}$.
        Each variant of the single-level and cultivation-MSD schemes is labeled as `sng-$(\dout, \dx, \dz, \dm)$' or `cmb-$(\dout, \dx, \dz, \dm, \dcult, \Nm, c_\mr{gap})$'.
        For the cultivation-MSD scheme, the values of $T_\mr{m}$ (minimum number of rounds between successive executions of the CN protocol), $T_\mr{intv}$ (average number of rounds between successive stages), and $T_\mr{idle}$ (average number of rounds that auxiliary patches idle before being consumed) are additionally specified.
        The values of $T_\mr{intv}$ and $T_\mr{idle}$ are estimated by simulating the procedure of the cultivation-MSD scheme described in Sec.~\ref{subsec:combined_scheme} for 1000 stages.
    }
    \label{table:msd_performance}
\end{table*}

In Fig.~\ref{fig:msd_performance_comparison}, we plot the effective spacetime costs (i.e., the spacetime costs divided by the success probabilities) and output infidelities of the schemes for various parameter combinations when $p \in \qty{10^{-4}, 5 \times 10^{-4}, 10^{-3}}$ and $\pyratio = 1/10$.
Note that the single-level scheme has four parameters $(\dout, \dx, \dz, \dm)$ and the cultivation-MSD scheme has seven parameters $(\dout, \dx, \dz, \dm, \dcult, \Nm, c_\mr{gap})$, where $\dout - \dz \leq \dm < 2\dz$, $\dcult \leq \dm$, and $c_\mr{gap}$ is the logical gap threshold.
To highlight the effect of post-selection during the growing operation, we also plot the cost of a variant of the cultivation-MSD scheme where growing is performed for a single round without post-selection.
The cultivation-MSD scheme is not presented for $p = 10^{-4}$ due to the high computational cost of simulating the growing operation (see Appendix~\ref{app:growing_operation_simulations}).
However, based on the other two cases, it is reasonable to expect that its cost would exhibit similar behavior to the green line (cultivation-MSD with post-selection-free growing) for $q_\mr{dist} \gtrapprox 10^{-14}$ and eventually achieve very low infidelities under $10^{-22}$.

Figure~\ref{fig:msd_performance_comparison} shows that, if the target infidelity is within the range that the single-level scheme can reach (i.e., $\infidMSD \gtrapprox 35(7p/3)^3$), it is generally more preferable than the cultivation-MSD scheme.
If the target infidelity is lower than this, the cultivation-MSD scheme should be used, which can achieve infidelities higher than $7.7 \times 10^{-22}$ at $p = 5 \times 10^{-4}$ and $2.0 \times 10^{-16}$ at $p = 10^{-3}$.
We additionally note that the output infidelities do not significantly depend on $\pyratio$, as demonstrated in Appendix~\ref{app:pyratio_dependency}. 
Specifically, the difference in infidelity for the two extreme cases, $\pyratio = 0$ and $\pyratio = 1$, is at most only a factor of $\sim 1.5$.

In Table~\ref{table:msd_performance}, we list several data points in Fig.~\ref{fig:msd_performance_comparison} with more details on their failure probabilities and individual resource costs, where single-level and cultivation-MSD schemes are labeled as `sng-$(\dout, \dx, \dz, \dm)$' and `cmb-$(\dout, \dx, \dz, \dm, \dcult, \Nm, c_\mr{gap})$', respectively.

\label{para:for_comparison}
For comparison, Fig.~\ref{fig:msd_performance_comparison} additionally presents the costs of other schemes: the original cultivation scheme in Ref.~\cite{gidney2024magic}, the color code MSD scheme based on transversal gates in Ref.~\cite{beverland2021cost}, and the surface code schemes in Ref.~\cite{litinski2019magic}.
The cultivation protocol is notably resource-efficient as it does not require encoded operations; however, it is challenging to reach very low infidelities (e.g., $\infidMSD \lessapprox 10^{-10}$ when $p = 10^{-3}$ and $\infidMSD \lessapprox 10^{-11}$ when $p = 5 \times 10^{-4}$) since its success probability exponentially decreases as the fault distance grows.
In contrast, our cultivation-MSD scheme offers a clear advantage in this regard, enabling the achievement of much lower infidelities suitable for implementing a wide range of quantum algorithms.
Compared to the MSD scheme described in Ref.~\cite{beverland2021cost}, which employs transversal operations between stacked patches, our schemes show a significant reduction in spacetime cost, improving by about two orders of magnitude.
In comparison to surface code schemes, our color code schemes exhibit higher spacetime costs for a given target infidelity.
This difference is estimated to be less than one order of magnitude.

\begin{figure}[!t]
	\centering
	\includegraphics[width=\linewidth]{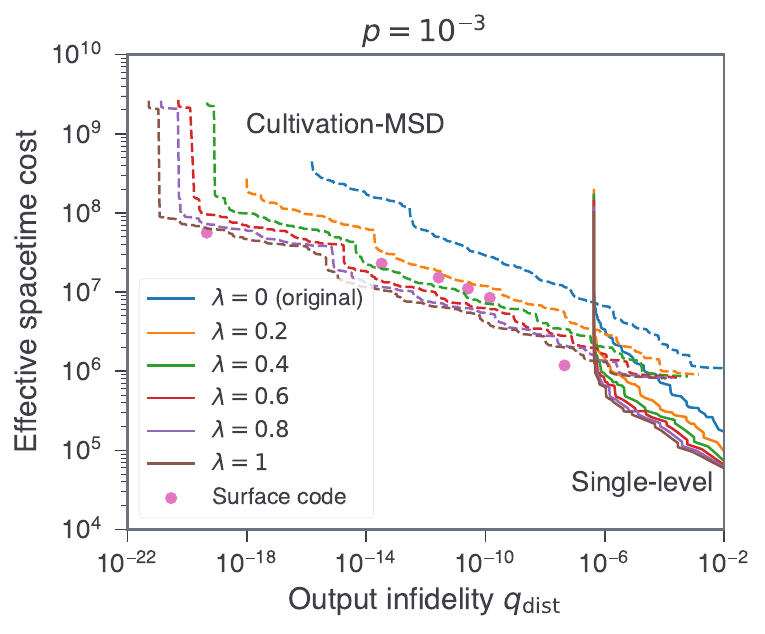}
	\caption{
        Costs of our MSD schemes under improved thresholds at $p=10^{-3}$, compared with those of the surface code schemes from Ref.~\cite{litinski2018lattice}.
        We assume that the threshold ($p_\mr{th}$) for each setting (including triangular/rectangular patch memory experiments and stability experiments) is improved to $(1 - \lambda) p_\mr{th} + \lambda(1\%)$, where $\lambda$ is a tuning parameter, while keeping all other parameters in the ansatz unchanged.
        Solid and dashed lines represent the single-level and cultivation-MSD schemes, respectively.
    }
	\label{fig:pth_adj_cost_analysis}
\end{figure}

\label{para:themainreason}
The main reason that MSD is more costly with color codes than with surface codes (despite their lower encoding rate and efficient lattice surgery capabilities) is their weaker fault tolerance under circuit-level noise. 
Specifically, the \textit{scaling threshold} (the parameter $p_\mathrm{th}$ in the sub-threshold ansatz given by Eq.~\eqref{eq:ansatz}) is estimated as $0.2\%$–$0.6\%$ for color codes under the concatenated MWPM decoder (see Appendix~\ref{subsec:ansatz_parameter_estimates}), whereas surface codes achieve a threshold of nearly $1\%$~\cite{fowler2009high,stephens2014fault,heim2016optimal,higgot2023improved}.
Nevertheless, circuit-level decoders for color codes continue to improve~\cite{beverland2021cost,gidney2023new,lee2025color}, motivating an analysis of how much further decoding methods need to improve for color codes to become preferable over surface codes in MSD applications.

In Fig.~\ref{fig:pth_adj_cost_analysis}, we roughly estimate the spacetime costs of our MSD schemes at $p=10^{-3}$, assuming that the scaling threshold ($p_\mr{th}$) for each setting used in our decoder simulations (see Appendix~\ref{app:decoder_simulations}) improves to $(1 - \lambda) p_\mr{th} + \lambda(1\%)$ with a tuning parameter $\lambda$, while all other parameters in the ansatz of Eq.~\eqref{eq:ansatz} are unchanged.
For the growing operation, we assume that the error rate at any $c_\mr{gap}$ is reduced by the same proportion as at $c_\mr{gap}=0$, while the acceptance rate remains invariant.
The figure indicates that at least $\lambda \approx 0.4$ is necessary for the cultivation-MSD scheme to be more resource-efficient than the surface code scheme for certain target infidelities.
This corresponds to thresholds of $p_\mr{th} = 0.54\%$ for triangular patches, $0.65\%$ ($\lgz$ failure) and $0.58\%$ ($\lgx$ failure) for rectangular patches, and $0.77\%$ for stability experiments (improved from the current thresholds of 0.24\%, 0.42\%, 0.31\%, and 0.62\%, respectively) \footnote{
    Note that these thresholds are \textit{scaling} thresholds, not \textit{crossing} thresholds.
    That is, they correspond to the parameter $p_\mathrm{th}$ in the sub-threshold ansatz given by Eq.~\eqref{eq:ansatz}, rather than the intersection points of curves for different code distances.
    These two thresholds generally differ, as the ansatz breaks down near threshold.
    Although the crossing threshold is commonly used as a standard definition in the literature, the scaling threshold is more relevant in the low-error regime as we consider in our analysis.
}.

\label{para:integrated_cultivation_growing}
Finally, we note that the performance analysis presented above relies on treating the cultivation and growing operations as separate modules. 
While this modular approach allows for flexibility in substituting different cultivation protocols, the interface between these two stages could potentially impact the overall performance. 
We investigate this effect in Appendix~\ref{app:integrated_cult_growing_simulations} by performing integrated simulations of cultivation followed immediately by the growing operation. 
These simulations account for detectors spanning the interface and their use in cultivation post-selection.
The results indicate a moderate increase in the spacetime cost of the cultivation-MSD scheme (by up to a factor of approximately 1.7 for the case $p = 10^{-3}$, $\dcult = 5$) compared to the separate analysis.
This increase, while not entirely negligible, remains relatively minor considering that the overall spacetime costs shown in Fig.~\ref{fig:msd_performance_comparison} span several orders of magnitude (from $10^4$ to $10^{10}$).
In addition, we estimate that the minimum achievable output infidelity could increase by roughly 3--30 times, which is again not negligible but still allows the scheme to remain sufficiently competitive, given that the original infidelity value was already extremely low at $\sim 2 \times 10^{-16}$.
See Appendix~\ref{app:integrated_cult_growing_simulations} for a detailed discussion and the corresponding simulation results.

\section{Remarks \label{sec:remarks}}

In this work, we proposed two magic state distillation (MSD) schemes for 2D color codes based on the 15-to-1 MSD circuit in Fig.~\ref{fig:msd_circuit}.
We presented the end-to-end descriptions of the schemes, from the definitions of logical qubits to the arrangement of logical-level operations, providing a comprehensive guide for their implementation.
The \textit{single-level MSD scheme}, covered in Sec.~\ref{sec:msd_scheme}, implements each $\pi/8$-rotation by employing the faulty T-measurement, i.e., measuring logical qubits non-fault-tolerantly in the basis of $\qty{\ketlgc{0} \pm e^{-i\pi/4} \ketlgc{1}}$.
The \textit{cultivation-MSD scheme}, covered in Sec.~\ref{sec:higher_quality_magic_states}, first fault-tolerantly prepares magic states via cultivation \cite{chamberland2020very,itogawa2024even,gidney2024magic}, grows the patches to accommodate the low error rate of the cultivated magic states, and finally consumes them to implement the $\pi/8$-rotations.
Both schemes can be performed via multiple stages of lattice surgery, which measures Pauli operators of multiple logical qubits by extending the lattice and connecting the patch boundaries.
The single-level scheme has a limited output infidelity of $\sim 35(7p/3)^3$ when $p$ is the physical error rate, while the cultivation-MSD scheme can achieve much lower infidelities (e.g., $\gtrapprox 2 \times 10^{-16}$ when $p=10^{-3}$) at the cost of higher resource overheads.

We optimized the resource costs of our MSD schemes through various approaches:
First, each lattice surgery operation measures two commuting Pauli operators in parallel, which is possible thanks to the rich topological structure of the color codes \cite{thomsen2024low}.
Thus, the 15 rotations can be conducted pairwise through eight stages.
Note that the way of pairing the rotations needs to be carefully chosen, as it affects the fault tolerance of MSD (see Sec.~\ref{subsec:scheme}).
In addition, logical qubits measured during MSD (such as the validation qubits A--D in Fig.~\ref{fig:msd_circuit}) can be encoded in patches with smaller code distances than the patch for outputting a magic state, which is an approach used in Ref.~\cite{litinski2019magic} for surface codes.
This is possible since memory errors in those qubits can be detected by the MSD circuit, while a logical error in the output patch directly affects the distilled magic state.
Moreover, the validation qubits are encoded in rectangular patches in Fig.~\ref{fig:color_code_patches}(b) instead of ordinary triangular patches in Fig.~\ref{fig:color_code_patches}(a).
Each rectangular patch encodes two logical qubits with their logical $Z$ ($\lgz$) operators supported on the same boundary (called its $Z$-boundary), while their $\lgx$ operators are supported on another boundary ($X$-boundary).
Since lattice surgery for the MSD circuit involves only $\lgz$ operators, we can let the ancillary region adjacent to only the $Z$-boundaries of the two rectangular patches, which is more resource-efficient than using four triangular patches.
Another advantage of using rectangular patches is that code distances ($\dx$, $\dz$) for $\lgx$ and $\lgz$ errors can be optimized separately, considering that they have different effects during MSD.
However, contrary to common belief, we revealed that the MSD circuit is tolerant to single-location $\lgx$ errors as well as $\lgz$ errors on validation qubits, thus $\dx$ and $\dz$ do not have to be significantly different as exemplified in Table~\ref{table:msd_performance}.

\label{para:forthecultivation}
For the cultivation-MSD scheme, a crucial problem was that the growing operation acted as a bottleneck, degrading a magic state generated through cultivation, as its fault tolerance depends on the code distance before growing.
We resolved this issue by employing post-selection based on the logical gap computed by the concatenated MWPM decoder \cite{lee2025color}.
The impact of post-selection is noteworthy: By aborting only 20\% of trials, the logical error rate can be suppressed by nearly three orders of magnitude when growing from $\dcult=3$ to $\dm\geq7$, as shown in Fig.~\ref{fig:growing_operation_plog_vs_pacc}.

By thoroughly investigating logical errors during MSD, we assessed the performance of our schemes based on their output infidelities and resource costs for various parameter combinations, as illustrated in Fig.~\ref{fig:msd_performance_comparison} and Table~\ref{table:msd_performance}. 
Notably, our schemes have about two orders of magnitude lower spacetime costs compared to the previous color code MSD scheme \cite{beverland2021cost}.
Furthermore, we verified that the cultivation-MSD scheme can reach sufficiently low output infidelities of $\sim 2 \times 10^{-16}$ when $p=10^{-3}$ and $\sim 7.8 \times 10^{-22}$ when $p=5 \times 10^{-4}$, which are challenging to achieve with either single-level MSD or cultivation alone.

\label{para:despitethisimprovement}
Despite this improvement, our color code schemes still do not surpass the performance of surface code MSD, exhibiting less than an order of magnitude higher spacetime overhead.
In order for color code schemes to reduce these overheads, we require improvements in color code decoders to achieve a higher circuit-level threshold (currently around 0.2\%--0.6\%), aiming closer to that of surface codes ($\sim 1\%$).
Specifically, by parametrizing the improvement in the \textit{scaling threshold} (the parameter $p_\mr{th}$ in the sub-threshold logical failure rate ansatz given by Eq.~\eqref{eq:ansatz}) as $p_\mr{th} \rightarrow (1 - \lambda) p_\mr{th} + \lambda (1\%)$, we estimated that achieving at least $\lambda = 0.4$ is necessary for our schemes to outperform the surface code scheme in Ref.~\cite{litinski2019magic}.
In practice, this requires a decoder achieving a scaling threshold of $\gtrapprox 0.54\%$ in memory experiments with triangular patches (improved from the current threshold of $0.24\%$).
Although this remains challenging, we emphasize that it may not be strictly necessary to reach the $1\%$ threshold to realize practical benefits.
The current \textit{cross thresholds} (the intersection points of logical failure rate curves at different code distances) of color code decoders are already around $0.5\%$ \cite{beverland2021cost,gidney2023new,zhang2024facilitating,lee2025color}, thus reducing the gap between the scaling and cross thresholds would be an important next step.
Note that the surface code scheme could also benefit from integration with cultivation, potentially widening the performance gap between surface and color codes.
However, converting a color code into a surface code (e.g., via the `grafting' technique proposed in Ref.~\cite{gidney2024magic}) introduces additional technical complexity and may require further simplification.
Exploring this would be a promising direction for future work as well.

\label{para:in_addition_our}
In addition, our analysis assumes that irregular checks for lattice surgery have a negligible effect on decoding performance.
Although this assumption was partially justified in Sec.~\ref{subsec:analysis_method}, developing refined color code decoders that explicitly account for domain walls would be a valuable next step. 
Such decoders would enable simulations of the complete end-to-end MSD circuit rather than relying on the modular approach used in our analysis, leading to a more accurate evaluation of MSD performance.

\label{para:moreover_utilizing}
Moreover, utilizing transversal Clifford gates for MSD is another intriguing avenue for investigation.
While lattice surgery typically requires $O(d)$ rounds, transversal gates can be executed in $O(1)$ rounds \cite{zhou2024algorithmic}, which enables a significant reduction in resource costs if Clifford gates between multiple logical qubits are implemented transversally rather than via lattice surgery.
However, two major challenges arise:
(i) Transversal gates between multiple logical patches inherently require long-range connectivity. 
Although stacking logical patches can partially alleviate this issue, gates involving distant logical patches may still demand a large number of swap operations unless hardware directly supports long-range connections.
(ii) Logical patches must have identical sizes to facilitate transversal gates, imposing constraints on the optimization of the MSD layout.
Nevertheless, leveraging transversal gates for MSD remains a promising direction for future research.
For instance, one might consider applying transversal gates to implement two-qubit Clifford operations between ancillary logical qubits of equal size, while using lattice surgery for gates connecting output and ancillary qubits with different code distances.
Furthermore, developing MSD schemes that fully exploit transversal gates could be particularly advantageous if hardware supports long-range connections, as explored in Refs.~\cite{zhou2024algorithmic,fazio2025low}.

\section*{Acknowledgements}

We thank Samuel C. Smith, Andrew Li, Lucas English, Sam Roberts, and Craig Gidney for helpful discussions and comments. 
This work is supported by the Australian Research Council via the Centre of Excellence in Engineered Quantum Systems (EQUS) project number CE170100009.  This article is based upon work supported by the Defense Advanced Research Projects Agency (DARPA) under Contract No.\ HR001122C0063. Any opinions, findings and conclusions or recommendations expressed in this article are those of the author(s) and do not necessarily reflect the views of the Defense Advanced Research Projects Agency (DARPA).

\nocite{github:color-code-msd-data}

\appendix

\section{Proof of fault tolerance for logical $X$ errors in the MSD circuit \label{app:MSD_X_error_tolerance_proof}}

\begin{figure*}[!t]
	\centering
	\includegraphics[width=\linewidth]{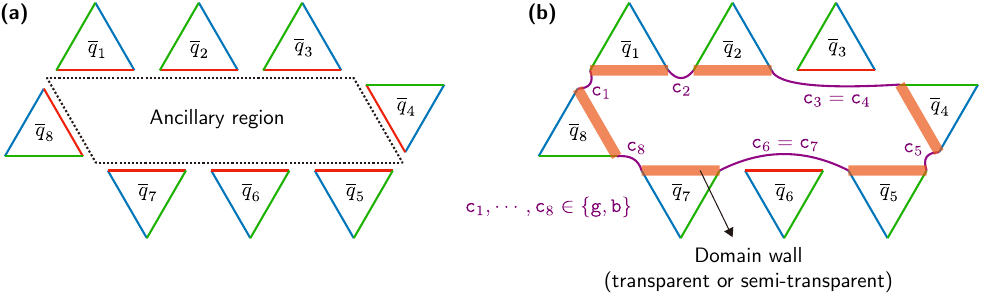}
	\caption{
        \subfig{a} Layout for lattice surgery, where $N=8$ triangular patches surround an ancillary region adjacent to their red boundaries.
        In each face, its red, green, and blue boundaries are placed clockwise.
        The eight logical qubits encoded in the patches are labeled as $\lqdata{1}, \lqdata{2}, \cdots, \lqdata{8}$.
        \subfig{b} Layout during the merging operation of lattice surgery, where $\lqdata{3}$ and $\lqdata{6}$ are not involved in the operation.
        The triangular patches are connected via the ancillary region with domain walls placed between the patches and the ancillary region.
        The ancillary region has boundaries (purple lines) between adjacent patches that are nontrivially involved in the operation.
        The colors of these boundaries are labeled as $\cbs_1, \cdots, \cbs_8 \in \qty{\gbs, \bbs}$, where $\cbs_3 = \cbs_4$ and $\cbs_6 = \cbs_7$ in this case.
    }
	\label{fig:general_lattice_surgery}
\end{figure*}

In this appendix, we prove that any single $\lgx$ error on a validation qubit in middle of the 15-to-1 MSD circuit in Fig.~\ref{fig:msd_circuit} does not make a logical error on the output qubit, assuming that there are no other logical errors.
Throughout this appendix, all the states and operators are in the logical level and overlines are omitted for readability.
We denote the set of qubits as $Q \coloneqq \qty{q_\mr{out}, q_\mr{A}, q_\mr{B}, q_\mr{C}, q_\mr{D}}$.

Denoting the $i$-th $\pi/8$-rotation of the MSD circuit as $\rot{P}{\pi/8}^{(i)}$ for each $i \in \qty{1, \cdots, 15}$, an $X$ error on a qubit $q$ after applying $\rot{P}{\pi/8}^{(n)}$ can be commuted to the beginning of the circuit and absorbed into the initial $\ket{+}$ state, which makes an error
\begin{align*}
    E_{q,n} \coloneqq \prod_{P \in \mathcal{P}_{q,n}} P_{-\pi/4},
\end{align*}
where
\begin{align}
    \mathcal{P}_{q,n} \coloneqq \left\{ P^{(j)} \; \middle\vert \; j \leq n \land q \in \supp P^{(j)}\right\}. \label{eq:P_qn_def}
\end{align}

We prove the following claim:
\begin{claim}
    For an arbitrary qubit state $\ket{\psi} \in \mathcal{H}_2$, a validation qubit $q \in \qty{q_\mr{A}, q_\mr{B}, q_\mr{C}, q_\mr{D}}$, and an integer $n \in \qty{1, \cdots, 15}$,
    \begin{multline}
        \qty(\identity_\mr{out} \otimes \bra{{+}{+}{+}{+}}_\mr{ABCD}) E_{q,n} \qty(\ket{\psi}_\mr{out} \otimes \ket{{+}{+}{+}{+}}_\mr{ABCD}) \\
        \propto \ket{\psi}_\mr{out}.  \label{eq:MSD_X_error_tolerance_proof}
    \end{multline} 
\end{claim}

\begin{proof}
    Since $\rot{P}{-\pi/4} = (\identity + iP)/\sqrt{2}$ for any Pauli operator $P$ and every $P^{(j)}$ contains only $Z$ operators, we have
    \begin{align}
        E_{q,n} \propto{}& \prod_{P \in \mathcal{P}_{q,n}} (\identity + iP) = \sum_{\widetilde{\mathcal{P}} \subseteq \mathcal{P}_{q,n}} \qty[ i^\abs{\widetilde{\mathcal{P}}} \prod_{P \in \widetilde{\mathcal{P}}} P ] \nonumber \\ 
        ={}& \sum_{\widetilde{\mathcal{P}} \subseteq \mathcal{P}_{q,n}} \qty[ i^\abs{\widetilde{\mathcal{P}}} \prod_{\substack{q' \in Q \\ \abs{\widetilde{\mathcal{P}}_{q'}}:\; \mr{odd}}} Z_{q'} ], \label{eq:E_qn}
    \end{align}
    where an empty product $\prod_{P \in \emptyset} P$ is defined as $\identity$ and
    \begin{align}
        \widetilde{\mathcal{P}}_{q'} \coloneqq \qty{P \in \widetilde{\mathcal{P}} \; \middle\vert \; q' \in \supp P}. \label{eq:abs_q_def}
    \end{align}
    It is straightforward to see that
    \begin{align*}
        \widetilde{\mathcal{P}}_{q} ={}& \widetilde{\mathcal{P}},\\ 
        \sum_{q' \in Q} \abs{\widetilde{\mathcal{P}}_{q'}} ={}& \sum_{P \in \widetilde{\mathcal{P}}} \abs{\supp P}.
    \end{align*}
    
    Assume for contradiction that the claim does not hold.
    Then the summation in Eq.~\eqref{eq:E_qn} contains at least one term involving only $Z_\mr{out}$ (and a constant factor), since all other terms except $\identity$ vanish when computing the left-hand side of Eq.~\eqref{eq:MSD_X_error_tolerance_proof}.
    Therefore, there exists a subset $\widetilde{\mathcal{P}} \subseteq \mathcal{P}_{q,n}$ such that $\abs{\widetilde{\mathcal{P}}_{q_\mr{out}}}$ is odd and $\abs{\widetilde{\mathcal{P}}_{q'}}$ is even for every $q' \in \qty{q_\mr{A}, q_\mr{B}, q_\mr{C}, q_\mr{D}}$, which entails that $\sum_{q' \in Q} \abs{\widetilde{\mathcal{P}}_{q'}} = \sum_{P \in \widetilde{\mathcal{P}}} \abs{\supp P}$ is odd.
    Since $\abs{\supp P^{(j)}}$ is odd for every $j$, $\abs{\widetilde{\mathcal{P}}} = \abs{\widetilde{\mathcal{P}}_{q}}$ is also odd, implying that $q$ should be the output qubit $q_\mr{out}$, which contradicts the precondition that $q$ is a validation qubit.
    This contradiction shows that the assumption is false, and hence the claim holds.
\end{proof}


\section{General lattice surgery scheme \label{app:general_lattice_surgery_scheme}}

In this appendix, we describe the macroscopic process of lattice surgery for measuring an arbitrary pair of commuting logical Pauli operators of triangular color codes.

\begin{figure*}[!t]
	\centering
	\includegraphics[width=\linewidth]{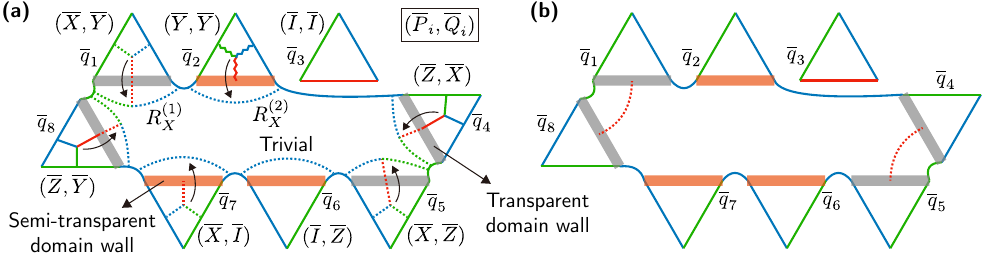}
	\caption{
        Example of a lattice surgery operation for measuring $\lgc{P} = \lgx_1\lgy_2\lgi_3\lgz_4\lgx_5\lgi_6\lgx_7\lgz_8$ and $\lgc{Q} = \lgy_1\lgy_2\lgi_3\lgx_4\lgz_5\lgz_6\lgi_7\lgy_8$.
        Solid, dotted, and wavy lines indicate Pauli-$Z$, $X$, and $Y$ string operators, respectively.
        \subfig{a} The string-net operator in the patch of $\lqdata{i}$, representing $\lgc{P}_i$, is transformed into a string-net operator $R_X^{(i)}$ inside the ancillary region via a transparent (gray) or semi-transparent (orange) domain wall, determined by Table~\ref{table:domain_wall_types}.
        Exceptionally, $\lqdata{3}$ and $\lqdata{6}$ are not involved in $\lgc{P}$, thus $R_X^{(3)}$ and $R_X^{(6)}$ are trivial (i.e., equivalent to identity).
        \subfig{b} The product of $R_X^{(i)}$'s is equivalent to a product of \rx-string operators connecting transparent domain walls within the ancillary region, meaning that the value of $\lgc{P}$ can be determined from the measurement outcomes of checks and red edges.
        The value of $\lgc{Q}$ can be determined similarly.
 }
	\label{fig:general_lattice_surgery_example}
\end{figure*}

We consider $N$ triangular logical patches surrounding an ancillary region adjacent to the red boundaries of the patches, as exemplified in Fig.~\ref{fig:general_lattice_surgery}(a) for $N=8$.
In each patch, the red, green, and blue boundaries are placed clockwise.
The $N$ logical qubits encoded in the patches are labeled as $\lqdata{1}, \cdots, \lqdata{N}$ in a clockwise order starting from an arbitrary patch.
We aim to measure two commuting Pauli operators $\lgc{P} = \bigotimes_{i=1}^N \lgc{P}_i$ and $\lgc{Q} = \bigotimes_{i=1}^N \lgc{Q}_i$ on these logical qubits, where $\lgc{P}_i, \lgc{Q}_i \in \qty{\lgi, \lgx, \lgy, \lgz}$ are logical Pauli operators on $\lqdata{i}$.
The procedure of the lattice surgery operation is as follows:
\begin{enumerate}
    \item Initialize the ancillary region with a red temporal boundary; that is, prepare the Bell state $\ket{\Phi_+} \coloneqq (\ket{00} + \ket{11})/\sqrt{2}$ on every red edge in the region.
    \item \textbf{(Merging operation)} Merge the triangular patches with the ancillary region by measuring appropriate checks, as shown in Fig.~\ref{fig:general_lattice_surgery}(b).
    The interior of the ancillary region has ordinary color-code checks.
    Domain walls are placed between this interior and the triangular patches, which are visualized as orange regions, where checks are deformed from ordinary color-code ones.
    Exceptionally, if $\lgc{P}_i = \lgc{Q}_i = \lgi$, the patch for $\lqdata{i}$ is not merged with the ancillary region, such as $\lqdata{3}$ and $\lqdata{6}$ in Fig.~\ref{fig:general_lattice_surgery}(b).
    The colors of the boundaries (purple lines) between the triangular patches  should be carefully chosen.
    Denoting them as $\cbs_1, \cdots, \cbs_N \in \qty{\gbs, \bbs}$ in a cyclic way from the boundary between $\lqdata{N}$ and $\lqdata{1}$, we determine them by the following rule:
    For every $i \in \qty{1, \cdots, N}$,
    \begin{align}
    \begin{split}
        \begin{cases}
            \cbs_{i+1} = \cbs_i & \text{if } \comm{\lgc{P}_i}{\lgc{Q}_i}=0, \\
            \cbs_{i+1} \neq \cbs_i & \text{otherwise},
        \end{cases}
    \end{split}
    \label{eq:rules_determining_boundary_colors}
    \end{align}
    where $\cbs_{N + 1} \coloneqq \cbs_1$.
    \item Wait for $d$ code cycles to correct timelike error strings in the ancillary region.
    \item \textbf{(Splitting operation)} Separate the triangular patches from the ancillary region by measuring it with a red temporal boundary. 
    The measurement outcomes of $\lgc{P}$ and $\lgc{Q}$ and the updated Pauli frames of the logical qubits are determined from the measurement outcomes of checks and red edges in the ancillary region.
\end{enumerate}

Note that, since $\lgc{P}$ and $\lgc{Q}$ commute, there are even flips of colors between neighboring boundaries connecting triangular patches, thus we can always find $\qty{\cbs_i}$ that satisfies the rule in Eq.~\eqref{eq:rules_determining_boundary_colors}.

The domain wall adjacent to a triangular patch $\lqdata{i}$ (nontrivially involved in the lattice surgery operation) is selected so that $\lgc{P}_i$ ($\lgc{Q}_i$) is equivalent to a Pauli-$X$ ($Z$) string-net operator $R_X^{(i)}$ ($R_Z^{(i)}$) in the ancillary region, which terminates at the adjacent boundaries of the ancillary region and may also terminate at the domain wall.
In other words, $\lgc{P}_i$ ($\lgc{Q}_i$) can be transformed into $R_X^{(i)}$ ($R_Z^{(i)}$) by multiplying checks in the domain wall.
See Fig.~\ref{fig:general_lattice_surgery_example}(a) for an example showing the transformation of $\lgc{P}$ when $\lgc{P} = \lgx_1\lgy_2\lgi_3\lgz_4\lgx_5\lgi_6\lgx_7\lgz_8$ and $\lgc{Q} = \lgy_1\lgy_2\lgi_3\lgx_4\lgz_5\lgz_6\lgi_7\lgy_8$.
As a result, $\lgc{P}$ ($\lgc{Q}$) is equivalent to $R_{X} \coloneqq \prod_{i=1}^N R_{X}^{(i)}$ ($R_{Z} \coloneqq \prod_{i=1}^N R_{Z}^{(i)}$) under stabilizer multiplication, which can be transformed further into a product of \rx-string (\rz-string) operators connecting domain walls in the ancillary region, as shown in Fig.~\ref{fig:general_lattice_surgery_example}(b).
Hence, we can determine the value of $\lgc{P}$ and $\lgc{Q}$ from the measurement outcomes of checks and red edges in the ancillary region including the domain walls.
The microscopic structure of the layout should be considered to specify exactly which checks and red edges are involved in this, which will not be covered in this appendix.


\begin{table*}[t!]
    \centering
    \begin{ruledtabular}
    \begin{tabular}{cccc}
        $\qty(\lgc{P}_i, \lgc{Q}_i)$ & \begin{tabular}[c]{@{}c@{}} Boundary colors \\ ($\cbs_i, \cbs_{i+1} \in \qty{\gbs, \bbs}$) \end{tabular} & Domain wall type & \begin{tabular}[c]{@{}c@{}} Effects on bosons \\ ($\pbs_i$, $\qbs_i$: Pauli charge labels for $\lgc{P}_i$, $\lgc{Q}_i$) \end{tabular} \\ \midrule
        $\comm{\lgc{P}_i}{\lgc{Q}_i} \neq 0$ & $\cbs_i \neq \cbs_{i+1}$ & Transparent & \begin{tabular}[c]{@{}c@{}} Permute charge labels with rules ($\lqdata{i}$ on left / Anc. on right):   \\ $\qty{\gbs \leftrightarrow \cbs_i, ~ \bbs \leftrightarrow \cbs_{i+1}, ~ \pbs_i \leftrightarrow \xbs, ~ \qbs_i \leftrightarrow \zbs}$. \end{tabular} \\
        $\lgc{P}_i = \lgc{Q}_i \neq \lgi$ & $\cbs_i = \cbs_{i+1}$ & Semi-transparent & Condense $(\rbs\pbs_i, \cbs_i\ybs)$; \ttt{em}-exchanging. \\
        $\lgc{P}_i \neq \lgc{Q}_i = \lgi$ & $\cbs_i = \cbs_{i+1}$ & Semi-transparent & Condense $(\rbs\pbs_i, \cbs_i\zbs)$; \ttt{em}-exchanging. \\
        $\lgc{Q}_i \neq \lgc{P}_i = \lgi$ & $\cbs_i = \cbs_{i+1}$ & Semi-transparent & Condense $(\rbs\qbs_i, \cbs_i\xbs)$; \ttt{em}-exchanging. \\
        $\lgc{P}_i = \lgc{Q}_i = \lgi$ & $\cbs_i = \cbs_{i+1}$ & Opaque & Condense \rx, \ry, \rz on $\lqdata{i}$ side and $\cbs_i\xbs$, $\cbs_i\ybs$, $\cbs_i\zbs$ on Anc. side. \\
    \end{tabular}
    \end{ruledtabular}
    \caption{Determination rule of the domain wall placed between the $i$-th triangular patch ($\lqdata{i}$) and the ancillary region (Anc.) during a lattice surgery operation for measuring $\lgc{P} = \bigotimes_{j=1}^N \lgc{P}_j$ and $\lgc{Q} = \bigotimes_{j=1}^N \lgc{Q}_j$.
    We say that a semi-transparent domain wall condenses $(\cbs\wbs,\cbs'\wbs')$ for two charge labels $\cbs\wbs$ and $\cbs'\wbs'$ if it condenses $\cbs\wbs$ and $\cbs'\wbs'$ on its sides facing $\lqdata{i}$ and the ancillary region, respectively.
    }
    \label{table:domain_wall_types}
\end{table*}

For the above argument to be valid, we should appropriately choose the type of each domain wall (for $\lqdata{i}$), which depends on $\lgc{P}_i$ and $\lgc{Q}_i$.
We present this determination rule in Table~\ref{table:domain_wall_types} for five different cases: (i) $\comm{\lgc{P}_i}{\lgc{Q}_i} \neq 0$, (ii) $\lgc{P}_i = \lgc{Q}_i \neq \lgi$, (iii) $\lgc{P}_i \neq \lgc{Q}_i = \lgi$, (iv) $\lgc{Q}_i \neq \lgc{P}_i = \lgi$, and (v) $\lgc{P}_i = \lgc{Q}_i = \lgi$.
Here, we denote the Pauli charge label corresponding to $\lgc{P}_i$ and $\lgc{Q}_i$ as $\pbs_i$ and $\qbs_i$, respectively (e.g., $\pbs_i = \xbs$ for $\lgc{P}_i = \lgx$). 
Let us examine these cases one by one, except the trivial case (v) where $\lqdata{i}$ and the ancillary region are disconnected (namely, an opaque domain wall is placed).

\subsection{$\comm{\lgc{P}_i}{\lgc{Q}_i} \neq 0$ ($\cbs_i \neq \cbs_{i+1}$)}


In this case, we place a transparent domain wall such that a boson crossing it from $\lqdata{i}$ to the ancillary region undergoes a charge permutation $\sigma(\cbs\wbs) = \cbs'\wbs'$, where
\begin{align*}
    \cbs' = \begin{cases}
        \cbs_i & \text{for } \cbs = \gbs, \\
        \cbs_{i+1} & \text{for } \cbs = \bbs, \\
        \rbs & \text{otherwise},
    \end{cases}\qquad
    \wbs' = \begin{cases}
        \xbs & \text{for } \wbs = \pbs_i, \\
        \zbs & \text{for } \wbs = \qbs_i, \\
        \ybs & \text{otherwise}.
    \end{cases}
\end{align*}
For instance, when $\qty(\lgc{P}_i, \lgc{Q}_i) = \qty(\lgy, \lgx)$ and $\qty(\cbs_i, \cbs_{i+1}) = \qty(\bbs, \gbs)$, $\lgc{P}_i=\lgy$ can be represented as a Pauli-$Y$ string-net operator connecting the three boundaries of the patch.
This operator can be transformed by multiplying stabilizers into a Pauli-$X$ string-net operator $R_X^{(i)}$ belonging to the ancillary region, as $\sigma(\ry) = \rx$, $\sigma(\gy)=\bx$, and $\sigma(\by)=\gx$.
Similarly, we can find a Pauli-$Z$ string-net operator $R_Z^{(i)} \sim \lgc{Q}_i= \lgx$, where `$U_1 \sim U_2$' for operators $U_1$ and $U_2$ denotes that $U_1$ and $U_2$ are equivalent under stabilizer multiplication.

\subsection{$\lgc{P}_i = \lgc{Q}_i \neq \lgi$ ($\cbs_i = \cbs_{i+1}$)}


In this case, we place a semi-transparent domain wall between $\lqdata{i}$ and the ancillary region.
This domain wall condenses $\rbs\pbs_i$ and $\cbs_i\ybs$ on the sides facing $\lqdata{i}$ and the ancillary region, respectively, and its charge-changing rule is \ttt{em}-exchanging.
For example, when $\lgc{P}_i = \lgc{Q}_i = \lgz$ and $\cbs_i = \cbs_{i+1} = \bbs$, the domain wall is characterized by the boson tables as
\begin{align*}
    \bosontableST[\lqdata{i}\quad\text{Anc.}]{r}{z}{b}{y}{2}{1},
\end{align*}
which is a notation defined in Sec.~\ref{subsubsec:STDW}.
Here, \gz and \bz ($\colorone{\magnetic}$) in $\lqdata{i}$ are deconfined and can be mapped to any of \bx and \bz ($\colorone{\electric}$) in the ancillary region, while \rz is condensed in $\lqdata{i}$.
Therefore, if we move a Pauli-$Z$ string-net operator that represents $\lgc{P}_i = \lgc{Q}_i = \lgz$ from the triangular patch to the ancillary region, we can obtain any of \bx- and \bz-string operators ($R_X^{(i)}$, $R_Z^{(i)}$) connecting the adjacent blue boundaries.
(Note that $R_X^{(i)} \sim R_Z^{(i)}$.)
Additionally, we can ensure that $\lgx$ and $\lgy$ cannot be moved outside the patch in a similar way since \gx, \bx, \gy, and \by are confined in the patch.

\subsection{$\lgc{P}_i \neq \lgc{Q}_i = \lgi$ or $\lgc{Q}_i \neq \lgc{P}_i = \lgi$ ($\cbs_i = \cbs_{i+1}$)}


If $\lgc{Q}_i$ is identity but $\lgc{P}_i$ is not, we place an \ttt{em}-exchanging semi-transparent domain wall that condenses $\rbs\pbs_i$ and $\cbs_i\zbs$ on its sides facing $\lqdata{i}$ and the ancillary region, respectively.
For example, when $\qty(\lgc{P}_i, \lgc{Q}_i) = \qty(\lgy, \lgi)$ and $\cbs_i = \cbs_{i+1} = \gbs$, the domain wall is characterized by
\begin{align*}
    \bosontableST[\lqdata{i}\quad\text{Anc.}]{r}{y}{g}{z}{2}{1}.
\end{align*}
Here, \gy and \by ($\colorone{\magnetic}$) are deconfined from $\lqdata{i}$ and can be mapped to any of \gx and \gy ($\colorone{\electric}$), while \ry is condensed on the $\lqdata{i}$ side of the domain wall.
Hence, we can find a \gx-string operator $R_X^{(i)}$ connecting the adjacent green boundaries that is equivalent to $\lgy$ of $\lqdata{i}$.
On the other hand, any \gz-string operator $R_Z^{(i)}$ connecting the adjacent green boundaries is trivial, as \gz is condensed on the ancillary region side of the domain wall.
Similar to the previous case, $\lgx$ and $\lgz$ are confined inside the patch.
For the opposite case where $\lgc{P}_i$ is identity but $\lgc{Q}_i$ is not, we can do similarly with a domain wall that condenses $\rbs\qbs_i$ and $\cbs_i\xbs$ on its two sides.

\section{Fault tolerance of the MSD layout \label{app:layout_fault_tolerance}}

In this appendix, we prove that our MSD layouts in Figs.~\ref{fig:msd_layout}(a) and~(b) are distance-preserving if and only if the conditions in Condition~\ref{cond:dist_rotation_grouping} are met.

\subsection{Notations}

\begin{figure}[!t]
    \centering
    \includegraphics[width=\columnwidth]{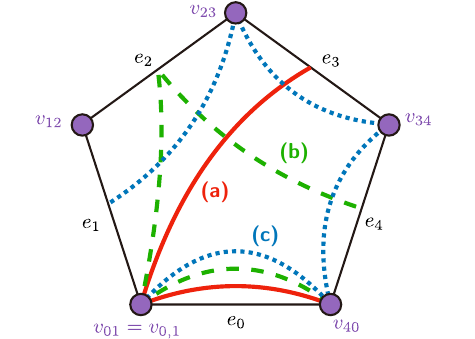}
    \caption{
        Schematics of the ancillary region and \subfig{a}-\subfig{c} three examples of cyclic piecewise green string (CPGS) operators.
    }
    \label{fig:dist_block_layout_conditions}
\end{figure}

We first define notations.
We schematize the ancillary region as a pentagon with edges labeled by $e_0$, $e_1$, $e_2$, $e_3$, and $e_4$ (in clockwise order), each of which is adjacent to one of the five patches (\patchout, \patchab, \patchcd, \patchalpha, and \patchbeta), as shown in Fig.~\ref{fig:dist_block_layout_conditions}.
We denote the set of these five edges as $E \coloneqq \qty{e_0, e_1, e_2, e_3, e_4}$.
We extend the series cyclically as $e_{i + 5l} = e_i$ for every integer $l$.
The patch adjacent to $e_i$ and the corresponding domain wall is referred to as \patch{i} and \dw{i}, respectively.
The vertex between $e_i$ and $e_j$ (where $j = i + 5l \pm 1$ for an integer $l$) is denoted as $v_{i,j}$ and a set of vertices $V \coloneqq \qty{v_{0,1}, v_{1,2}, v_{2,3}, v_{3,4}, v_{4,0}}$ is defined.
(We also sometimes omit commas in $v_{i,j}$, e.g., $v_{01} \coloneqq v_{0,1}$.)
Every vertex or edge represents a specific 1D subregion in the ancillary region:
An edge $e_i$ represents the ancillary region side of \dw{i} and a vertex $v_{i,j}$ represents the green boundary connecting \dw{i} and \dw{j}.
(Hereafter we identify each edge or vertex with the subregion represented by it.)
We also define the $i$-th code distance
\begin{align*}
    d_i \coloneqq \begin{cases}
        \dout & \text{if \patch{i} is \patchout}, \\ 
        \dz & \text{if \patch{i} is \patchab or \patchcd}, \\
        \dm & \text{otherwise}. \\
    \end{cases}
\end{align*}


We now consider a specific pair of Pauli operators $\qty(\lgc{P}, \lgc{Q})$ to measure.
Since every vertex $v_{i,i+1}$ is a green boundary, the set of bosons condensed by the vertex is
\begin{align*}
    \condset{v_{i,i+1}} \coloneqq \qty{\gx, \gy, \gz}.
\end{align*}
On the other hand, each edge condenses all or some of the green bosons depending on the Pauli operator measured on the corresponding patch.
Namely, the set of bosons condensed by $e_i$ is given as
\begin{align}
    \condset{e_i}= \condsete{i} \coloneqq \begin{cases}
        \qty{\gx, \gy, \gz} & \text{if}~ \lgc{P}_i = \lgc{Q}_i = \identity, \\ 
        \qty{\gx} & \text{if}~ \lgc{Q}_i \neq \lgc{P}_i = \identity, \\ 
        \qty{\gz} & \text{if}~ \lgc{P}_i \neq \lgc{Q}_i = \identity, \\ 
        \qty{\gy} & \text{if}~ \lgc{P}_i = \lgc{Q}_i \neq \identity, \\ 
        \emptyset & \text{otherwise},
    \end{cases}
    \label{eq:condset}
\end{align}
where $\lgc{P}_i,\lgc{Q}_i \in \qty{\lgi, \lgz}^{\otimes m_i}$ are respectively the restrictions of $\lgc{P}$ and $\lgc{Q}$ on the $m_i \in \qty{1, 2}$ logical qubit(s) encoded in \patch{i}; see Fig.~\ref{fig:lattice_surgery_full}(b) for an example.
For instance, if $\text{\patch{0}}=\text{\patchout}$, $\text{\patch{2}}=\text{\patchab}$, and $\lgc{P} = \lgztp{OAC\talpha}$, we get $\lgc{P}_0 = \lgc{P}_\mr{out} = \lgz$ and $\lgc{P}_2 = \lgc{P}_\mr{AB} = \lgz \otimes \lgi$.
Note that, if \patch{i} is \patchalpha (\patchbeta), $\condsete{i} = \qty{\gz}$ ($\qty{\gx}$).
For an index set $I = \qty{i_1, i_2, \cdots, i_l} \subset \mathbb{Z}$, we also define
\begin{align*}
    \condsete{I} = \condsete{i_1, i_2, \cdots, i_l} &\coloneqq \begin{cases}
        \qty{\gx, \gy, \gz} & \text{if } I = \emptyset, \\ 
        \bigcap_{i \in I} \condsete{i} & \text{otherwise}.
    \end{cases}
\end{align*}
We define a cyclic order of the set $V \cup E$ as 
\begin{multline*}
    v_{4,0} \prec e_0 \prec v_{0,1} \prec e_1 \prec v_{1,2} \prec e_2 \\
    \prec v_{2,3} \prec e_3 \prec v_{3,4} \prec e_4 \prec v_{4,0}.
\end{multline*}
Unlike a normal linear ordering, a cyclic order is only well-defined between three or more elements, such as $v_{4,0} \prec e_2 \prec v_{3,4}$.
For $o, o' \in V \cup E$, we define
\begin{align*}
    \indexbtn{o, o'} \coloneqq \qty{i \in \qty{0, 1, 2, 3, 4} \mid o \prec e_i \prec o'}.
\end{align*}
We say that $o$ and $o'$ are topologically connected with respect to a \gw boson (denoted by $o \sim_\wbs o'$) if and only if
\begin{align*}
    \gw \in \condset{o} \cap \condset{o'} \cap \qty[\condsete{\indexbtn{o, o'}} \cup \condsete{\indexbtn{o', o}}],
\end{align*}
meaning that an endpoint of a \gw-string operator at $o$ can be freely moved to $o'$.


For two objects $o_1, o_2 \in V \cup E$ and a Pauli label $\wbs \in \qty{\xbs, \ybs, \zbs}$, we define a set $\stropset{\wbs}{o_1}{o_2}$ to contain all \gw-string operators terminating \emph{only at} $o_1$ and $o_2$.
Here, we use the term `only at' to clarify that $\stropset{\wbs}{o_1}{o_2}$ does not include `piecewise' string operators such as $R_\mr{ac} R_\mr{cb}$ for $R_\mr{ac} \in \stropset{\wbs}{o_1}{o_3}$ and $R_\mr{cb} \in \stropset{\wbs}{o_3}{o_2}$ whose endpoints on $o_3$ are not the same.
Note that $\stropset{\wbs}{o_1}{o_2}$ is non-empty if and only if $\gbs\wbs \in \condset{o_1} \cap \condset{o_2}$.
We define
\begin{align*}
    \distpauli{\wbs}{o_1}{o_2} &\coloneqq \begin{dcases}
        \min_{R \in \stropset{\wbs}{o_1}{o_2}} \wt{R} & \text{if } \stropset{\wbs}{o_1}{o_2} \neq \emptyset, \\ 
        \infty & \text{otherwise,}
    \end{dcases} \\ 
    \dist{o_1,o_2} &\coloneqq \min_{\wbs \in \qty{\xbs, \ybs, \zbs}} \distpauli{\wbs}{o_1}{o_2},
\end{align*}
where $\wt{\cdot}$ denotes the weight of the given operator.
Note that $\distpauli{\wbs}{o_1}{o_2}$ is either $\dist{o_1,o_2}$ or $\infty$ for each \wbs.
We also define a `primed' distance as
\begin{align}
    \distprime{o_1, o_2} \coloneqq \min_{\substack{o'_1 \in \qty{o_1} \cup \ep{o_1} \\ o'_2 \in \qty{o_2} \cup \ep{o_2}}} \dist{o'_1, o'_2},
    \label{eq:primed_distance}
\end{align}
where $\ep{o}$ is the set of the endpoints of $o$, which is empty if $o$ is a vertex and $\qty{v_{i,i\pm1}}$ if $o$ is an edge $e_i$.
We denote a chain sum of primed or unprimed distances simply as
\begin{multline*}
    \Delta^{(\prime)}\qty(o_1, o_2, o_3, o_4, \cdots, o_l) \coloneqq \Delta^{(\prime)}\qty(o_1,o_2) + \Delta^{(\prime)}\qty(o_2,o_3) \\
    + \Delta^{(\prime)}\qty(o_3,o_4) + \cdots + \Delta^{(\prime)}\qty(o_{l-1},o_l)
\end{multline*}


We recall that some of the notations defined above, such as $\condset{\cdot}$, $\stropset{\wbs}{\cdot}{\cdot}$, and $\Delta^{(\prime)}(\cdot, \cdot)$, depend on the pair of Pauli operators $(\lgc{P}, \lgc{Q})$ to measure.
Therefore, when using this notation, we always suppose a specific distillation stage explicitly or implicitly.

\subsection{Condition for a layout to be distance-preserving}

The following is a sufficient condition for the layout to be distance-preserving:
\begin{condition}[\textbf{Fault tolerance of the ancillary region}]
    For every distillation stage and each $\wbs \in \qty{\xbs, \ybs, \zbs}$, if $\gw \notin \condsete{i} \cup \condsete{i-2,i-1,i+1,i+2}$ (that is, $v_{i,i-1}$ and $v_{i,i+1}$ are not topologically connected with respect to \gw), any operator equivalent to an operator in $\stropset{\wbs}{v_{i,i-1}}{v_{i,i+1}}$ should have a weight not less than $d_i$ to prevent any possible logical errors on \patch{i}.
\end{condition}
To exhaustively address every operator equivalent to an operator in $\stropset{\wbs}{v_{i,i-1}}{v_{i,i+1}}$, we consider \emph{cyclic piecewise green string (CPGS) operators}, which are trivial operators that can be written as products of nontrivial green string operators with the same Pauli label.
More formally, an $L$-segmented CPGS operator (for $L \geq 2$) has a form of
\begin{align}
    R_\mr{cyc} = \prod_{j=0}^{L-1} \strop{\wbs}{o_j}{\tilde{o}_j}, \label{eq:Rcyc_def}
\end{align}
where $\wbs \in \qty{\xbs, \ybs, \zbs}$, $o_j,\tilde{o}_j \in V \cup E$, $\strop{\wbs}{o_j}{\tilde{o}_j} \in \stropset{\wbs}{o_j}{\tilde{o}_j}$, and $o_j \not\sim_\wbs \tilde{o}_j \sim_\wbs o_{j+1}$ for each $j$ (denoting $o_L \coloneqq o_0$).
Objects $o_1, \tilde{o}_1, \cdots, o_{L-1}, \tilde{o}_{L-1}$ are called the endpoints of the CPGS operator.
For example, when only $e_1$ among the five edges condenses $\gbs\zbs$, $\strop{\zbs}{v_{4,0}}{v_{0,1}} \strop{\zbs}{v_{1,2}}{v_{2,3}} \strop{\zbs}{v_{2,3}}{v_{4,0}}$ is a 3-segmented CPGS operator.
It is worth noting that $o_j$ and $\tilde{o}_j$ cannot be adjacent edges (such as $e_0$ and $e_1$) or an edge and one of its endpoints (such as $e_0$ and $v_{0,1}$).
Without loss of generality, we can regard $\qty[o_0, \tilde{o}_0, o_1, \tilde{o}_1, \cdots, o_{L-1}, \tilde{o}_{L-1}, o_0]$ to be placed clockwise in order (namely, $o_0 \prec \tilde{o}_0 \prec o_1 \prec \cdots \prec \tilde{o}_{L-1} \prec o_0$); see Appendix~\ref{subapp:justification_clockwise_CPGS} below for its justification.

For $R_\mr{cyc}$ in Eq.~\eqref{eq:Rcyc_def} to exist, we need $\stropset{\wbs}{o_j}{\tilde{o}_j}$ to be non-empty and $o_j \not\sim_\wbs \tilde{o}_j \sim_\wbs o_{j+1}$ for each $j$.
Hence, a necessary and sufficient condition for this is
\begin{widetext}
\begin{align}
    \gw \in &\bigcap_{j=0}^{L-1} \big[ \overbrace{\condset{o_j} \cap \condset{\tilde{o}_j}}^{\text{(a)}} \cap \overbrace{\qty(\condsete{\indexbtn{\tilde{o}_j, o_{j+1}}} \cup \condsete{\indexbtn{o_{j+1}, \tilde{o}_j}})}^{\text{(b)}} \setminus \overbrace{\qty( \condsete{\indexbtn{o_j, \tilde{o}_j}} \cup \condsete{\indexbtn{\tilde{o}_j, o_j}})}^{\text{(c)}} \big] \nonumber \\
    =& \bigcap_{j=0}^{L-1} \qty[\condset{o_j} \cap \condset{\tilde{o}_j} \cap \condsete{\indexbtn{\tilde{o}_j, o_{j+1}}} \setminus \condsete{\indexbtn{o_j, \tilde{o}_j}}].
\label{eq:Rcyc_exist_cond}
\end{align}
\end{widetext}
Here we require part (a) for $\stropset{\wbs}{o_j}{\tilde{o}_j}$ to be non-empty, part (b) for $\tilde{o}_j \sim_\wbs o_{j+1}$, and part (c) for $o_j \not\sim_\wbs \tilde{o}_j$.
The equality holds since $\indexbtn{o_j, \tilde{o}_j} \subseteq \indexbtn{o_{j+1}, \tilde{o}_{j}}$ thus $\condsete{\indexbtn{o_{j+1}, \tilde{o}_{j}}} \setminus \condsete{\indexbtn{o_j, \tilde{o}_j}} = \emptyset$, and also, $\indexbtn{o_{j+1}, \tilde{o}_{j+1}} \subseteq \indexbtn{\tilde{o}_j, o_j}$ thus $\condsete{\indexbtn{\tilde{o}_j, o_j}} \subseteq \condsete{\indexbtn{o_{j+1}, \tilde{o}_{j+1}}}$.

For a given list of endpoints $\qty[o_1, \tilde{o}_1, \cdots, o_{L-1}, \tilde{o}_{L-1}]$, the corresponding CPGS operator exists for at least one Pauli label if and only if the right-hand side of Eq.~\eqref{eq:Rcyc_exist_cond} is not empty, which is equivalent to
\begin{align}
    \bigcap_{j=0}^{L-1} \qty[\condset{o_j} \cap \condset{\tilde{o}_j} \cap \condsete{\indexbtn{\tilde{o}_j, o_{j+1}}}] \nsubseteq \bigcup_{j=0}^{L-1} \condsete{\indexbtn{o_j, \tilde{o}_j}}.
    \label{eq:CPGS_exist_cond}
\end{align}
Provided that a CPGS operator $R_\mr{cyc}$ with $o_0 = v_{i-1,i}$ and $\tilde{o}_0 = v_{i,i+1}$ exists, $\strop{\wbs}{v_{i-1, i}}{v_{i,i+1}}$ is equivalent to $R_\mr{cyc} \strop{\wbs}{v_{i-1,i}}{v_{i,i+1}} = \prod_{j=1}^{L-1} \strop{\wbs}{o_j}{\tilde{o}_j}$, thus we need to impose a requirement
\begin{align}
    \sum_{j=1}^{L-1} \dist{o_j,\tilde{o}_j}  \geq d_i.
    \label{eq:distance_condition}
\end{align}
In addition, we can equivalently use the primed distance defined in Eq.~\eqref{eq:primed_distance} as
\begin{align}
    \sum_{j=1}^{L-1} \distprime{o_j,\tilde{o}_j}  \geq d_i.
    \label{eq:primed_distance_condition}
\end{align}
since, if an edge $e_j$ is an endpoint of $R_\mr{cyc}$, then $e_j \sim_\wbs v_{j,j\pm1}$, thus Eq.~\eqref{eq:distance_condition} also should hold if $e_j$ is replaced with $v_{j,j\pm1}$.

In summary, for given pairs of Pauli operators $\qty{\qty(\PstageLS{k}, \QstageLS{k})}_{k=3}^8$ measured in the last six distillation stages, the following is a sufficient condition for the layout to be distance-preserving:
\begin{condition}
    For each distillation stage that measures $\PstageLS{k}$ and $\QstageLS{k}$ (where $k \geq 3$), each $i \in \qty{0, 1, 2, 3, 4}$, and each possible odd-length sequence of vertices and edges $\left[o_0, \tilde{o}_0, o_1, \tilde{o}_1, \cdots, o_{L-1}, \tilde{o}_{L-1}, o_L \right]$ that are arranged clockwise in the pentagon and fulfills $o_0=o_L=v_{i-1,i}$ and $\tilde{o}_0=v_{i,i+1}$, if Eq.~\eqref{eq:CPGS_exist_cond} holds with respect to $\qty(\condsete{0}, \condsete{1}, \condsete{2}, \condsete{3}, \condsete{4})$ determined by $\qty(\PstageLS{k}, \QstageLS{k})$, then Eq.~\eqref{eq:distance_condition} or equivalently Eq.~\eqref{eq:primed_distance_condition} should be satisfied.
    \label{cond:dist_layout_condition}
\end{condition}

In Fig.~\ref{fig:dist_block_layout_conditions}, we present three examples of potential CPGS operators for $i=0$, which respectively give the following requirements:
\begin{align*}
        \text{(a):}\quad& \condsete{34} \subseteq \condsete{0} \cup \condsete{12} \Rightarrow \distprime{v_{01},e_3} \geq d_0, \\ 
        \text{(b):}\quad&\condsete{24} \subseteq \condsete{0} \cup \condsete{1} \cup \condsete{3} \Rightarrow \distprime{v_{01}, e_2, e_4} \geq d_0, \\ 
        \text{(c):}\quad&\condsete{1} \subseteq \condsete{0} \cup \condsete{2} \cup \condsete{3} \cup \condsete{4} \Rightarrow \distprime{e_1,v_{23},v_{34},v_{40}} \geq d_0.
\end{align*}

\begin{widetext}

\subsection{Verifying that the layouts in Fig.~\ref{fig:msd_layout} are distance-preserving}

We now verify that the layout of Fig.~\ref{fig:msd_layout}(a) or~(b) is distance-preserving if Condition~\ref{cond:dist_rotation_grouping} is met. 
We set \patchout, \patchalpha, \patchab, \patchcd, and \patchbeta respectively as \patch{0} to \patch{4}.
First thing to note is the primed distances between vertices and edges in the layout:
\begin{itemize}
    \item $\distprime{o,o'} = 0$ if $o$ and $o'$ are either neighboring edges or an edge and its endpoint (trivial case).
    \item $\distprime{o,o'} = \dm + 1$ if $o \in \qty{v_{j-1, j}, e_{j-1}}$ and $o' \in \qty{v_{j,j+1}, e_{j+1}}$ for $j \in \qty{1, 4}$.
    \item $\distprime{o,o'} = \dz$ if $o \in \qty{v_{j-1, j}, e_{j-1}}$ and $o' \in \qty{v_{j,j+1}, e_{j+1}}$ for $j \in \qty{2, 3}$.
    \item $\distprime{v_{40}, e_2} = \distprime{v_{40}, v_{23}} = \distprime{e_0, v_{23}} = \max\qty(\dm+1, \dz)$.
    \item $\distprime{o,o'} \geq \dout + 1$ otherwise.
\end{itemize}
\begin{subequations}
We suppose that the layout satisfies
\begin{align}
    &\dm, \dz < \dout, \label{eq:dist_block_distance_cond_1} \\ 
    &\dout - \dz \leq \dm \leq 2\dz, \label{eq:dist_block_distance_cond_2}
\end{align}
\label{eq:dist_block_distance_conds}%
\end{subequations}
where the second one corresponds to item (i) of Condition~\ref{cond:dist_rotation_grouping}, which will be justified later.
In addition, we recall that
\begin{align}
    \condsete{1} = \qty{\gz}, \quad \condsete{4} = \qty{\gx}, \label{eq:aux_qubit_condsets}
\end{align}
and, for each $k \geq 3$,
\begin{align*}
    \Pstage{k}, \Qstage{k} \in \mathcal{Z} \coloneqq \{ &\lgztp{OAB}, \lgztp{OAC}, \lgztp{OAD}, \lgztp{OBC}, \lgztp{OBD}, \lgztp{OCD}, \\
    &\lgztp{ABC}, \lgztp{ABD}, \lgztp{ACD}, \lgztp{BCD}, \lgztp{OABCD} \}.
\end{align*}

Although Condition~\ref{cond:dist_layout_condition} gives a vast amount of requirements, most of them are trivially true.
We list the remaining (potentially) nontrivial requirements as follows:
\begin{subequations}
\begin{align}
    i=0:~ &\condsete{234} \nsubseteq \condsete{0} \cup \condsete{1} \Rightarrow \distprime{v_{01}, e_2} = \dm + 1 \geq \dout, \label{eq:dist_layout_req_i0_234} \\ 
    &\condsete{134} \nsubseteq \condsete{0} \cup \condsete{2} \Rightarrow \distprime{e_1, e_3} = \dz \geq \dout, \label{eq:dist_layout_req_i0_134} \\
    &\condsete{12} \nsubseteq \condsete{0} \cup \condsete{34} \Rightarrow \distprime{e_2, v_{40}} = \max\qty(\dm+1, \dz) \geq \dout, \label{eq:dist_layout_req_i0_12} \\
    &\condsete{123} \nsubseteq \condsete{0} \cup \condsete{4} \Rightarrow \distprime{e_3, v_{40}} = \dm + 1 \geq \dout, \label{eq:dist_layout_req_i0_123} \\
    &\condsete{23} \nsubseteq \condsete{0} \cup \condsete{1} \cup \condsete{4} \Rightarrow \distprime{v_{01}, e_2} + \distprime{e_3, v_{40}} = 2(\dm + 1) \geq \dout, \label{eq:dist_layout_req_i0_23} \\ 
    &\condsete{14} \nsubseteq \condsete{0} \cup \condsete{2} \cup \condsete{3} \Rightarrow \distprime{e_1, v_{23}, e_4} = 2\dz \geq \dout, \label{eq:dist_layout_req_i0_14} \\ 
    i=1:~
    &\condsete{340} \nsubseteq \condsete{1} \cup \condsete{2} \Rightarrow \distprime{v_{12}, e_3} = \dz \geq \dm, \label{eq:dist_layout_req_i1_340} \\ 
    &\condsete{240} \nsubseteq \condsete{1} \cup \condsete{3} \Rightarrow \distprime{e_2, e_4} = \dz \geq \dm, \label{eq:dist_layout_req_i1_240} \\ 
    i=2:~ &\condsete{301} \nsubseteq \condsete{2} \cup \condsete{4} \Rightarrow \distprime{e_3, e_0} = \dm + 1 \geq \dz, \label{eq:dist_layout_req_i2_301} \\ 
    &\condsete{30} \nsubseteq \condsete{2} \cup \condsete{4} \cup \condsete{1} \Rightarrow \distprime{e_3, e_0, v_{12}} = 2(\dm + 1) \geq \dz, \label{eq:dist_layout_req_i2_30} \\ 
    i=3:~ &\condsete{012} \nsubseteq \condsete{3} \cup \condsete{4} \Rightarrow \distprime{v_{34}, e_0} = \dm + 1 \geq \dz, \label{eq:dist_layout_req_i3_012} \\
    &\condsete{402} \nsubseteq \condsete{3} \cup \condsete{1} \Rightarrow \distprime{e_0, e_2} = \dm + 1 \geq \dz, \label{eq:dist_layout_req_i3_402} \\
    &\condsete{02} \nsubseteq \condsete{3} \cup \condsete{4} \cup \condsete{1} \Rightarrow \distprime{v_{34}, e_0, e_2} = 2(\dm + 1) \geq \dz, \label{eq:dist_layout_req_i3_02} \\ 
    i=4:~ &\condsete{013} \nsubseteq \condsete{4} \cup \condsete{2} \Rightarrow \distprime{e_1, e_3} = \dz \geq \dm, \label{eq:dist_layout_req_i4_013} \\ 
    &\condsete{012} \nsubseteq \condsete{4} \cup \condsete{3} \Rightarrow \distprime{e_2, v_{34}} = \dz \geq \dm. \label{eq:dist_layout_req_i4_012}
\end{align}
\label{eq:dist_layout_requirements}
\end{subequations}

We can derive `(ii) $\lgztp{OCD} \in \qty{\Pstage{k}}_{k=1}^6$' in Condition~\ref{cond:dist_rotation_grouping} from Eq.~\eqref{eq:dist_layout_req_i0_12} as follows:
Since its consequent is false, its presequent should be false for every distillation stage.
From Eqs.~\eqref{eq:dist_operators_to_measure}, \eqref{eq:condset}, and~\eqref{eq:aux_qubit_condsets}, we obtain
\begin{align*}
    \condsete{12} \subseteq \condsete{0} \cup \condsete{34} &\equiv  \gz \notin \condsete{2}\setminus\condsete{0} \equiv \neg\qty(\lgc{Q}_\mr{out}^{(k)} \neq \lgc{Q}_\mr{AB}^{(k)} = \identity),
\end{align*}
where `$\neg$' is logical negation.
The last proposition always holds except the case where $\Qstage{k} = \lgztp{OCD}$.

Likewise, we can derive `(iii) if $\dz > 2\dm + 2$, $\rot{\qty(\lgz_\qty{s})}{\pi/8}$ and $\rot{\qty(\lgz_\qty{t})}{\pi/8}$ are not paired where $(s,t)$ is any of (OAC, OBC), (OAD, OBD), (OAC, OAD), (OBC, OBD), (OCD, OABCD), (OAB, OABCD)' in Condition~\ref{cond:dist_rotation_grouping} from Eqs.~\eqref{eq:dist_layout_req_i2_30} and~\eqref{eq:dist_layout_req_i3_02}:
If $\dz > 2\dm + 2$, the presequent of Eq.~\eqref{eq:dist_layout_req_i2_30} should be false.
Since
\begin{align*}
    \condsete{03} \subseteq \condsete{2} \cup \qty{\gx, \gz} \equiv \gy \notin \condsete{03}\setminus\condsete{2} \equiv \neg \qty(\Pstage{k}_\mr{out} = \Qstage{k}_\mr{out} \land \Pstage{k}_\mr{CD} = \Qstage{k}_\mr{CD} \land \Pstage{k}_\mr{AB} \neq \Qstage{k}_\mr{AB}),
\end{align*}
where `$\land$' is logical AND, we need to avoid measuring $\lgz_\qty{s}$ and $\lgz_\qty{t}$ simultaneously for each $(s, t) \in \qty{(\mr{OAC}, \mr{OBC}), (\mr{OAD}, \mr{OBD}), (\mr{OCD}, \mr{OABCD})}$.
Analogously, the negation of Eq.~\eqref{eq:dist_layout_req_i3_02} implies
\begin{align*}
    \neg \qty(\Pstage{k}_\mr{out} = \Qstage{k}_\mr{out} \land \Pstage{k}_\mr{AB} = \Qstage{k}_\mr{AB} \land \Pstage{k}_\mr{CD} \neq \Qstage{k}_\mr{CD}),
\end{align*}
thus $(s, t)=(\mr{OAC}, \mr{OAD}), (\mr{OBC}, \mr{OBD}), (\mr{OAB}, \mr{OABCD})$ are additionally prohibited.

All the other requirements in Eq.~\eqref{eq:dist_layout_requirements} except Eqs.~\eqref{eq:dist_layout_req_i0_12}, \eqref{eq:dist_layout_req_i2_30}, and~\eqref{eq:dist_layout_req_i3_02} are always fulfilled regardless of the way to pair the $\pi/8$-rotations.
We prove that each of these requirements has a false presequent as follows, where `Preseq.' stands for the presequent:
\begin{itemize}
    \item Eq.~\eqref{eq:dist_layout_req_i0_234}: $\neg\text{(Preseq.)} \equiv \gx \notin \condsete{23}\setminus\condsete{0} \equiv \neg\qty(\lgc{P}_\mr{out} \neq \lgc{P}_\mr{AB} = \lgc{P}_\mr{CD} = \identity)$, which holds for every $\lgc{P} \in \mathcal{Z}$.
    \item Eq.~\eqref{eq:dist_layout_req_i0_134}: $\condsete{134} = \emptyset$.
    \item Eq.~\eqref{eq:dist_layout_req_i0_123}: Counterpart of Eq.~\eqref{eq:dist_layout_req_i0_234} where $\gx$ and $\lgc{P}$ are respectively replaced with $\gz$ and $\lgc{Q}$.
    \item Eq.~\eqref{eq:dist_layout_req_i0_23}: $\neg\text{(Preseq.)} \equiv \gy \notin \condsete{23} \setminus \condsete{0} \equiv \neg\qty( \lgc{P}_\mr{ABCD} = \lgc{Q}_\mr{ABCD} \land \lgc{P}_\mr{out} \neq \lgc{Q}_\mr{out})$, which holds for every pair $(\lgc{P}, \lgc{Q}) \in \mathcal{Z}^{\times 2}$ that satisfies $\lgc{P} \neq \lgc{Q}$.
    \item Eq.~\eqref{eq:dist_layout_req_i0_14}: $\condsete{14} = \emptyset$.
    \item Eq.~\eqref{eq:dist_layout_req_i1_340}: $\neg\text{(Preseq.)} \equiv \gx \notin \condsete{03}\setminus\condsete{2} \equiv \neg\qty(\lgc{P}_\mr{AB} \neq \lgc{P}_\mr{out} = \lgc{P}_\mr{CD} = \identity)$, which holds for every $\lgc{P} \in \mathcal{Z}$.
    \item Eq.~\eqref{eq:dist_layout_req_i1_240}: Counterpart of Eq.~\eqref{eq:dist_layout_req_i1_340} where CD is replaced with AB.
    \item Eq.~\eqref{eq:dist_layout_req_i2_301}: Counterpart of Eq.~\eqref{eq:dist_layout_req_i1_340} where $\gx$ and $\lgc{P}$ are respectively replaced with $\gz$ and $\lgc{Q}$.
    \item Eq.~\eqref{eq:dist_layout_req_i3_012}: Counterpart of Eq.~\eqref{eq:dist_layout_req_i2_301} where CD is replaced with AB.
    \item Eq.~\eqref{eq:dist_layout_req_i3_402}: Same as Eq.~\eqref{eq:dist_layout_req_i1_240}.
    \item Eq.~\eqref{eq:dist_layout_req_i4_013}: Same as Eq.~\eqref{eq:dist_layout_req_i2_301}.
    \item Eq.~\eqref{eq:dist_layout_req_i4_012}: Same as Eq.~\eqref{eq:dist_layout_req_i3_012}.
\end{itemize}

We lastly argue why the requirement `(i) $\dout - \dz \leq \dm < 2\dz$' is needed as follows:
If $\dout - \dz > \dm$, we additionally have a nontrivial requirement for $i=0$:
\begin{align*}
    \condsete{24} \nsubseteq \condsete{0} \cup \condsete{1} \cup \condsete{3} \Rightarrow \distprime{v_{01}, e_2, e_4} = \dm + \dz + 1 \geq \dout.
\end{align*}
The negation of its presequent implicates $\neg\qty(\Pstage{k}_\mr{out},\Pstage{k}_\mr{CD} \neq \identity = \Pstage{k}_\mr{AB})$, which prohibits the case of $\Pstage{k} = \lgztp{OCD}$.
This contradicts Condition~\ref{cond:dist_rotation_grouping}.
If $\dm > 2\dz$, we have two additional nontrivial requirements (for $i=1,4$)
\begin{align*}
    \condsete{40} \nsubseteq \condsete{1} \cup \condsete{2} \cup \condsete{3} \Rightarrow \distprime{v_{01}, v_{23}, e_4} = 2\dz \geq \dm,  \\ 
    \condsete{01} \nsubseteq \condsete{4} \cup \condsete{2} \cup \condsete{3} \Rightarrow \distprime{e_1, v_{23}, v_{34}} = 2\dz \geq \dm.
\end{align*}
However, these two contradict each other.
They respectively give $\neg \qty(\Pstage{k}_\mr{AB}, \Pstage{k}_\mr{CD} \neq \identity = \Pstage{k}_\mr{out})$ and $\neg \qty(\Qstage{k}_\mr{AB}, \Qstage{k}_\mr{CD} \neq \identity = \Qstage{k}_\mr{out})$, thus operators such as $\lgztp{ABC}$ cannot be assigned to either $\Pstage{k}$ or $\Qstage{k}$.
\end{widetext}

\subsection{Justification of considering only clockwise CPGS operators \label{subapp:justification_clockwise_CPGS}}

We now justify that, in the above discussion, it is sufficient to consider only CPGS operators with clockwise orders by proving the following claim:
\begin{claim}
    For a given CPGS operator $R_\mr{cyc}$ in Eq.~\eqref{eq:Rcyc_def} defined by $\wbs \in \qty{\xbs, \ybs, \zbs}$ and $\qty{o_j, \tilde{o}_j}_{j=0}^{L-1}$ where $o_0$ and $\tilde{o}_0$ are neighboring vertices, there exists a CPGS operator $R'_\mr{cyc}$ with endpoints $\qty{o'_j, \tilde{o}'_j}_{j=0}^{L-1}$ that is arranged clockwise such that $\qty{o'_0, \tilde{o}'_0} = \qty{o_0, \tilde{o}_0}$ and $\wt{R'_\mr{cyc}} \leq \wt{R_\mr{cyc}}$.
    \label{claim:justification_clockwise_only}
\end{claim}

\begin{proof}
Let us consider a CPGS operator $R_\mr{cyc}$ in the claim.
$R_\mr{cyc}$ can be represented by an undirected graph
\begin{align*}
    G = \qty(V \cup E, \qty{\qty{o_j, \tilde{o}_j}}_j),
\end{align*}
referred to as the \emph{schematic graph} of $R_\mr{cyc}$.
To avoid confusion, we refer to vertices and edges of $G$ here as `nodes' and `links', respectively, and `vertices' and `edges' only mean the elements of $V$ and $E$, respectively.
$G$ can be embedded in the pentagon in a way that each node is placed on the corresponding vertex or the center of the corresponding edge.
Each link of $G$ indicates that its endpoints are connected by a \gw-string operator in $R_\mr{cyc}$.
Note that $G$ cannot have self-cycles (which correspond to trivial string operators) but may have parallel links.

Let us first assume that all the five edges do not condense \gw.
In this case, all the edges are isolated nodes in $G$ and, since no two vertices are topologically connected with respect to \gw, $\tilde{o}_j = o_{j+1}$ for each $j$.
Thus, $G$ is composed of a circuit on (at most five) vertices and some isolated nodes.
The graph $G$ embedded in the pentagon may have intersecting links, which correspond to intersecting \gw-string operators.
We can always redraw them not to intersect while keeping the supports of the string operators \cite{kesselring2022anyon}.
Note that the link in $G$ between $o_0$ and $\tilde{o}_0$ (which are adjacent vertices) cannot intersect with another link, thus it still remains in the redrawn graph.
If the redrawn graph has two or more cycles (i.e., simple circuits) including parallel links, we remove all the cycles except the one that contains a link between $o_0$ and $\tilde{o}_0$.
This process leads to a new schematic graph $G'$ and the corresponding CPGS operator $R'_\mr{cyc}$ that is not longer than $R_\mr{cyc}$.
Since $R'_\mr{cyc}$ does not have intersecting string operators, its endpoints can be relabeled to be arranged clockwise.

Let us now suppose that some edges condense \gw.
We define the \emph{reduced} schematic graph $\widetilde{G}$ where every set of nodes of $\tilde{G}$ that are topologically connected with respect to \gw is merged into a single node.
Thus, for each $j$, $o_j$ and $\tilde{o}_j$ correspond to different nodes in $\widetilde{G}$, while $\tilde{o}_j$ and $\tilde{o}_{j+1}$ correspond to the same node.
Therefore, $\widetilde{G}$ contains a circuit on at most $[5 - (\text{number of edges condensing \gw})]$ nodes and all the other nodes are isolated.
By replacing intersecting links and leaving only one cycle from $\widetilde{G}$ as done in the previous case where no edges condense \gw, we can find a CPGS operator (that is not longer than $R_\mr{cyc}$) that has a reduced schematic graph $\widetilde{G}'$ with only one cycle and no intersecting links.
Note that $\widetilde{G}'$ still contains a link between $o_0$ and $\tilde{o}_0$ because, if it intersects with another link (having an endpoint on the only edge between $o_0$ and $\tilde{o}_0$) in $\widetilde{G}$, it means that $o_0$ and $\tilde{o}_0$ are topologically connected with respect to \gw. 
By relabeling nodes appropriately (while keeping $\qty{o'_0, \tilde{o}'_0} = \qty{o_0, \tilde{o}_0}$), we can ensure $o'_j \prec \tilde{o}'_j \prec \tilde{o}'_{j+1}$ and $o'_j \prec o'_{j+1} \prec \tilde{o}'_{j+1}$, but the order between $\tilde{o}'_j$ and $o'_{j+1}$ remains ambiguous unlike the previous case as they correspond to the same node in $\widetilde{G}'$.
If they have a wrong order ($o'_j \prec o'_{j+1} \prec \tilde{o}'_j \prec \tilde{o}'_{j+1}$), the links $(o'_j, \tilde{o}'_j)$ and $(o'_{j+1}, \tilde{o}'_{j+1})$ intersect, thus we can make them non-intersect by the same technique as above, which leads to the links $(o'_j, \tilde{o}'_{j+1})$ and $(\tilde{o}'_j, o'_{j+1})$.
We then remove the link $(\tilde{o}'_j, o'_{j+1})$ that corresponds to a trivial string operator, which strictly shorten the CPGS operator.
By repeating the above process until the schematic graph does not have intersecting links, we finally obtain the desired CPGS operator $R'_\mr{cyc}$.
\end{proof}

\section{Simulations of the growing operation \label{app:growing_operation_simulations}}

In this appendix, we elaborate on the simulation method for the growing operation described in Sec.~\ref{subsec:growing_operation}, where errors are decoded using the concatenated MWPM decoder with post-selection. 
We also present detailed numerical results that are not shown in the main text.

\subsection{Methods}

We simulate a circuit implementing the growing operation shown in Fig.~\ref{fig:growing_operation}(a), followed by $\dm$ rounds of syndrome extraction. 
Each Bell state is prepared by initializing two qubits to $\ket{+}$ and $\ket{0}$, respectively, and applying a \cnot gate between them. 
Syndrome extraction is performed using the circuit in Fig.~\ref{fig:syndrome_extraction_circuit}; see Appendix~\ref{app:entangling_gate_schedule} for further details.

To calculate the $\lgz$ ($\lgx$) failure rate, we consider preparing the $\ketlgc{0}$ ($\ketlgc{+}$) state in the initial patch with distance $\dcult$ and measuring the final patch with distance $\dm$ in the $\lgz$ ($\lgx$) basis.
These logical preparation and measurement are hypothetical components required only to define a logical observable of the circuit (deterministic when there are no errors) for using \textit{Stim}. 
They are not presented in actual implementations, where the process is preceded by cultivation and followed by idling or lattice surgery.
Since we aim to quantify the effect of errors during the growing operation in isolation from errors outside of it, the logical preparation and measurement are assumed to be perfect.

To calculate the logical gap using the concatenated MWPM decoder, we slightly modify the decoder implemented in Ref.~\cite{lee2025color} by treating the observable as a detector and running the decoder twice with different detector values.
The logical gap is then determined from the difference in log-likelihood weights between these two predictions.
A decoding outcome is accepted only when the logical gap is greater than a predetermined threshold $c_\mr{gap}$.
This method is implemented in the latest version of the \textit{color-code-stim} module \cite{github:colorcodestim}.

We compute the logical error rate as $p_\mr{log} = p_\mr{fail}^Z + p_\mr{fail}^X$, where $p_\mr{fail}^P$ is the obtained $\overline{P}$ failure rate.
Note that this is a conservative estimation of the actual logical error rate, as $\lgy$ errors are overcounted.
The acceptance rate is computed as $p_\mr{acc} = p_\mr{acc}^Z p_\mr{acc}^X$, where $p_\mr{acc}^Z$ and $P_\mr{acc}^X$ are respectively the acceptance rates of the two settings, assuming that decodings for $\lgz$ and $\lgx$ failures are aborted independently.

\subsection{Results}

\begin{figure*}[!t]
	\centering
	\includegraphics[width=0.48181818\textwidth]{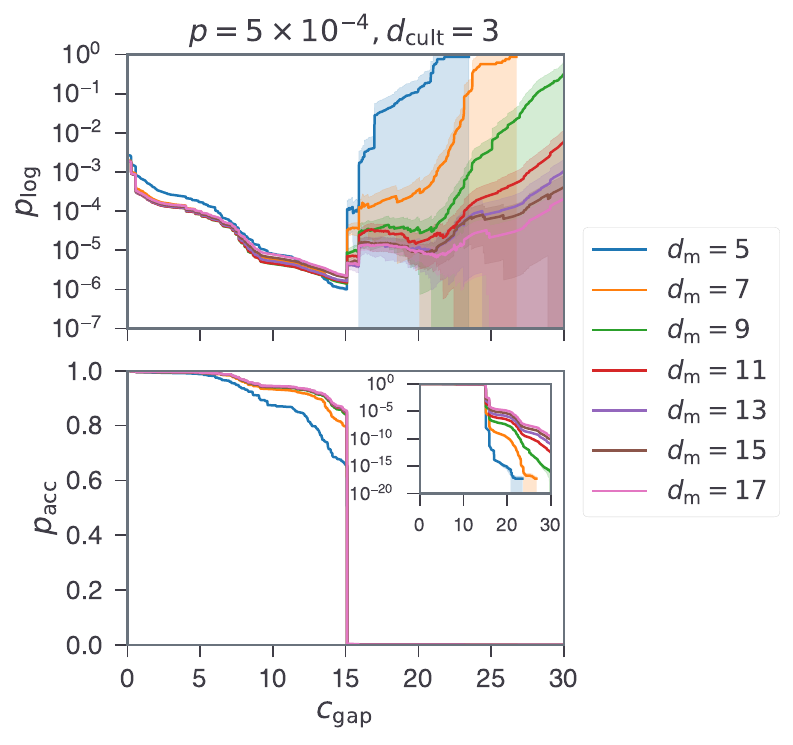}
    \includegraphics[width=0.48181818\textwidth]{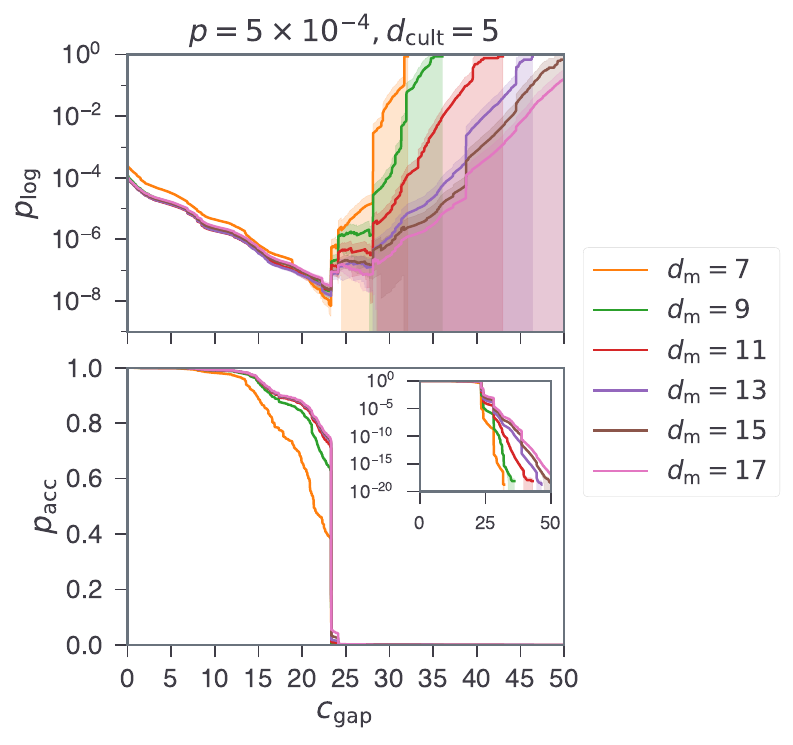}
    \includegraphics[width=0.48181818\textwidth]{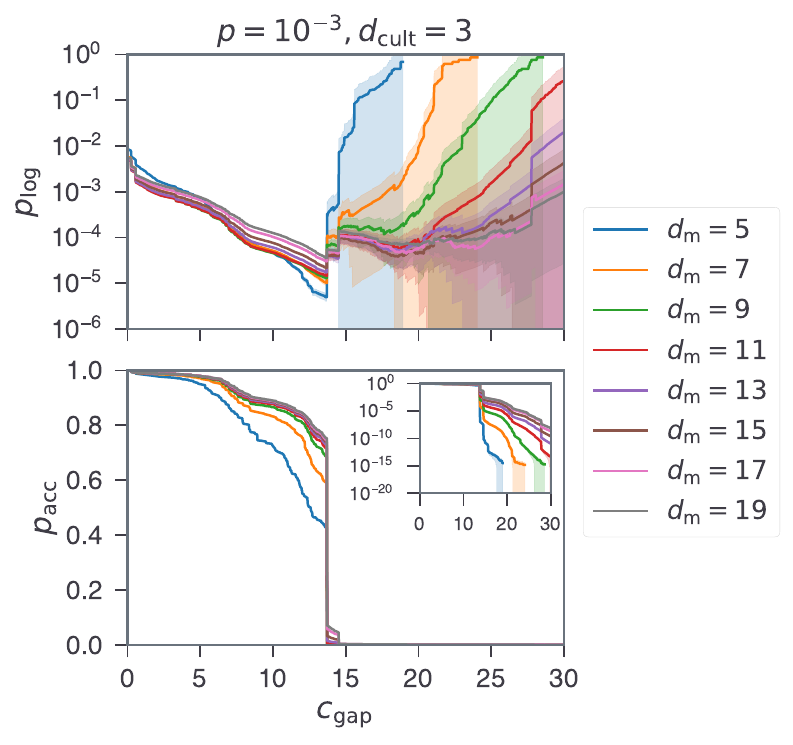}
    \includegraphics[width=0.48181818\textwidth]{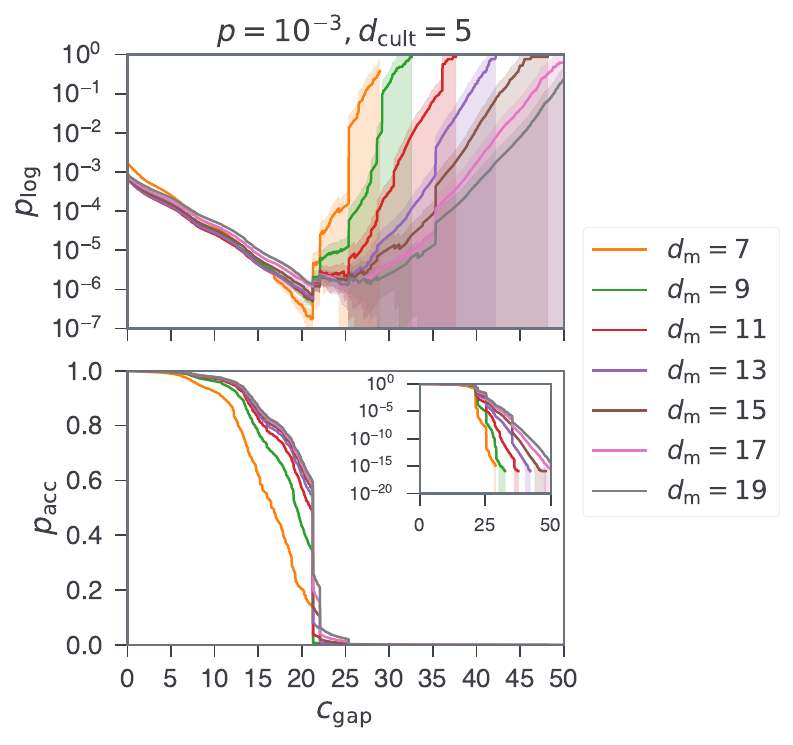}
	\caption{
        Decoder simulation results for the growing operation with post-selection.
        The logical error rate $p_\mr{log}$ and the acceptance rate $p_\mr{acc}$ are plotted against the logical gap threshold $c_\mr{gap}$ for physical noise strength $p \in \qty{5 \times 10^{-4}, 10^{-3}}$ and pre-growing code distance $\dcult \in \qty{3, 5}$, with varying post-growing code distance $\dm > \dcult$.
        Insets in the acceptance rate plots display their logarithmic-scale versions.
        The shaded regions represents 99\% confidence intervals.
    }
	\label{fig:growing_probs_vs_gap}
\end{figure*}

Now we present the results of our numerical simulations on the growing operation with varying $p$, $\dcult$, $\dm$, and $c_\mr{gap}$.
In Fig.~\ref{fig:growing_probs_vs_gap}, we plot $p_\mr{log}$ and $p_\mr{acc}$ against $c_\mr{gap}$ for $p \in \qty{5 \times 10^{-4}, 10^{-3}}$, $\dcult \in \qty{3, 5}$, and varying $\dm > \dcult$.
The achievable minimum logical error rate is 
\begin{align*}
\begin{cases}
    1.0 \times 10^{-6} & \text{for } p = 5 \times 10^{-4},~\dcult = 3, \\
    6.9 \times 10^{-9} & \text{for } p = 5 \times 10^{-4},~\dcult = 5, \\
    5.0 \times 10^{-6} & \text{for } p = 10^{-3},~\dcult = 3, \\
    1.7 \times 10^{-7} & \text{for } p = 10^{-3},~\dcult = 5.
\end{cases}
\end{align*}
Note that the logical error and acceptance rates show sudden jumps at
\begin{align}
c_\mr{gap} &\approx
\begin{cases}
    15.0 & \text{for } p = 5 \times 10^{-4},~\dcult = 3, \\
    23.3 & \text{for } p = 5 \times 10^{-4},~\dcult = 5, \\
    13.7 & \text{for } p = 10^{-3},~\dcult = 3, \\
    21.0 & \text{for } p = 10^{-3},~\dcult = 5. \\
\end{cases}
\label{eq:cgap_sudden_jumps}
\end{align}
We conjecture that this phenomenon arises because the concatenated MWPM decoder may fail to find a minimum-weight solution. 
That is, the jump locations would correspond to the maximum possible logical gap values under the hypergraph MWPM decoder, which can always identify a minimum-weight solution.
However, since the concatenated MWPM decoder may produce a higher-weight correction, the logical gap can occasionally exceed these maximum values with low probability.
We leave the verification of this conjecture for future work.

\begin{figure*}[!t]
	\centering
	\includegraphics[width=0.48181818\textwidth]{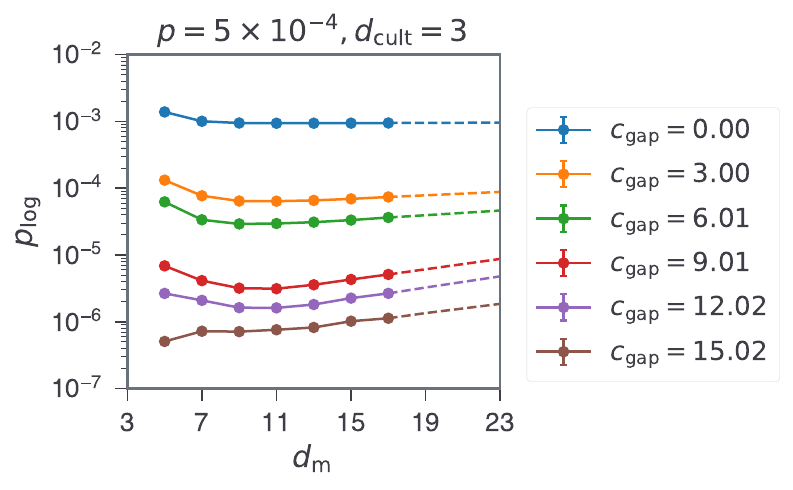}
    \includegraphics[width=0.48181818\textwidth]{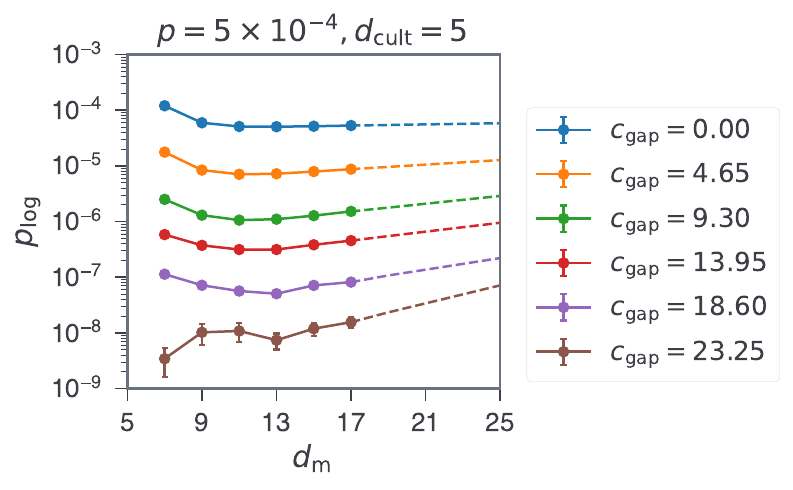}
    \includegraphics[width=0.48181818\textwidth]{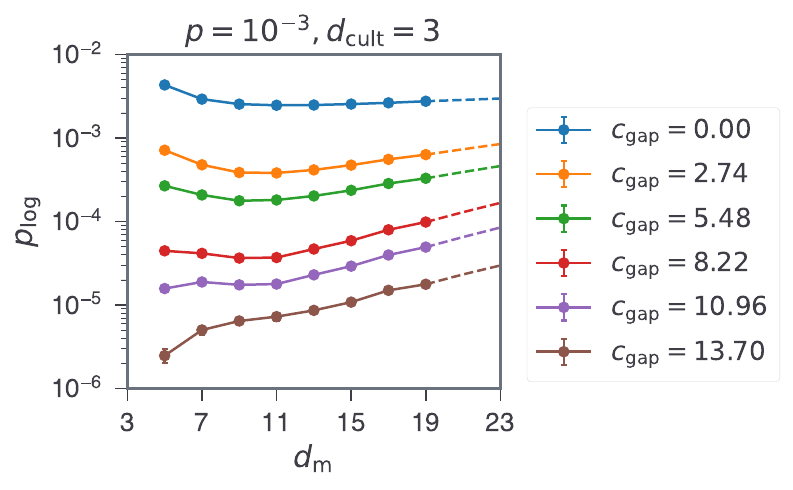}
    \includegraphics[width=0.48181818\textwidth]{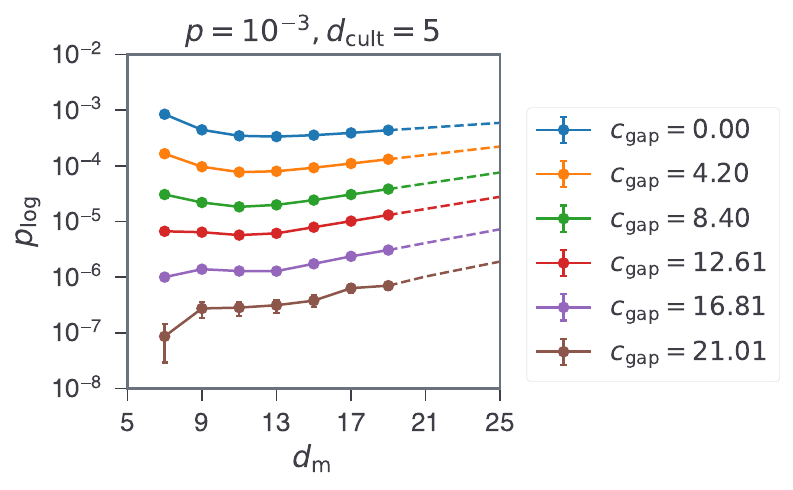}
	\caption{
        Dependency of the logical error rate $p_\mr{log}$ on $\dm$ for the growing operation with post-selection, at $p \in \qty{5 \times 10^{-4}, 10^{-3}}$ and $\dcult \in \qty{3, 5}$, with varying $c_\mr{gap}$. 
        The values of $c_\mr{gap}$ are chosen evenly spaced between zero and the points of sudden jumps given in Eq.~\eqref{eq:cgap_sudden_jumps}.
        The error bars indicate 99\% confidence intervals, while data points without error bars have confidence intervals smaller than the marker size.
        Extrapolations, shown as dashed lines, are obtained via linear fitting (on a logarithmic scale) using the last three data points.
    }
	\label{fig:growing_plog_vs_dm}
\end{figure*}

Additionally, in Fig.~\ref{fig:growing_plog_vs_dm}, we illustrate the dependence of $p_\mr{log}$ on $\dm$ with varying $c_\mr{gap}$, which is not clearly visible in Fig.~\ref{fig:growing_probs_vs_gap}.
Due to computational constraints, our simulations are limited to $\dm \leq 17$ when $p = 5 \times 10^{-4}$ and $\dm \leq 19$ when $p = 10^{-3}$. 
Beyond these limits, we extrapolate using linear regression (on a logarithmic scale) based on the last three data points.
Note that we do not extrapolate $p_\mr{acc}$; instead, we use its values at $\dm = 17$ or 19 for larger $\dm$'s.
This approach does not overestimate performance, as $p_\mr{acc}$ increases with $\dm$, as shown in Fig.~\ref{fig:growing_probs_vs_gap}.

\section{Method for numerical analysis \label{app:numerical_analysis_method}}

In this appendix, we elaborate on the method for numerical analysis, which is outlined in Sec.~\ref{subsec:analysis_method}

\subsection{Logical error rates of patches \label{subsec:patch_logical_error_rates}}

We assume that the logical error rate of multiple rounds of syndrome extraction is proportional to the number of rounds $T$.
For a given noise strength $p$, the per-round $\lgx$/$\lgy$/$\lgz$ error rate of a triangular logical patch with distance $d$ is denoted as $\plogtri{X/Y/Z}(p, d)$.
For a rectangular logical patch with code distances $\dx$ and $\dz$ in Fig.~\ref{fig:color_code_patches}(d), its per-round $\lgx_1$ error rate is denoted as $\plogrec{X_1}(p, \dx, \dz)$ and functions $\plogrec{Z_1}$, $\plogrec{X_2}$, and $\plogrec{Z_2}$ are defined analogously.
These errors are caused by Pauli-$X$ or $Z$ error strings connecting opposite boundaries.
Pauli-$Y$ error strings connecting opposite boundaries incur correlations of $\lgx$ or $\lgz$ errors between the two logical qubits, with error rates denoted as $\plogrec{X_1 X_2}$ and $\plogrec{Z_1 Z_2}$.
In addition, `diagonal' blue string operators connecting the top left corner and the bottom right corner in Fig.~\ref{fig:color_code_patches}(d) have weights of at least $\max(\dx, \dz)$ and make correlated logical errors $\lgx_1 \lgz_2$ (for \bx), $\lgz_1 \lgx_2$ (for \bz), or $\lgy_1 \lgy_2$ (for \by).
We ignore such correlations in our error analysis, which can be justified because (i) diagonal error strings are significantly rarer than other error strings for sufficiently large code distances \footnote{
    Assuming $\dx \geq \dz$, there exist $N_\mr{diag} \coloneqq \binom{l}{r}$ shortest paths for diagonal error strings, where $l \coloneqq \dx/2 - 1$ and $r \coloneqq (\dx - \dz) / 2$.
    For vertical error strings (causing $\lgx_1$ or $\lgx_2$ errors), there are approximately $N_\mr{v} \coloneqq \dz \cdot 2^{\dx/2-1}$ shortest paths.
    Comparing these two values, $N_\mr{diag} / N_\mr{v} < 0.1$ for $\dx \geq 8$ or $\dz \geq 6$, and $N_\mr{diag} / N_\mr{v} < 0.05$ for $\dx \geq 14$ or $\dz \geq 8$.
    In general, we obtain $N_\mr{diag}/N_\mr{v} < 2/[\dz \sqrt{\pi(\dx - 2)}]$ from the inequality $\binom{l}{r} \leq \binom{l}{l/2} < \sqrt{2/(\pi l)} 2^l$.
    Hence, although diagonal and vertical error strings have the same minimum weight, they differ significantly in their frequency of occurrence, assuming that circuit-level noise acts at similar levels in both directions.
} and (ii) the resulting correlated logical errors are not specifically more detrimental than other logical errors (i.e., they are detectable by the final $\lgx$ measurements of the MSD circuit alike uncorrelated $\lgz$ errors).
In addition, we denote the logical error rates of the growing operation and the subsequent $\dm$ rounds of syndrome extraction (which are jointly decoded with post-selection) for the cultivation-MSD scheme as $\pgrow{X/Y/Z}(p, \dcult, \dm, c_\mr{gap})$, where $c_\mr{gap}$ is the logical gap threshold for post-selection (see Sec.~\ref{subsec:growing_operation} and Appendix~\ref{app:growing_operation_simulations}).

Through numerical simulations, we only can compute logical failure rates of observables, not individual logical error rates.
(We say, e.g., $\lgx$ fails if a $\lgy$ or $\lgz$ error occurs.)
To extract individual logical error rates from the logical failure rates of a triangular patch, we assume that $\lgx$ and $\lgz$ errors occur with equal probability \cite{lee2025color}, while the probability of a $\lgy$ error is $\pyratio$ times this probability, where $\pyratio$ is a given non-negative number.
Namely, we assume
\begin{align*}
    \plogtri{X} = \plogtri{Z} = \frac{\plogtri{\mr{fail}}}{2(1 + \pyratio)}, \qquad \plogtri{Y} = \pyratio \plogtri{X},
\end{align*}
where $\plogtri{\mr{fail}}$ is the summation of the $\lgx$ and $\lgz$ failure rates.
Likewise, we suppose that
\begin{align*}
    \plogrec{X_1} = \plogrec{X_2} &= \frac{\plogrec{Z.\mr{fail}}}{2(1 + \pyratio)}, & \plogrec{X_1 X_2} &= \pyratio \plogrec{X_1}, \\
    \plogrec{Z_1} = \plogrec{Z_2} &= \frac{\plogrec{X.\mr{fail}}}{2(1 + \pyratio)}, & \plogrec{Z_1 Z_2} &= \pyratio \plogrec{Z_1}, \\
    \pgrow{X} = \pgrow{Z} &= \frac{\pgrow{\mr{fail}}}{2(1 + \pyratio)}, & \pgrow{Y} &= \pyratio \pgrow{X},
\end{align*}
where $\plogrec{Z(X).\mr{fail}}$ is the summation of the $\lgz_1$ ($\lgx_1$) and $\lgz_2$ ($\lgx_2$) failure rates of a rectangular patch, and $\pgrow{\mr{fail}}$ is the summation of the $\lgx$ and $\lgz$ failure rates of the growing operation.

\begin{figure}[!t]
	\centering
	\includegraphics[width=\linewidth]{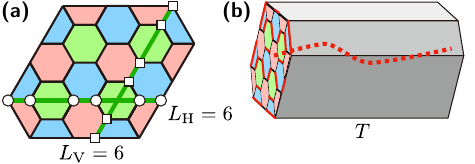}
	\caption{
        \subfig{a} Patch for analyzing timelike error strings through a stability experiment, which encodes no logical qubits. The parameters $L_\mr{H}$ and $L_\mr{V}$ determine the size of the patch, and $L_\mr{H} = L_\mr{V}=6$ in this example.
        \subfig{b} Spacetime picture of the stability experiment, which is composed of $T$ rounds of syndrome extraction surrounded by two red temporal boundaries.
        The red dotted line represents an example of a timelike error string.
 }
	\label{fig:stability_experiment}
\end{figure}

We further assume that the output state of cultivation has probabilistic Pauli noise and define $\pcult{X/Y/Z}(p, \dm)$ as its $\lgx$/$\lgy$/$\lgz$ error rate, satisfying
\begin{align*}
    \pcult{X} = \pcult{Z} &= \frac{q_\mr{cult}}{2 + \pyratio}, \\
    \pcult{Y} &= \pyratio \pcult{X},
\end{align*}
where $q_\mr{cult}$ is the output infidelity of cultivation.

In addition, we need to deal with timelike error strings that connect two red temporal boundaries of the ancillary region, such as $T$ time consecutive $Z$-type check measurement errors on a single red face.
For this, we consider a stability experiment \cite{gidney2022stability} using a patch in the shape of Fig.~\ref{fig:stability_experiment}(a), which is characterized by parameters $L_\mr{H}$ and $L_\mr{V}$ (i.e., the weights of the shortest green string operators connecting the pairs of opposite parallel boundaries horizontally and vertically).
Note that the patch does not encode logical qubits as it only has green boundaries.
In the experiment, the patch undergoes $T$ rounds of syndrome extraction of the patch surrounded by two red temporal boundaries, as shown in Fig.~\ref{fig:stability_experiment}(b).
After decoding, we check whether a timelike error occurs from the product of the outcomes of the $Z$- or $X$-type checks on all the red and blue faces at any round, which should be $+1$ if there are no errors.
We denote the timelike error rates as $\ptimelike{Z}(p, L_\mr{H}, L_\mr{V}, T)$ and $\ptimelike{X}(p, L_\mr{H}, L_\mr{V}, T)$, which respectively correspond to corrupted $Z$- and $X$-type checks.
In our analysis, we only consider the cases of $L_\mr{H} = L_\mr{V} = T$ for simulations and assume that 
\begin{align*}
    \ptimelike{Z/X}(p, L_\mr{H}, L_\mr{V}, T) = \frac{L_\mr{H} L_\mr{V}}{T^2} \ptimelike{Z/X}(p, T, T, T).
\end{align*}


\subsection{Computation of the output infidelity and success probability of MSD}

Each type of logical Pauli error that occurs during MSD can be mapped to a noise channel acted on the output and validation qubits between the last $\pi/8$-rotation and the final $\lgx$ measurements of the validation qubits.
The noise channel is of the form
\begin{align*}
    \Lambda_{\lgc{U},p_\mr{err}}: \; \lgc{\rho} \mapsto (1 - p_\mr{err})\lgc{\rho} + p_\mr{err} \lgc{U} \lgc{\rho} \lgc{U}^\dagger,
\end{align*}
where $\lgc{\rho}$ is the logical state of the output and validation qubits, $p_\mr{err}$ is the logical error rate, and $\lgc{U}$ is a product of $\pi/2$- or $(\pm \pi/4)$-rotations in bases consisting of $\lgz$'s only.
In Appendix~\ref{app:determining_noise_channels}, we present rules to determine noise channels from various possible error sources in the logical patches and the ancillary region (including timelike errors), covering both the single-level and cultivation-MSD schemes.
Following these rules, each error rate ($p_\mr{err}$) of the noise channel is expressed in terms of the logical error rate functions (defined in Appendix~\ref{subsec:patch_logical_error_rates}) and additional variables including $T_\mr{intv}$ and $T_\mr{idle}$ for the cultivation-MSD scheme.

Denoting the set of noise channels obtained by the rules as $\{\Lambda_{\lgc{U}_i, p_i}\}_{i=1}^{m}$, the unnormalized output state after the final measurements is
\begin{align*}
    \lgc{\rho}_\mr{out} = (\identity_\mr{out} \otimes \bra{\lgc{+}\lgc{+}\lgc{+}\lgc{+}}_\mr{ABCD}) \Lambda_\mr{noise} \qty(\ketbra{\lgc{\psi}_\mr{init}}),
\end{align*}
where
\begin{align*}
    \Lambda_\mr{noise} &\coloneqq \Lambda_{\lgc{U}_1,p_1} \circ \cdots \circ \Lambda_{\lgc{U}_m,p_m}, \\
    \ket{\lgc{\psi}_\mr{init}} &\coloneqq \ket{\lgc{A}_-}_\mr{out} \otimes \ket{\lgc{+}\lgc{+}\lgc{+}\lgc{+}}_\mr{ABCD}, \\
    \ket{\lgc{A}_-} &\coloneqq \frac{1}{\sqrt{2}} \qty(\ketlgc{0} + e^{-i\pi/4}\ketlgc{1}).
\end{align*}
The success probability $q_\mr{succ}$ of the scheme and the output infidelity $\infidMSD$ are then respectively given as
\begin{align*}
    q_\mr{succ} = \Tr\qty(\lgc{\rho}_\mr{out}), \quad
    \infidMSD = 1 - \frac{1}{q_\mr{succ}} \bra{\lgc{A}_-} \lgc{\rho}_\mr{out} \ket{\lgc{A}_-}
\end{align*}
In Appendix~\ref{app:MSD_performance_expressions}, we present analytic expressions of $q_\mr{succ}$ and $\infidMSD$ as functions of the physical error rate $p$, the code distances, and the logical error rates of several patches.

To calculate $q_\mr{succ}$ and $\infidMSD$ numerically, we need to explicitly specify the logical error rate functions, which vary depending on the decoder used to predict errors from syndrome outcomes.
We employ the concatenated MWPM decoder, proposed in Ref.~\cite{lee2025color}, which is suitable for our analysis since it is predicted to outperform other up-to-date matching-based decoders in terms of its logical failure rate under circuit-level noise, especially when $p$ is sufficiently small ($\leq 10^{-3}$).
Note that we use the same decoder for post-selection during the growing operation as well, thus it remains consistent.
In Appendix~\ref{app:decoder_simulations}, we elaborate on our decoder simulations for memory and stability experiments, presenting the methods and the obtained results.

\section{Selection of the entangling gate schedule \label{app:entangling_gate_schedule}}

\begin{figure*}[!t]
	\centering
	\includegraphics[width=\linewidth]{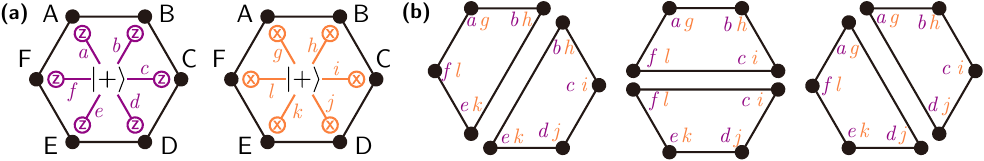}
	\caption{
    \subfig{a} Syndrome extraction circuits for measuring $Z$-type (left) and $X$-type (right) checks on a hexagonal face.
    In each circuit, the syndrome qubit at the center of the face is prepared to $\ket{+}$, undergoes multiple controlled-$Z$ (\cz) or controlled-$X$ (\cnot) gates with data qubits, and is finally measured in the $X$ basis.
    Twelve variables $a, b, \cdots, l$ are positive integers specifying the time steps that the corresponding \cnot gates are applied.
    \subfig{b} Six types of weight-4 faces that can be placed along boundaries with the corresponding time-step variables of the \cnot gates, colored in purple (orange) for the measurements of $Z$-type ($X$-type) checks.
    The syndrome extraction circuits are in the same form as \subfig{a}.
    }
	\label{fig:syndrome_extraction_circuit}
\end{figure*}

One important factor to consider when running circuit-level error simulations is the scheduling of entangling gates in the syndrome extraction circuit.
Since the numerical results in the original concatenated MWPM decoder paper~\cite{lee2025color} are based on the schedule verified and optimized only for triangular patches, we need to adjust it to work well also with rectangular patches.
More specifically, we identify 24 valid schedules with seven time steps that can properly extract check eigenvalues and, for each of them, simulate 6 rounds of syndrome extraction of a rectangular logical patch with $\dx=\dz=6$, obtaining the failure rates of four logical Pauli operators $\lgz_1$, $\lgz_2$, $\lgx_1$, and $\lgx_2$ by using the concatenated MWPM decoder.
(Here, we say that a logical Pauli operator fails if a logical error anticommuting with it occurs.)
We then select the schedule that minimizes the worst-case failure rate (i.e., largest one among the four).

To elaborate on this process, it is as follows:
Checks on a hexagonal face are measured by using the circuits in Fig.~\ref{fig:syndrome_extraction_circuit}(a), where each syndrome qubit for a $Z$-type ($X$-type) check is prepared to $\ket{+}$, undergo \cz (\cnot) gates with data qubits, and then are measured in the $X$ basis.
Six possible types of weight-4 faces that can be placed along boundaries are presented in Fig.~\ref{fig:syndrome_extraction_circuit}(b) and the corresponding checks are measured in a similar way to be consistent.
Note that the weight-4 face at the top left corner of Fig.~\ref{fig:color_code_patches}(c) is regarded to have the same form as the second face (with $b, c, d, e, h, i, j, k$) of Fig.~\ref{fig:syndrome_extraction_circuit}(b).
The entangling gate schedule is specified by 12 positive integers $\mathcal{A} = \qty[a, b, c, d, e, f; g, h, i, j, k, l]$, which contains all the integers from 1 to $\max\mathcal{A}$ (called the length of the schedule).
Each integer represents the time step at which the corresponding entangling gate in Fig.~\ref{fig:syndrome_extraction_circuit}(a) is applied.
The first (last) half of the schedule is referred to as its $Z$-type ($X$-type) part.

The following conditions need to be satisfied for the schedule to be valid \cite{beverland2021cost}:
\begin{enumerate}
    \item Each qubit can be involved in at most one \cnot gate at each time step, meaning that $(a,b,c,d,e,f)$, $(g,h,i,j,k,l)$, $(a,c,e,g,i,k)$, and $(b,d,f,h,j,l)$ are all tuples of distinct numbers.
    \item $X$- and $Z$-type check measurements should not interfere with each other. For example, the stabilizer $X$ of a $Z$-type syndrome qubit is propagated to $Z$ operators on the surrounding data qubits. We should prevent a situation that an odd number among these operators are again propagated to an $X$-type syndrome qubit. Therefore, all the following 13 numbers should be positive: 
    \begin{align}
    \begin{cases}
        (a-g)(b-h)(c-i)(d-j)(e-k)(f-l), \\
        (a-g)(b-h)(f-l)(e-k), \\
        (a-g)(b-h)(c-i)(f-l), \\
        (a-g)(b-h)(c-i)(d-j), \\
        (c-i)(d-j)(e-k)(f-l), \\
        (a-g)(d-j)(e-k)(f-l), \\
        (b-h)(c-i)(d-j)(e-k), \\
        (a-k)(b-j), \quad (b-l)(c-k), \quad (c-g)(d-l), \\
        (d-h)(e-g), \quad (e-i)(f-h), \quad (f-j)(a-i).
    \end{cases} \label{eq:13_positive_numbers}
    \end{align}
\end{enumerate}

There are $4 \times 3 \times 2 = 24$ length-7 schedules satisfying the above conditions, which are
\begin{align*}
    &[2, 1, 4, 5, 6, 3; 3, 2, 5, 6, 7, 4], \quad [2, 3, 6, 5, 4, 1; 3, 4, 7, 6, 5, 2],\\
    &[1, 2, 3, 6, 5, 4; 2, 3, 4, 7, 6, 5], \quad[1, 4, 5, 6, 3, 2; 2, 5, 6, 7, 4, 3]
\end{align*}
and their variations considering $\ang{120}$ rotation (i.e., $\mathcal{A} \rightarrow [c, d, e, f, a, b; i, j, k, l, g, h]$) and the exchange of the $X$- and $Z$-type parts (i.e., $\mathcal{A} \rightarrow [g, h, i, j, k, l; a, b, c, d, e, f]$).
To select one of these 24 schedules, we consider 6 rounds of syndrome extraction of a rectangular logical patch with distances $\dx = \dz = 6$ and compute the failure rates of four logical Pauli operators $\lgz_1$, $\lgz_2$, $\lgx_1$, and $\lgx_2$ defined in Fig.~\ref{fig:color_code_patches}(d) via Monte-Carlo simulations.
(Here, we say $\lgz_1$ fails if an $\lgx_1$ or $\lgy_1$ error occurs, and similarly for the other operators.)
We then select the schedule that minimizes the largest one among these four failure rates.
The computed logical failure rates are presented in Fig.~\ref{fig:cnot_schedule_comparison} for six schedules, which are respectively
\begin{align}
\begin{cases}
    1. &[3, 6, 5, 4, 1, 2; 4, 7, 6, 5, 2, 3], \\
    2. &[4, 5, 6, 3, 2, 1; 5, 6, 7, 4, 3, 2], \\
    3. &[5, 4, 1, 2, 3, 6; 6, 5, 2, 3, 4, 7], \\
    4. &[2, 1, 4, 5, 6, 3; 3, 2, 5, 6, 7, 4], \\
    5. &[1, 4, 5, 6, 3, 2; 2, 5, 6, 7, 4, 3], \\
    6. &[1, 2, 3, 6, 5, 4; 2, 3, 4, 7, 6, 5].
\end{cases}
\label{eq:valid_cnot_schedules}
\end{align}
Here we only consider six out of the 24 schedules since the rectangular patch is symmetric under a $\ang{180}$ rotation (i.e., $\mathcal{A} \rightarrow [d,e,f,a,b,c; j,k,l,g,h,i]$), and exchanging the $X$- and $Z$-type parts simply swaps their corresponding failure rates ($\lgx_1 \leftrightarrow \lgx_2$ and $\lgz_1 \leftrightarrow \lgz_2$).
Hence, we select the first one in Eq.~\eqref{eq:valid_cnot_schedules} as the default entangling gate schedule for our scheme, shown in Fig.~\ref{fig:selected_schedule}(a) and~(b).

\begin{figure}[!t]
	\centering
	\includegraphics[width=\linewidth]{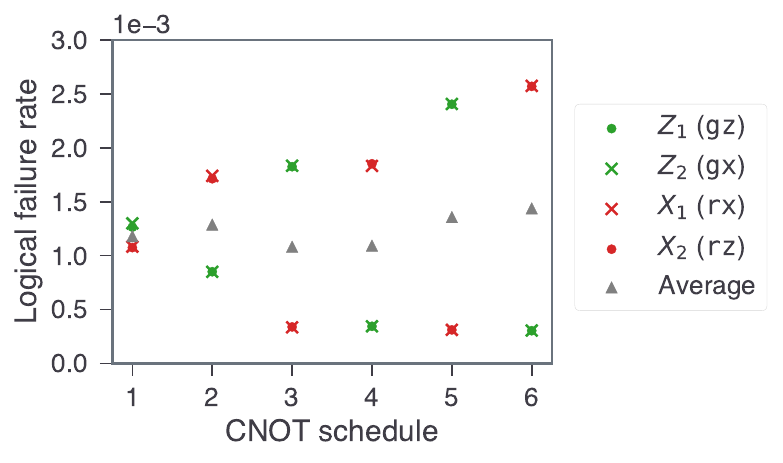}
	\caption{
        Failure rates of the four logical Pauli operators $\lgz_1$, $\lgz_2$, $\lgx_1$, and $\lgx_2$ for six different entangling gate schedules in Eq.~\eqref{eq:valid_cnot_schedules}, considering 6 rounds of syndrome extraction of a rectangular logical patch with distances $\dx = \dz = 6$ under the circuit-level noise model of strength $p = 10^{-3}$, decoded via the concatenated MWPM decoder \cite{lee2025color}.
        Note that $\lgz_1$, $\lgz_2$, $\lgx_1$, and $\lgx_2$ are respectively \gz-, \gx-, \rx-, and \rz-string operators.
        The averages of the four failure rates are also presented.
        We select schedule~1 that minimizes the largest logical failure rate among the four.
 }
	\label{fig:cnot_schedule_comparison}
\end{figure}

\begin{figure}[!t]
	\centering
	\includegraphics[width=\linewidth]{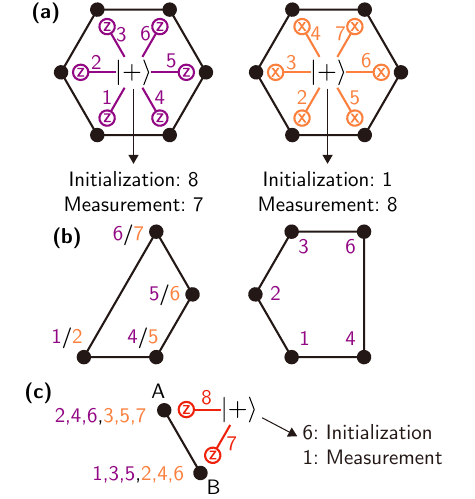}
	\caption{
        \subfig{a} Selected entangling gate schedule for measuring a $Z$-type check (left) and an $X$-type check (right).
        Syndrome qubits for $Z$-type ($X$-type) checks are initialized at time step~8 (1) and measured at time step~7 (8).
        \subfig{b} Examples of schedules for weight-4 and 5 checks.
        Note that only the schedule for the $Z$-type check is presented on the weight-5 check since the face can support only one check.
        \subfig{c} Schedule for measuring a weight-2 check $Z_\mr{A} \otimes Z_\mr{B}$ in a domain wall.
        An additional syndrome qubit is prepared at time step~6, undergo \cz gates with qubits~B and~A at time steps~7 and~8, respectively, and is measured at time step~1.
    }
	\label{fig:selected_schedule}
\end{figure}

By using one of the schedules in Eq.~\eqref{eq:valid_cnot_schedules}, a single round of syndrome extraction can be done in eight time steps, where syndrome qubits for $Z$-type ($X$-type) checks are measured at time step~7 (8) and initialized at time step~8 (1).
It is worth noting that two-body check measurements for lattice surgery in Figs.~\ref{fig:lattice_surgery_simple} and~\ref{fig:lattice_surgery_full} also can be done during these eight time steps, as exemplified in Fig.~\ref{fig:selected_schedule}(c).
This is because, for each pair of neighboring qubits, one of them is involved in entangling gates during time steps~1 to~6, while the other is involved during time steps~2 to~7, implying that two extra entangling gates for a two-body check measurement can be executed at time steps~7 and~8.
The additional syndrome qubit is initialized at time step~6 and measured at time step~1 of the next round.

In addition, we need to consider $Y$-type checks and mixed-Pauli checks as well, which may exist in domain walls.
Such checks can be measured by replacing entangling gates with corresponding controlled-Pauli gates while keeping the schedule.
For example, let us consider a Pauli-permuting domain wall ($\xbs \leftrightarrow \zbs$) crossing a face, with one side having qubits~A, B, C, and~F and the other side having qubits~D and~E in Fig.~\ref{fig:syndrome_extraction_circuit}(a).
In this case, we can replace the $Z$-type check $\sgz{f}$ with $S_1 \coloneqq Z_{\mr{A}} Z_{\mr{B}} Z_{\mr{C}} X_{\mr{D}} X_{\mr{E}} Z_{\mr{F}}$ and apply \cnot gates instead of \cz gates on qubits~D and~E when measuring $S_1$.
The $X$-type check $\sgx{f}$ is replaced with $S_2 \coloneqq X_{\mr{A}} X_{\mr{B}} X_{\mr{C}} Z_{\mr{D}} Z_{\mr{E}} X_{\mr{F}}$ and measured analogously.
Note that the schedule is still valid even if $\sgz{f}$ is replaced with $S_2$ and $\sgx{f}$ is replaced with $S_1$.
We can verify that the second condition for the schedule to be valid (i.e., check measurements should not interfere each other) holds even if entangling gates are replaced like this.
For example, in the above example, there is a risk that $S_1$ interferes with the $Z$-type check on the adjacent face containing qubits~D and~E.
However, it does not happen since $(a - e)(b - d) > 0$ for all the 24 valid length-7 schedules, although it is not one of Eq.~\eqref{eq:13_positive_numbers}.
In general, six numbers $(a-e)(b-d)$, $(b-f)(c-e)$, $(c-a)(d-f)$, $(g-k)(h-j)$, $(h-l)(i-k)$, and $(i-g)(j-l)$ are all positive, thus we do not need to worry about this problem.

It is worth noting that the selected schedule is the same as the one used in Ref.~\cite{lee2025color} rotated by \ang{60}, implying that the numerical results in Ref.~\cite{lee2025color} can be directly applied to inverted (base-up) triangular patches (including \patchout and some auxiliary patches) in our MSD layouts.
Our layouts also contain upright (base-down) triangular patches, but we assume for simplicity that their logical error rates are the same as those of inverted triangular patches having the same code distances 
As numerical evidence for justifying this assumption, we simulate 7 rounds of syndrome extraction of a distance-7 upright triangular patch at $p=10^{-3}$ based on the selected schedule.
The obtained $\lgz$ and $\lgx$ failure rates are $(7.24 \pm 0.04) \times 10^{-4}$ and $(7.22 \pm 0.04) \times 10^{-4}$ (99\% confidence intervals), respectively, which differ from the failure rates $(7.19 \pm 0.04) \times 10^{-4}$ obtained in Ref.~\cite{lee2025color} only by $\sim 0.7\%$.

\section{Decoder simulations for memory and stability experiments \label{app:decoder_simulations}}

In this appendix, we provide details of our decoder simulations for memory and stability experiments using the concatenated MWPM decoder under circuit-level noise. 
We outline the simulation methods, present the corresponding results, and discuss our reasoning behind selecting the ansatz given in Eq.~\eqref{eq:ansatz}.

\subsection{Methods for simulating memory and stability experiments \label{subsec:decoder_simulation_methods}}

For a memory experiment with a triangular or rectangular patch, we consider $4d$ rounds of syndrome extraction (following the schedule in Fig.~\ref{fig:selected_schedule}) for code distance $d$ (by setting $\dx = \dz = d$ for rectangular patches).
This process is preceded and followed by Pauli-$Z$ temporal boundaries; namely, all the physical qubits are initially prepared as $\ket{0}$ and finally measured in the $Z$ basis.
For a triangular patch, this means that the logical qubit encoded in the patch is initially prepared as $\ketlgc{0}$ and finally measured in the $\lgz$ basis.
For a rectangular patch, the two logical qubits are initially prepared as $\ketlgc{0} \otimes \ketlgc{+}$ and finally measured in the basis of $\lgc{Z} \otimes \lgc{X}$.
Each round of syndrome extraction is conducted according to the schedule in Fig.~\ref{fig:selected_schedule}(a).
Note that this setting can only detect $\lgz$ failures (for a triangular patch) or $\lgz_1$ and $\lgx_2$ failures (for a rectangular patch).
For complete analysis covering other types of failures, we consider the same setting again but with a schedule modified such that its $X$-part and $Z$-part are swapped.

For a stability experiment, we consider the patch in Fig.~\ref{fig:stability_experiment}(a) with $L_\mr{H} = L_\mr{V} = T$, which encodes no logical qubits, undergoing $T$ rounds of syndrome extraction with the same schedule.
This process is preceded and followed by red temporal boundaries, namely, all the red edges are initially prepared as $\ket{0}\ket{0} + \ket{1}\ket{1}$ and measured in the Bell basis.
If there are no errors, the product of red and blue checks with the same Pauli type should be always $+1$ whenever they are measured.
This fact can be used for determining whether a nontrivial timelike error string exists after decoding.

For simulating the above setting, we employ the \textit{color-code-stim} module \cite{lee2025color} with a slight modification to handle memory experiments with rectangular checks and stability experiments.
The module automatically generates a noisy color code circuit for given color code lattice, entangling gate schedule, and bit-flip or circuit-level noise model, with detectors (i.e., products of measurement outcomes that are deterministic when there are no errors and can be used for decoding) annotated.
It also provides features to simulate the circuit using the \textit{Stim} library \cite{gidney2021stim} and decode errors via the concatenated MWPM decoder \cite{lee2025color}, where each MWPM subroutine is implemented by the \textit{PyMatching} library \cite{higgott2022pymatching}.

\subsection{Results}

\begin{figure*}[!t]
	\centering
	\includegraphics[width=\linewidth]{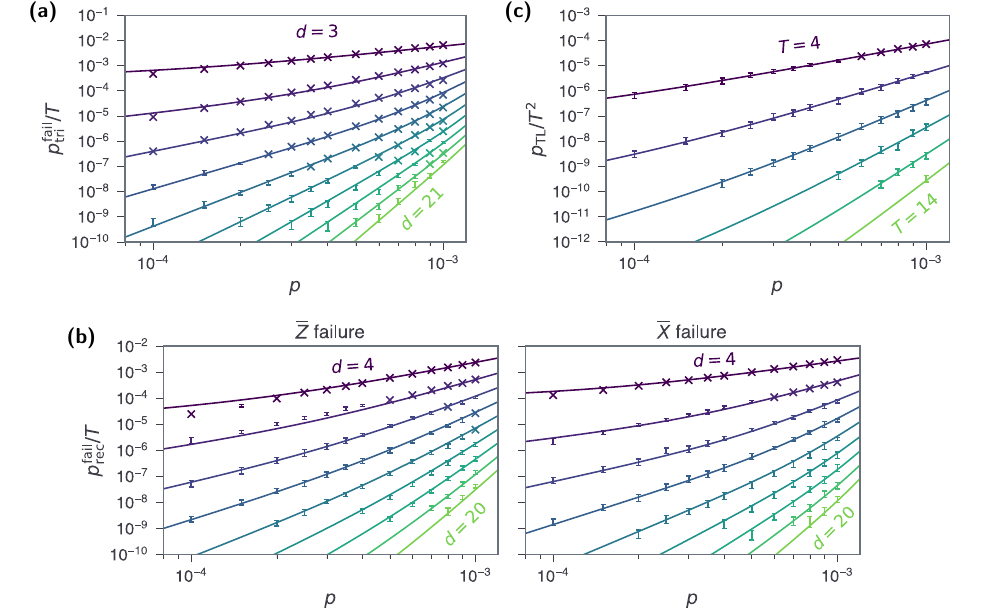}
	\caption{
        Per-round or per-area logical failure rates plotted against the physical error rate $p$ at varying $d$ (code distance) or $T$ (number of rounds), for \subfig{a} triangular patches, \subfig{b} rectangular patches, and \subfig{c} stability experiments.
        We consider $d \in \qty{3, 5, 7, \cdots, 21}$ for triangular patches, $d \in \qty{4, 6, 8, \cdots, 20}$ for rectangular patches, and $T \in \qty{4, 6, 8, \cdots, 14}$ for stability experiments.
        In \subfig{b}, the per-round failure rates of $\lgz$ and $\lgx$ (i.e., $\plogrec{Z.\mr{fail}}/T$ and $\plogrec{X.\mr{fail}}/T$) are separately presented.
        Each data point is marked as an error bar indicating its 99\% confidence interval or as an `X' symbol if its relative margin of error (i.e., the ratio of half the width of its confidence interval to its center) is less than 10\%.
        The solid lines are regression lines corresponding to individual $d$ values, based on the ansatz in Eq.~\eqref{eq:ansatz}. 
 }
	\label{fig:pfails_plot}
\end{figure*}

For a triangular patch, we denote the per-round failure rate of the $\lgx$ ($\lgz$) observable as $\plogtri{X(Z).\mr{fail}} \coloneqq \plogtri{Z(X)} + \plogtri{Y}$ and define $\plogtri{\mr{fail}} \coloneqq \plogtri{X.\mr{fail}} + \plogtri{Z.\mr{fail}}$.
Similarly, for a rectangular patch, we denote the per-round logical failure rates as $\plogrec{Z_{1(2)}.\mr{fail}} \coloneqq \plogrec{X_{1(2)}} + \plogrec{X_1 X_2}$ and $\plogrec{X_{1(2)}.\mr{fail}} \coloneqq \plogrec{Z_{1(2)}} + \plogrec{Z_1 Z_2}$, and define $\plogrec{Z.\mr{fail}} \coloneqq \plogrec{Z_1.\mr{fail}} + \plogrec{Z_2.\mr{fail}}$ and $\plogrec{X.\mr{fail}} \coloneqq \plogrec{X_1.\mr{fail}} + \plogrec{X_2.\mr{fail}}$.

In Fig.~\ref{fig:pfails_plot}, logical failure rates obtained from the memory and stability experiments are plotted for various values of $p$ and $d$ (or $T$ for the stability experiments).
Figures~\ref{fig:pfails_plot}(a) and~(b) are for triangular and rectangular patches, respectively, which show $\plogtri{\mr{fail}}/T$ for triangular patches and $\plogrec{Z.\mr{fail}}/T$ and $\plogrec{X.\mr{fail}}/T$ for rectangular patches, where $T=4d$.
Figure~\ref{fig:pfails_plot}(c) is for stability experiments, presenting the per-area failure rate $\ptimelike{}/T^2 = (\ptimelike{X} + \ptimelike{Z})/T^2$, where $\ptimelike{X/Z}$ is defined in Appendix~\ref{subsec:patch_logical_error_rates}.

\subsection{Ansatz selection}

We now select the ansatz by which these per-round and per-area failure rates can be approximated, which is a function of $p$ and $d$ (denoting $d \coloneqq T$ for stability experiments).
A representative ansatz commonly used in the literature \cite{fowler2019low,litinski2019game,litinski2019magic,lee2025color} is
\begin{align*}
    f(p, d) = \alpha \qty(\frac{p}{p_\mr{th}})^{\beta d + \eta},
\end{align*}
which depends on four parameters $p_\mr{th}$, $\alpha$, $\beta$, and $\eta$.
Note that the per-round logical error rate of a surface code under circuit-level noise can be approximated by this ansatz with $(p_\mr{th}, \alpha, \beta, \eta) \approx (0.01, 0.1, 0.5, 0.5)$ \cite{fowler2019low}.
However, this ansatz is asymptotically valid for sufficiently small values of $p$, where higher-order terms of $p$ are negligible.
The ansatz cannot approximate the data in Fig.~\ref{fig:pfails_plot} sufficiently well, thus we add one higher-order term in the ansatz.
The modified ansatz can be in the form of
\begin{align}
    f(p, d) = \alpha \qty(\frac{p}{p_\mr{th}})^{\beta d + \eta} \qty[1 + \epsilon \qty(\frac{p}{p_\mr{th}})^{g(d)}],
    \label{eq:modified_ansatz}
\end{align}
where $\epsilon$ is an additional parameter and $g$ is a function of $d$ such that $g(d) > 0$ for every $d$.
Since the form of the function $g$ is unknown, we choose it among the seven candidates
\begin{gather*}
    g(d) = \zeta, \quad g(d) = \zeta d, \quad g(d) = \zeta_0 + \zeta_1 d, \\
    g(d) = \zeta_0 + \zeta_1 d + \zeta_2 d^2, \quad g(d) = d^\lambda, \\
    g(d) = \zeta d^\lambda, \quad g(d) = \zeta_0 + \zeta_1 d^\lambda,
\end{gather*}
which respectively have at most three additional parameters.
Note that having more parameters in an ansatz does not always lead to better performance due to the risk of overfitting, which can reduce the model's ability to generalize to unseen data.
To prevent overfitting and ensure generalizability, we use leave-one-out cross-validation (LOOCV) \cite{efron1983leisurely}, a model validation technique where the model is trained on the dataset excluding a single data point and tested on that excluded data point, iteratively for all data points. 
We then calculate the root mean square deviation (RMSD) of the tested data points as the LOOCV score, where a lower score indicates better performance.
Here we note that the regressions are done through the least-squares method in the logarithmic scale with a transformed ansatz
\begin{multline*}
    \log f(p, d) = \log\alpha + (\beta d + \eta)(\log p - \log p_\mr{th}) \\
    + \log\qty[1 + \epsilon e^{(\log p - \log p_\mr{th})g(d)}],
\end{multline*}
using the function \texttt{curve\_fit} in the Python library \textit{SciPy} \cite{virtanen2020scipy}.
The LOOCV score is computed based on this logarithmic scale as well.

Table~\ref{table:ansatz_selection} presents the computed LOOCV scores for these seven candidates, all of which are significantly better than the ansatz without a higher-order term.
The candidate $g(d) = \zeta d^\lambda$ shows the best overall performance, thus we select this as our ansatz.
Note that $g(d) = d^\lambda$ performs better for the $\lgz$ failure rate of rectangular patches, but we choose $g(d) = \zeta d^\lambda$ for consistency.

\begin{table*}[t!]
    \centering
    \begin{ruledtabular}
    \begin{tabular}{ccccc}
        \multirow{3}{*}{Ansatz} & \multicolumn{4}{c}{LOOCV score} \\
         & \multirow{2}{*}{Triangular patches} & \multicolumn{2}{c}{Rectangular patches} & \multirow{2}{*}{Stability experiments} \\
         &  & $\lgz$ failure & $\lgx$ failure &  \\ \hline
        No higher-order term ($\epsilon = 0$) & 0.291 & 0.217 & 0.254 & 0.109 \\
        $g(d) = \zeta$ & 0.285 & 0.189 & 0.231 & 0.108 \\
        $g(d) = \zeta d$ & 0.129 & 0.195 & 0.138 & 0.0922 \\
        $g(d) = \zeta_0 + \zeta_1 d$ & 0.129 & 0.153 & 0.133 & 0.0905 \\
        $g(d) = \zeta_0 + \zeta_1 d + \zeta_2 d^2$ & 0.130 & 0.150 & 0.135 & 0.0905 \\
        $g(d) = d^\lambda$ & 0.162 & 0.125 & 0.159 & 0.0970 \\
        ($\bigast$) $g(d) = \zeta d^\lambda$ & \textbf{0.129} & \textbf{0.145} & \textbf{0.134} & \textbf{0.0887} \\
        $g(d) = \zeta_0 + \zeta_1 d^\lambda$ & 0.129 & 0.164 & 0.140 & 0.0970
    \end{tabular}
    \end{ruledtabular}
    \caption{
        Leave-one-out cross-validation (LOOCV) scores of several different ansatz candidates for logical failure rates, where $g$ is the function in Eq.~\eqref{eq:modified_ansatz}.
        A lower LOOCV score indicates better generalizability.
        We select the ansatz with $g(d) = \zeta d^\lambda$, highlighted by ($\bigast$).
    }
    \label{table:ansatz_selection}
\end{table*}

\subsection{Ansatz parameter estimates \label{subsec:ansatz_parameter_estimates}}

By fitting the data into the selected ansatz, we obtain the regression lines in Fig.~\ref{fig:pfails_plot}.
The estimated values of the ansatz parameters are as follows.

\begin{itemize}
    \item Triangular patches [Fig.~\ref{fig:pfails_plot}(a)]:
    \begin{gather*}
        p_\mr{th} = 2.41 \times 10^{-3}, \quad \alpha = 6.19 \times 10^{-4}, \quad \beta = 0.537, \\
        \eta = -1.45, \quad \epsilon = 27.2, \quad \zeta = 0.404, \quad \lambda = 0.933.
    \end{gather*}
    \item Rectangular patches, $\lgz$ failure [Fig.~\ref{fig:pfails_plot}(b), left]:
    \begin{gather*}
        p_\mr{th} = 4.17 \times 10^{-3}, \quad \alpha = 5.68 \times 10^{-4}, \quad \beta = 0.439, \\
        \eta = -1.04, \quad \epsilon = 88.1, \quad \zeta = 0.927, \quad \lambda = 0.332,
    \end{gather*}
    \item Rectangular patches, $\lgx$ failure [Fig.~\ref{fig:pfails_plot}(b), right]:
    \begin{gather*}
        p_\mr{th} = 3.07 \times 10^{-3}, \quad \alpha = 2.07 \times 10^{-4}, \quad \beta = 0.553, \\
        \eta = -2.05, \quad \epsilon = 73.1, \quad \zeta = 0.515, \quad \lambda = 0.742.
    \end{gather*}
    \item Stability experiments [Fig.~\ref{fig:pfails_plot}(c)]:
    \begin{gather*}
        p_\mr{th} = 6.24 \times 10^{-3}, \quad \alpha = 6.91 \times 10^{-6}, \quad \beta = 0.601, \\
        \eta = -1.61, \quad \epsilon = 543, \quad \zeta = 0.800, \quad \lambda = 0.389.
    \end{gather*}
\end{itemize}

\begin{widetext}

\section{Rules to determine noise channels from error sources during MSD \label{app:determining_noise_channels}}

In this appendix, we present the rules to determine noise channels from various possible error sources during MSD, as an extension of the discussion in Sec.~\ref{subsec:analysis_method}.
We here denote the logical-level support of a logical unitary operator $\lgc{U}$ as $\logsupp \lgc{U}$.

\subsection{Memory errors of logical patches \label{subsec:memory_errors}}

First, each logical patch suffers memory errors during or between lattice surgery.
We include initialization and measurement errors of logical qubits (except for faulty T-measurements) into this category for convenience.
Importantly, during the merging operation of each stage for the single-level MSD scheme, we ignore $\lgx$ and $\lgy$ errors of the qubits involving the stage.
This is because new logical operators replacing such $\lgx$ and $\lgy$ operators during the merging operation are represented by string-net operators having weights strictly larger than the code distances, as exemplified in Fig.~\ref{fig:lattice_surgery_full}(d).
Likewise, in the cultivation-MSD scheme, $\lgx$ and $\lgz$ errors can be ignored during merging operations.

For each stage~$k$ including the interval before the next merging operation (or including the final $\lgx$ measurements if $k=8$), there exist noise channels $\Lambda_{\lgc{U}, p_\mr{err}}$ characterized as follows:
\begin{enumerate}
    \item $\lgz$ error on each active qubit $\lgq{} \in \qty{\lgq{out}, \lgq{A}, \lgq{B}, \lgq{C}, \lgq{D}}$:
    \begin{align*}
        \lgc{U} &= \rot{\qty(\lgz_\qty{\lgq{}})}{\pi/2}, \\
        p_\mr{err} &= \begin{cases}
            \qty(\dm + 1)p^Z & \text{for the single-level scheme with any $k$ and the cultivation-MSD scheme with $k = 8$}, \\
            T_\mr{intv} p^Z & \text{for the cultivation-MSD scheme with $k < 8$},
        \end{cases}
    \end{align*}
    where $T_\mr{intv}$ is the average number of rounds between successive stages and $p^Z$ is the corresponding $\lgz$ error rate $\plogtri{Z}(p, \dout)$, $\plogrec{Z_1}(p, \dx, \dz)$ (for qubits~A and~C), or $\plogrec{Z_2}(p, \dx, \dz)$ (for qubits~B and~D).
    
    \item $\lgx$ error on each active qubit $\lgq{} \in \qty{\lgq{out}, \lgq{A}, \lgq{B}, \lgq{C}, \lgq{D}}$:
    \begin{align*}
        \lgc{U} &= \prod_{\substack{j \leq k \\ \lgq{} \in \logsupp\Pstage{j}}} \rot{\Pstage{j}}{-\pi/4} \prod_{\substack{j \leq k \\ \lgq{} \in \logsupp\Qstage{j}}} \rot{\Qstage{j}}{-\pi/4} \\
        p_\mr{err} &= \begin{cases}
            \qty(\dm + 1 - \dm\lambda_{\lgq{}}^{(k)})p^X & \text{for the single-level scheme with any $k$} \\
            & \qquad \text{and the cultivation-MSD scheme with $k = 8$}, \\
            \qty[T_\mr{intv} - \dm\lambda_{\lgq{}}^{(k)}]p^X & \text{for the cultivation-MSD scheme with $k < 8$},
        \end{cases}
    \end{align*}
    where $p^X$ is the corresponding $\lgx$ error rate $\plogtri{X}(p, \dout)$, $\plogrec{X_1}(p, \dx, \dz)$, or $\plogrec{X_2}(p, \dx, \dz)$, and
    \begin{align*}
        \lambda_{\lgq{}}^{(k)} = \begin{cases}
            1 & \text{if } \lgq{} \in \logsupp \Pstage{k} \cup \logsupp \Qstage{k}, \\
            0 & \text{otherwise}.
        \end{cases}
    \end{align*}

    \item $\lgy$ error on the output qubit $\lgq{out}$ when it is active:
    \begin{align*}
        \lgc{U} &= \rot{\qty(\lgztp{O})}{\pi/2} \prod_{\substack{j \leq k \\ \lgq{out} \in \logsupp\Pstage{j}}} \rot{\Pstage{j}}{-\pi/4} \prod_{\substack{j \leq k \\ \lgq{out} \in \logsupp\Qstage{j}}} \rot{\Qstage{j}}{-\pi/4} \\
        p_\mr{err} &= \begin{cases}
            \qty(\dm + 1 - \dm\lambda_{\lgq{out}}^{(k)})\plogtri{Y}(p, \dout) & \text{for the single-level scheme with any $k$} \\
            & \qquad \text{and the cultivation-MSD scheme with $k = 8$}, \\
            \qty[T_\mr{intv} - \dm\lambda_{\lgq{out}}^{(k)}]\plogtri{Y}(p, \dout) & \text{for the cultivation-MSD scheme with $k < 8$}.
        \end{cases}
    \end{align*}

    \item Correlated $\lgz$ errors on each pair $(\lgq{1}, \lgq{2}) \in \{ (\lgq{A}, \lgq{B}), (\lgq{C}, \lgq{D}) \}$, caused by Pauli-$Y$ error strings:
    \begin{align*}
        \lgc{U} &= \rot{\qty(\lgz_\qty{\lgq{1}})}{\pi/2} \rot{\qty(\lgz_\qty{\lgq{2}})}{\pi/2}, \\
        p_\mr{err} &= \begin{cases}
            \qty(\dm + 1)\plogtri{Z_1 Z_2}(p, \dx, \dz) & \text{for the single-level scheme with any $k$} \\
            & \qquad \text{and the cultivation-MSD scheme with $k = 8$}, \\
            T_\mr{intv} \plogtri{Z_1 Z_2}(p, \dx, \dz) & \text{for the cultivation-MSD scheme with $k < 8$}.
        \end{cases}
    \end{align*}

    \item Correlated $\lgx$ errors on each pair $(\lgq{1}, \lgq{2}) \in \{ (\lgq{A}, \lgq{B}), (\lgq{C}, \lgq{D}) \}$, caused by Pauli-$Y$ error strings:
    \begin{align*}
        \lgc{U} &= \prod_{\substack{j \leq k \\ \lgq{1} \in \logsupp\Pstage{j} \oplus \lgq{2} \in \logsupp\Pstage{j} }} \rot{\Pstage{j}}{-\pi/4} \prod_{\substack{j \leq k \\ \lgq{1} \in \logsupp\Qstage{j} \oplus \lgq{2} \in \logsupp\Qstage{j}}} \rot{\Qstage{j}}{-\pi/4} \\
        p_\mr{err} &= \begin{cases}
            \qty(\dm + 1 - \dm\lambda_{\lgq{}}^{(k)})\plogtri{X_1 X_2}(p, \dx, \dz) & \text{for the single-level scheme with any $k$} \\
            & \qquad \text{and the cultivation-MSD scheme with $k = 8$}, \\
            \qty[T_\mr{intv} - \dm\lambda_{\lgq{}}^{(k)}]\plogtri{X_1 X_2}(p, \dx, \dz) & \text{for the cultivation-MSD scheme with $k < 8$},
        \end{cases}
    \end{align*}
    where `$\oplus$' denotes the logical XOR operation.
    
    \item (Only for the single-level scheme) $\lgc{W} \in \qty{\lgx, \lgy, \lgz}$ error on each of qubits~\talpha and~\tbeta:
    \begin{align*}
        \lgc{U} = \begin{cases}
            \rot{\Pstage{k}}{\theta} & \text{for qubit~\talpha}, \\
            \rot{\Qstage{k}}{\theta} & \text{for qubit~\tbeta},
        \end{cases} \qquad
        p_\mr{err} = \begin{cases}
            \dm\plogtri{Z}(p, \dm') & \text{if } W=Z, \\
            \plogtri{W}(p, \dm') & \text{otherwise},
        \end{cases}
    \end{align*}
    where
    \begin{align*}
        \theta = \begin{cases}
            -\pi/4  & \text{if } W = X, \\
            \pi/4 & \text{if } W = Y, \\
            \pi/2 & \text{if } W = Z, \\
        \end{cases} \qquad
        \dm' = \begin{cases}
            \dm & \text{if } k \geq 3, \\
            \dz - 1 & \text{otherwise}.
        \end{cases}
    \end{align*}
    
    \item (Only for the cultivation-MSD scheme) On each of qubits~\talpha and~\tbeta, $\lgz$ and $\lgx$ errors after the lattice surgery and $\lgy$ errors after the magic state is prepared (including the growing operation if $\dm > \dcult$ and the idling operation):
    \begin{align*}
        \lgc{U} &= \begin{cases}
            \rot{\Pstage{k}}{\pi/2} & \text{for qubit~\talpha}, \\
            \rot{\Qstage{k}}{\pi/2} & \text{for qubit~\tbeta},
        \end{cases} \\
        p_\mr{err} &= \plogtri{Z}(p, \dm)/2 + \plogtri{X}(p, \dm)/2 + (\dm + 1)\plogtri{Y}(p, \dm) + (1 - \delta_{\dm, \dcult}) \plogtri{Y}(p, \dcult) + T_\mr{idle} \plogtri{Y}(p, \dm),
    \end{align*}
    where $T_\mr{idle}$ is the average number of rounds that the auxiliary patches idle before they are consumed.

    \item (Only for the cultivation-MSD scheme) $\lgx$ and $\lgz$ errors during the growing operation (if $\dm > \dcult$) and the idling operation on each of the two qubits~\talpha and~\tbeta: Two additional noise channels characterized by $(\lgc{U}_+, p_{\mr{err}})$ and $(\lgc{U}_-, p_{\mr{err}})$, where
    \begin{align*}
        \lgc{U}_\pm &= \begin{cases}
            \rot{\Pstage{k}}{\pm\pi/4} & \text{for qubit~\talpha},\\
            \rot{\Qstage{k}}{\pm\pi/4} & \text{for qubit~\tbeta},\\
        \end{cases} \\
        p_\mr{err} &= \frac{1 - \delta_{\dm, \dcult}}{2} \qty[\plogtri{X}(p, \dcult) + \plogtri{Z}(p, \dcult)] + \frac{T_\mr{idle}}{2} \qty[\plogtri{X}(p, \dm) + \plogtri{Z}(p, \dm)].
    \end{align*}
    A $\lgx$ or $\lgz$ error flips the outcome of the lattice surgery, which results in a $\rot{\Pstage{k}}{\pi/4}$ or $\rot{\Pstage{k}}{-\pi/4}$ error depending on the measurement outcome.
\end{enumerate}

\subsection{Errors in the ancillary region \label{subsec:errors_ancillary_region}}

During the merging operation of stage~$k$, the ancillary region may have spacelike and timelike error strings, which cause logical errors.
For simpler calculations, we assume that irregular checks in domain walls have negligible effects on the fault tolerance of the color code lattice.

We first consider spacelike error strings when $k \geq 3$.
Certain Pauli-$X$ ($Z$) error strings in the ancillary region may incur $\lgz$ errors on the logical qubits involved in $\PstageLS{k}$ ($\QstageLS{k}$).
It is guaranteed that their weights are not less than the corresponding code distances, as shown in Appendix~\ref{app:layout_fault_tolerance}.
Hence, considering the case where the output qubit is involved in $\PstageLS{k}$ as an example, the output qubit has an additional $\lgz$ error rate, which can be estimated from the $\lgz_2$ error rate of an imaginary rectangular patch with distances $\dx' = \dV + 6$ and $\dz' = \dout + 1$, where
\begin{align*}
    \dH &= \max(2\dz, \dout + 1) + (\Nm - N_\mr{m.side})(\dm + 1), \\
    \dV &= \max[\dz, N_\mr{m.side} (\dm + 1)]
\end{align*}
are the horizontal and vertical dimensions of the ancillary patch (defining $\Nm = N_\mr{m.side} = 1$ for the single-level MSD scheme).
Here $\dx'$ has a constant term $+6$ for covering interface regions in addition to the ancillary patch.
In the following, we list the estimated additional $\lgz$ error rates of the logical qubits made by spacelike errors in the ancillary region during stage~$k \geq 3$:
\begin{itemize}
    \item (Only for $k \geq 3$) Additional $\lgz$ error rate on each qubit in $\logsupp{\Pstage{k}} \subseteq \qty{\lgq{out}, \lgq{A}, \lgq{B}, \lgq{C}, \lgq{D}}$, which originates from horizontal Pauli-$Z$ error strings in the ancillary region:
    \begin{align*}
        p_\mr{err.add} = \dm \plogrec{Z_2} (p, \dV + 6, \dz'),
    \end{align*}
    where $\dz' = \dout + 1$ for the output qubit and $\dz' = \dz$ for the validation qubits.
    \item (Only for $k \geq 3$) Additional $\lgz$ error rate on each qubit in $\logsupp{\Qstage{k}} \subseteq \qty{\lgq{out}, \lgq{A}, \lgq{B}, \lgq{C}, \lgq{D}}$, which originates from horizontal Pauli-$X$ error strings in the ancillary region:
    \begin{align*}
        p_\mr{err.add} = \dm \plogrec{Z_1} (p, \dV + 6, \dz'),
    \end{align*}
    where $\dz' = \dout + 1$ for the output qubit and $\dz' = \dz$ for the validation qubits.
    \item (Only for $k \geq 3$) Additional $\lgz$ error rate on each qubit in $\logsupp{\Pstage{k}} \triangle \logsupp{\Qstage{k}} \subseteq \qty{\lgq{out}, \lgq{A}, \lgq{B}, \lgq{C}, \lgq{D}}$, where $A \triangle B \coloneqq A \cup B \setminus (A \cap B)$, which originates from horizontal Pauli-$Y$ error strings in the ancillary region:
    \begin{align*}
        p_\mr{err.add} = \dm \plogrec{Z_1 Z_2} (p, \dV + 6, \dz'),
    \end{align*}
    where $\dz' = \dout + 1$ for the output qubit and $\dz' = \dz$ for the validation qubits.
    \item (Only for the single-level scheme with $k \geq 3$) Additional $\lgz$ error rate on each of qubits~\talpha and~\tbeta, which originates from vertical Pauli-$Z$ and $Y$ error strings (for qubit~\talpha) or vertical Pauli-$X$ and $Y$ error strings (for qubit~\tbeta) in the ancillary region:
    \begin{align*}
        p_\mr{err.add} = \dm \plogrec{X_i} (p, \dm + 1, \dH + 6) + \dm \plogrec{X_1 X_2} (p, \dm + 1, \dH + 6),
    \end{align*}
    where $i=1$ for qubit~\talpha and $i=2$ for qubit~\tbeta.
    \item (Only for the cultivation-MSD scheme with $k \geq 3$) Additional $\lgy$ error rate on each of qubits~\talpha and~\tbeta (item~7 of Sec.~\ref{subsec:memory_errors}), which originates from vertical Pauli-$Z$ and $Y$ error strings (for qubit~\talpha) or vertical Pauli-$X$ and $Y$ error strings (for qubit~\tbeta) in the ancillary region:
    \begin{align*}
        p_\mr{err.add} ={} &\frac{N_\mr{m.side}}{\Nm} \dm \qty[\plogrec{X_i} (p, \dm + 1, \dH + 6) + \plogrec{X_1 X_2} (p, \dm + 1, \dH + 6)] \\
        &+ \frac{\Nm - N_\mr{m.side}}{\Nm} \dm \qty[\plogrec{Z_j} (p, \dV + 6, \dm + 1) + \plogrec{Z_1 Z_2} (p, \dV + 6, \dm + 1)],
    \end{align*}
    where $(i, j) = (1, 2)$ for qubit~\talpha and $(i, j) = (2, 1)$ for qubit~\tbeta.
    This is the average of error rates for the $2\Nm$ auxiliary patches.
\end{itemize}

When $k < 3$, the single-level MSD scheme employs the simple layout as shown in Figs.~\ref{fig:lattice_surgery_simple}(a) and~(b).
As in the case of $k \geq 3$, we consider an imaginary rectangular patch covering the thin ancillary region of the layout for estimating additional logical error rates from spacelike error strings.
On the other hand, the cultivation-MSD scheme uses the full layout in Fig.~\ref{fig:msd_layout_twolevel} even for $k < 3$, and the ancillary region is not guaranteed to be distance-preserving unlike the cases of $k \geq 3$.
For example, stage~1 involves the measurement of $\PstageLS{1} = \lgz_\mr{A} \otimes \lgy_\mralpha$, thus a $\lgz$ error on qubit~A is equivalent to a $\lgy$ error on qubit~\talpha, meaning that there may exist \gx-string operators with weight $\min(\dz, \dm + 1)$ in the ancillary region that cause a $\lgz$ error on qubit~A.
Hence, we instead consider an imaginary rectangular patch with distances $\dx'=\dV$ and $\dz' = \min(\dz, \dm + 1)$.
We thus obtain the following additional error rates for $k < 3$:
\begin{itemize}
    \item (Only for the single-level scheme with $k < 3$) Additional $\lgz$ error rate on each of qubits~A, C (if $k=1$) or qubits~B, D (if $k=2$):
    \begin{align*}
        p_\mr{err.add} = \dm\max\qty[\plogrec{Z_1} (p, 4, \dz), \plogrec{Z_2} (p, 4, \dz)] + \dm\plogrec{Z_1 Z_2}(p, 4, \dz).
    \end{align*}
    Note that we take the maximum between $\plogrec{Z_1}$ and $\plogrec{Z_2}$ due to ambiguity from a Pauli-permuting domain wall.
    \item (Only for the cultivation-MSD scheme with $k < 3$) Additional $\lgz$ error rate on each of qubits~A, C (if $k=1$) or qubits~B, D (if $k=2$):
    \begin{align*}
        p_\mr{err.add} = \dm\plogrec{Z_i}(p, \dV + 6, \min(\dz, \dm + 1)) + \dm\plogrec{Z_1 Z_2}(p, \dV + 6, \min(\dz, \dm + 1))
    \end{align*}
    where $i=2$ for qubits~A and~B and $i=1$ for qubits~C and~D.
    Note that $p^{(k)}$ involves qubits~A or~B and $q^{(k)}$ involves qubits~C or~D.
\end{itemize}

Next, the ancillary region may also have timelike error strings.
A timelike error string that corrupts red $X$-type checks (e.g., $\dm$ consecutive measurement errors of a red $X$-type check) flips the measurement outcome of $\Pstage{k}$, as its value is determined from the product of $X$-type check outcomes in the ancillary region; see Fig.~\ref{fig:lattice_surgery_full}(a).
In the single-level scheme, this is equivalent to a $\lgx$ error on qubit~\talpha after the lattice surgery.
In the cultivation-MSD scheme, this makes a $\rot{\Pstage{k}}{\pi/4}$ or $\rot{\Pstage{k}}{-\pi/4}$ error on the output and validation qubit, depending on the measurement outcome of $\Pstage{k} \otimes \lgy_\mralpha$, which is the same effect as a $\lgx$ error during the growing operation (item~8 in Sec.~\ref{subsec:memory_errors}).
Considering timelike error strings corrupting $Z$-type checks as well, we obtain the following noise channels for each stage $k$:
\begin{itemize}
    \item (Only for the single-level scheme) Additional $\lgx$ error rate on each of qubits~\talpha and~\tbeta:
    \begin{align*}
        p_\mr{err.add} = \begin{cases}
            \ptimelike{X}(p, \dH + 6, \dV + 6, \dm) & \text{for qubit~\talpha when $k > 3$}, \\
            \ptimelike{Z}(p, \dH + 6, \dV + 6, \dm) & \text{for qubit~\tbeta when $k > 3$}, \\
            \max\qty[\ptimelike{X}(p, \dz, 4, \dm), \ptimelike{Z}(p, \dz, 4, \dm)] & \text{when $k < 3$},
        \end{cases}
    \end{align*}
    \item (Only for the cultivation-MSD scheme) Additional error rate on $p_\mr{err}$ of item~8 in Sec.~\ref{subsec:memory_errors}:
    \begin{align*}
        p_\mr{err.add} = \begin{cases}
            \ptimelike{X}(p, \dH + 6, \dV + 6, \dm)/2 & \text{for qubit~\talpha}, \\
            \ptimelike{Z}(p, \dH + 6, \dV + 6, \dm)/2 & \text{for qubit~\tbeta}, \\
        \end{cases}
    \end{align*}
\end{itemize}

\subsection{Errors from non-Clifford components \label{subsec:errors_from_nonclifford_components}}

The non-Clifford components of our MSD schemes comprise the faulty T-measurement for the single-level scheme and cultivation for the cultivation-MSD scheme.
We now explore how errors in these components affect the output logical states.

For the single-level scheme, based on the argument in Sec.~\ref{subsec:faulty_T_measurement}, we handle errors caused by faulty T-measurements as follows:
\begin{itemize}
    \item (Only for the single-level scheme) For each stage $k$, each of qubits~\talpha and~\tbeta additionally has a $\lgx$ error rate of $2p/3$, a $\lgy$ error rate of $2p/3$, and a $\lgz$ error rate of $5p/3$.
\end{itemize}

For the cultivation-MSD scheme, a $\lgx$ or $\lgz$ error on a magic state has the same effect as the memory error during the growing operation, which is covered in item~8 of Sec.~\ref{subsec:memory_errors}.
A $\lgy$ error can be handled just as an additional $\lgy$ error.
Hence, the effects of errors from cultivation are as follows for each stage $k$:
\begin{itemize}
    \item (Only for the cultivation-MSD scheme) Additional error rate on $p_\mr{err}$ of item~8 in Sec.~\ref{subsec:memory_errors}:
    \begin{align*}
        p_\mr{err.add} = \frac{1}{2}\qty[\pcult{X} + \pcult{Z}]
    \end{align*}
    \item (Only for the cultivation-MSD scheme) Additional $\lgy$ error rate on each of qubits~\talpha and~\tbeta:
    \begin{align*}
        p_\mr{err.add} = \pcult{Y}
    \end{align*}
\end{itemize}

\section{Analytic expressions of the success probability and output infidelity of MSD \label{app:MSD_performance_expressions}}

In this appendix, we present the analytic expressions of the success probability $q_\mr{succ}$ and the output infidelity $\infidMSD$ of the single-level and cultivation-MSD schemes, as functions of code distances ($\dout, \dx, \dz, \dm$), physical noise strength $p$, and the logical error rates of several patches.
The expressions are obtained via the method in Sec.~\ref{subsec:analysis_method} and Appendix~\ref{app:determining_noise_channels}.
For simplicity, we assume that, for a parameter $\pyratio \geq 0$,
\begin{align*}
    \plogtri{}(p, d) &\coloneqq \plogtri{X}(p, d) = \plogtri{Z}(p, d) = \plogtri{}(p, d) / \pyratio, \\
    \plogrec{X}(p, \dx, \dz) &\coloneqq \plogrec{X_1}(p, \dx, \dz) = \plogrec{X_2}(p, \dx, \dz) = \plogrec{X_1 X_2}(p, \dx, \dz) / \pyratio, \\
     \plogrec{Z}(p, \dx, \dz) &\coloneqq \plogrec{Z_1}(p, \dx, \dz) = \plogrec{Z_2}(p, \dx, \dz) = \plogrec{Z_1 Z_2}(p, \dx, \dz) / \pyratio, \\
     \pcult{} &\coloneqq \pcult{X} = \pcult{Z} = \pcult{Y} / \pyratio, \\
     \pgrow{}(p, \dcult, \dm, c_\mr{gap}) &\coloneqq \pgrow{X}(p, \dcult, \dm, c_\mr{gap}) = \pgrow{Z}(p, \dcult, \dm, c_\mr{gap}) = \pgrow{Y}(p, \dcult, \dm, c_\mr{gap}) / \pyratio, \\
     L_\mr{H} L_\mr{V} \ptimelike{} (p, T) &\coloneqq \ptimelike{X}(p, L_\mr{H}, L_\mr{V}, T) = \ptimelike{Z}(p, L_\mr{H}, L_\mr{V}, T),
\end{align*}
where the last line implies that $\ptimelike{X}(p, L_\mr{H}, L_\mr{V}, T) = \ptimelike{Z}(p, L_\mr{H}, L_\mr{V}, T) \propto L_\mr{H} L_\mr{V} \propto \text{(Area of the patch)}$.

\subsection{Single-level MSD scheme \label{subsec:analytic_expression_sng}}

We define
\begin{align*}
    p_\mr{out} &\coloneqq \plogtri{}(p, \dout), & p_\mr{anc.out} &\coloneqq \plogrec{Z}(p, \dV + 6, \dout + 1), \\
    p_\mr{rec.X} &\coloneqq \plogrec{X}(p, \dx, \dz), & p_\mr{anc.rec} &\coloneqq \dm(1 + \pyratio)\plogrec{Z}(p, \dV + 6, \dz), \\
    p_\mr{rec.Z} &\coloneqq \pyratio (\dm + 1) \plogrec{Z}(p, \dx, \dz), & p_\mr{anc.rec.pre} &\coloneqq \dm(1 + \pyratio)\plogrec{Z}(p, 4, \dz),\\
    p_\mr{m} &\coloneqq (2\dm + 1 + \pyratio) \plogtri{}(p, \dm), & p_\mr{anc.m} &\coloneqq \dm(1 + \pyratio)\plogrec{X}(p, \dm + 1, \dH + 6), \\
    p_\mr{timelike} &\coloneqq (\dH + 6)(\dV + 6)\ptimelike{}(p, \dm), & &
\end{align*}
where $\dH \coloneqq \max\qty(\dout+1, 2\dz)$ and $\dV \coloneqq \max\qty(\dm+1, \dz)$.
We present terms of $q_\mr{succ}$ and $\infidMSD$ up to $O(p^3)$, $O(p^2\lgc{p})$, and $O(\lgc{p}\lgc{p}')$, where $\lgc{p}$ and $\lgc{p}'$ are any of the symbols defined above.
In addition, we show only the lowest-order term among non-constant terms having comparable orders (e.g., if the expression has a term involving $p p_\mr{rec.X}$, terms involving $p^2 p_\mr{rec.X}$ are omitted).
The obtained expressions of $\infidMSD$ and $q_\mr{succ}$ are as follows:
\begin{align*}
    \infidMSD \approx {}& 35 \qty(\frac{7}{3}p)^{3}
    + \frac{16 \dm + 19 + \pyratio}{4} p_\mr{out}
    + \dm \qty(7 + \pyratio) p_\mr{anc.out} \\
    &+ \frac{7}{2} \qty(\frac{7}{3}p)^{2} \qty(15  p_\mr{m} + 22 p_\mr{anc.m} + 8 p_\mr{anc.rec.pre} + \frac{11\qty(\dH + 6) \qty(\dV + 6) + 16 d_\mr{Z}}{\qty(\dH + 6) \qty(\dV + 6)} p_\mr{timelike}) \\
    &+ \qty(8p_\mr{rec.Z} + 4p_\mr{anc.rec} + \frac{27 + 19\pyratio}{8} p_\mr{rec.X})\qty(\frac{14}{3}p + p_\mr{m} + 2p_\mr{anc.m} + p_\mr{timelike}) \\
    &+ 4 p_\mr{rec.X} p_\mr{rec.Z}
    + 2 p_\mr{rec.X} p_\mr{anc.rec}
    + \frac{31 + 68 \pyratio}{8} p_\mr{rec.X}^{2}, \\
    q_\mr{succ} \approx {}&  1 - 35p
    - \frac{15}{2} p_\mr{m}
    - 11 p_\mr{anc.m}
    - 4 p_\mr{anc.rec.pre}
    - \frac{11\qty(\dH + 6) \qty(\dV + 6) + 16 d_\mr{Z}}{2 \qty(\dH + 6) \qty(\dV + 6)} p_\mr{timelike}
    - \frac{16 \qty(2 + \pyratio)}{\pyratio} p_\mr{rec.Z} \\
    &- \frac{4 \qty(5 + 3 \pyratio)}{1 + \pyratio} p_\mr{anc.rec}
    - \frac{4 \dm + 71 + 35 \pyratio}{4} p_\mr{rec.X}
    - 2 \qty(1 + \pyratio) p_\mr{out}.
\end{align*}
It is worth noting that $q_\mr{dist}$ does not contain a term of $O(p_\mr{rec.X})$, which is consistent with the argument in Sec.~\ref{subsec:distillation_circuit} showing that the MSD circuit is tolerant to a single-location $\lgx$ error on one of the validation qubits.

\subsection{Cultivation-MSD scheme \label{subsec:analytic_expression_combined_scheme}}

\begin{figure*}[!t]
	\centering
	\includegraphics[width=0.8\linewidth]{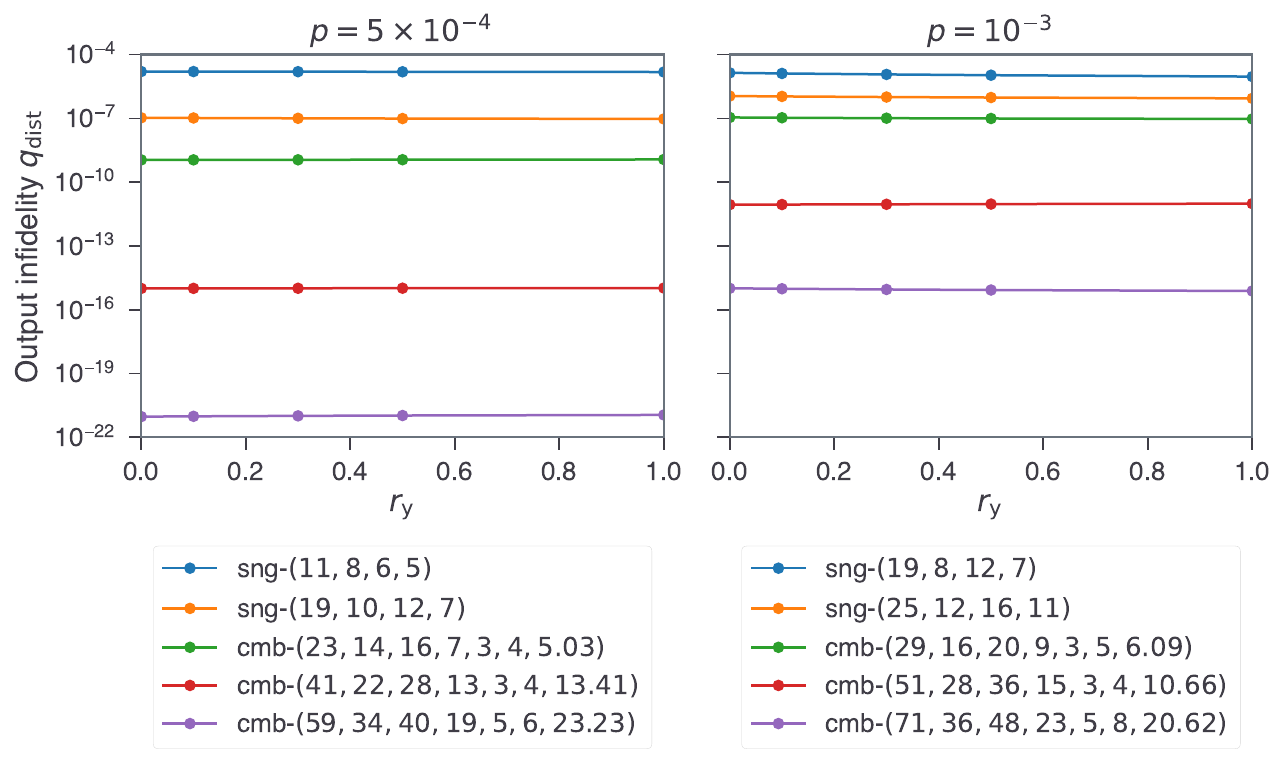}
	\caption{
        Dependency of the output infidelity on $\pyratio$ at $p \in \qty{5 \times 10^{-4}, 10^{-3}}$ for several variants of the single-level and cultivation-MSD schemes, denoted as `sng-$(\dout, \dx, \dz, \dm)$' and `cmb-$(\dout, \dx, \dz, \dm, \dcult, \Nm, c_\mr{gap})$', respectively.
 }
	\label{fig:pyratio_dependency}
\end{figure*}

We define
\begin{align*}
    p_\mr{out} &\coloneqq \plogtri{}(p, \dout), \\
    p_\mr{rec.X} &\coloneqq (T_\mr{intv} - \dm) \plogrec{X}(p, \dx, \dz), \\
    p_\mr{rec.Z} &\coloneqq \pyratio T_\mr{intv} \plogrec{Z}(p, \dx, \dz), \\
    p_\mr{m} &\coloneqq \qty[(1 + \pyratio)(T_\mr{idle} + \dm) + 1] \plogtri{}(p, \dm), \\
    p_\mr{timelike} &\coloneqq (\dH + 6)(\dV + 6) \ptimelike{}(p, \dm), \\
    p_\mr{magic} &\coloneqq (1 + \pyratio) \qty[\pcult{} + \qty(1 - \delta_{\dm, \dcult}) \pgrow{} (p, \dcult, \dm, c_\mr{gap})], \\
    p_\mr{anc.out} &\coloneqq \plogrec{Z}(p, \dV + 6, \dout + 1), \\
    p_\mr{anc.rec} &\coloneqq \dm(1 + \pyratio)\plogrec{Z}(p, \dV + 6, \dz), \\
    p_\mr{anc.rec.pre} &\coloneqq \dm(1 + \pyratio)\plogrec{Z}(p, \dm + 1, \dV + 6), \\
    p_\mr{anc.m} &\coloneqq \dm(1 + \pyratio) \plogrec{X}(p, \dm + 1, \dH + 6),
\end{align*}
where $\dH$ and $\dV$ are defined in Eq.~\eqref{eq:combined_scheme_anc_region_dim}, $T_\mr{intv}$ is the average number of rounds between successive stages, $T_\mr{idle}$ is the average number of rounds that auxiliary patches idle before they are consumed.
Note that $T_\mr{intv}$ and $T_\mr{idle}$ can be estimated by simulating the procedure of the cultivation-MSD scheme, as described in Sec.~\ref{subsec:combined_scheme} and implemented Ref.~\cite{github:msd-magic-state-prep-cycle-simulation}.
We compute terms of $q_\mr{succ}$ and $\infidMSD$ up to $O(\lgc{p}\lgc{p}'\lgc{p}'')$, where $\lgc{p}$, $\lgc{p}'$, and $\lgc{p}''$ are any of the symbols defined above and show only the lowest-order term (except a constant) among terms having comparable orders.
The obtained expressions of $\infidMSD$ and $q_\mr{succ}$ are then as follows:
\begin{align*}
    \infidMSD \approx{}& 35 \qty(p_\mr{magic} + p_\mr{m} + p_\mr{anc.m} + \frac{1}{2} p_\mr{timelike})^3 \\
    &+ 28p_\mr{anc.rec.pre} \qty(p_\mr{anc.m} + p_\mr{magic} + p_\mr{m} + \frac{1}{2} p_\mr{timelike})^2 \\
    &+ \qty(\frac{27 + 19 \pyratio}{4} p_\mr{rec.X} + 16 p_\mr{rec.Z} + 8p_\mr{anc.rec} + 6 p_\mr{anc.rec.pre}^{2}) \qty(p_\mr{magic} + p_\mr{m} + p_\mr{anc.m} + \frac{1}{2} p_\mr{timelike}) \\
    &+ \frac{19 T_\mr{intv} - 3 \dm + (T_\mr{intv} - \dm) \pyratio}{4} p_\mr{out} 
    + \dm \qty(7 + \pyratio) p_\mr{anc.out} \\
    &+ \frac{31 + 68 \pyratio}{8} p_\mr{rec.X}^{2} + \qty[\qty(\frac{5}{2} + \pyratio) p_\mr{anc.rec.pre}^2 + 2p_\mr{anc.rec} + 4p_\mr{rec.Z}] p_\mr{rec.X}, \\
    q_\mr{succ} \approx{}& 1 - 15 \qty(p_\mr{magic} + p_\mr{m} + p_\mr{anc.m} + \frac{1}{2} p_\mr{timelike})
    - \qty(\frac{71 + 35 \pyratio}{4} + \frac{\dm}{T_\mr{intv} - \dm}) p_\mr{rec.X}
    - \frac{16 \qty(2 + \pyratio)}{\pyratio} p_\mr{rec.Z} \\
    &- \frac{20 + 12 \pyratio}{1 + \pyratio} p_\mr{anc.rec}
    -4 p_\mr{anc.rec.pre}.
\end{align*}
\end{widetext}

\section{Dependency analysis of MSD performance on $\pyratio$ \label{app:pyratio_dependency}}

In Fig.~\ref{fig:pyratio_dependency}, we plot the output infidelities $q_\mr{dist}$ of our single-level and cultivation-MSD schemes against $\pyratio$ for several combinations of parameters at $p \in \qty{5 \times 10^{-4}, 10^{-3}}$, showing that $q_\mr{dist}$ does not significantly depend on $\pyratio$.
Even if we compare the two extreme cases of $\pyratio = 0$ and $\pyratio = 1$, the difference in infidelity is at most only a factor of $\sim 1.5$.

\section{Integrated cultivation + growing simulations \label{app:integrated_cult_growing_simulations}}

When analyzing the cultivation-MSD scheme in Sec.~\ref{sec:performance_analysis}, we handled the cultivation and growing operations as two separate modules.
In other words, we assumed that the cultivation process outputs a magic state with a given infidelity and success probability, which is inputted to the subsequent growing operation.
We simulated each operation independently and combined the results.
This was for keeping the modularity of our scheme by treating the cultivation part as a `black box'.
However, in an actual process, the `interface' between these two modules may have non-negligible effects on the overall performance of the scheme, thus simulating the two modules independently may give inaccurate results.

In this appendix, we analyze such effects in the case of using the scheme in Ref.~\cite{gidney2024magic} for cultivation.
We combine the cultivation and growing circuits into a single circuit and simulate it jointly.
Note that, as done in Ref.~\cite{gidney2024magic}, we replace $T$ with $S$ in the cultivation circuit for its efficient simulation using \textit{Stim} \cite{gidney2021stim} and track the value of $\lgy$.
Consequently, the cultivation logical error rate is estimated by doubling the obtained logical error rate of the modified circuit, following the same assumption in Ref.~\cite{gidney2024magic}.

The main difference from our original analysis arises from detectors connecting the cultivation and growing parts.
Originally, the last layer of detectors in the cultivation circuit for its simulation connects the final logical check measurements and the following virtual perfect syndrome measurements after the cultivation.
Similarly, the first layer of detectors (involving the initial patch) in the growing circuit connects the first syndrome measurements and the virtual perfect logical preparation layer before the growth.
By combining the two circuits, these two layers are merged naturally into a single layer of detectors.
Importantly, the values of these detectors are now used for post-selection of cultivation, not for decoding the growing operation.
Namely, we abort cultivation if any of these detectors gives $-1$, which is more probable than in the original scenario (assuming perfect syndrome measurements after the cultivation), implying that the success probability of cultivation may be reduced by this modification.

\begin{table}[bt!]
\centering
\begin{ruledtabular}
\begin{tabular}{cccc}
$p$ & $\dcult$ & $q_\mr{cult}^\mr{succ}$ (original) & $q_\mr{cult}^\mr{succ}$ (integrated) \\ \hline
$5 \times 10^{-4}$ & 3 & 0.83 & 0.80 \\
$5 \times 10^{-4}$ & 5 & 0.35 & 0.35 \\
$10^{-3}$ & 3 & 0.65 & 0.64 \\
$10^{-3}$ & 5 & 0.15 & 0.12
\end{tabular}
\end{ruledtabular}
\caption{Success rates $q_\mr{cult}^\mr{succ}$ of cultivation for the original circuit \cite{gidney2024magic} and the integrated cultivation + growing circuit.}
\label{table:cultivation_succ_rate_org_int_comparison}
\end{table}

We now present the results obtained from our simulations.
First, Table~\ref{table:cultivation_succ_rate_org_int_comparison} shows the success rates of the cultivation module within the integrated cultivation + growing circuit, alongside those from the original circuit for comparison.
This aligns with our expectation that integration would decrease the success rate, although the observed differences are not significant ($< 4\%$ difference) except in the case of $(p, \dcult) = (10^{-3}, 5)$.

An end-to-end resource cost analysis is presented in Figs.~\ref{fig:integrated_cult_growing_time_cost} and~\ref{fig:integrated_cult_growing_spacetime_cost}, showing the logical error rate as a function of the effective time cost and the spacetime cost, respectively, for $p \in \qty{5 \times 10^{-4}, 10^{-3}}$.
For comparison, these figures also contain the curves (dashed lines) obtained by combining results from our original simulations.

\begin{figure*}[!t]
	\centering
	\includegraphics[width=\textwidth]{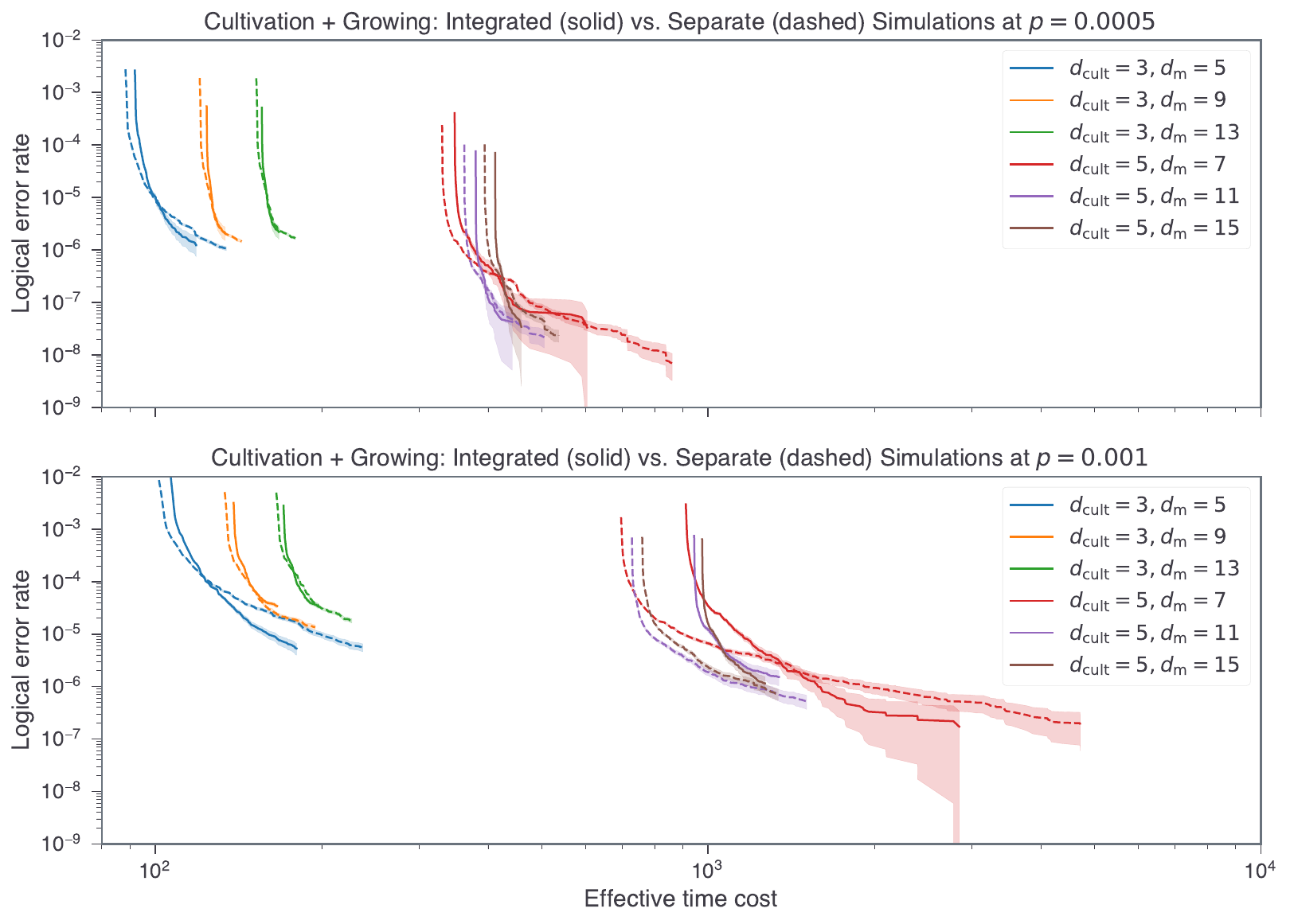}
	\caption{Logical error rate of the cultivation + growing circuit, plotted against the effective time cost for $p \in \qty{5 \times 10^{-4}, 10^{-3}}$ and various combinations of $\dcult$ and $\dm$. For comparison, integrated scenarios (where the cultivation and growing modules are simulated jointly) and separate scenarios (where the two modules are simulated separately and the results are combined straightforwardly) are represented by solid and dashed lines.}
	\label{fig:integrated_cult_growing_time_cost}
\end{figure*}

\begin{figure*}[!t]
	\centering
	\includegraphics[width=\textwidth]{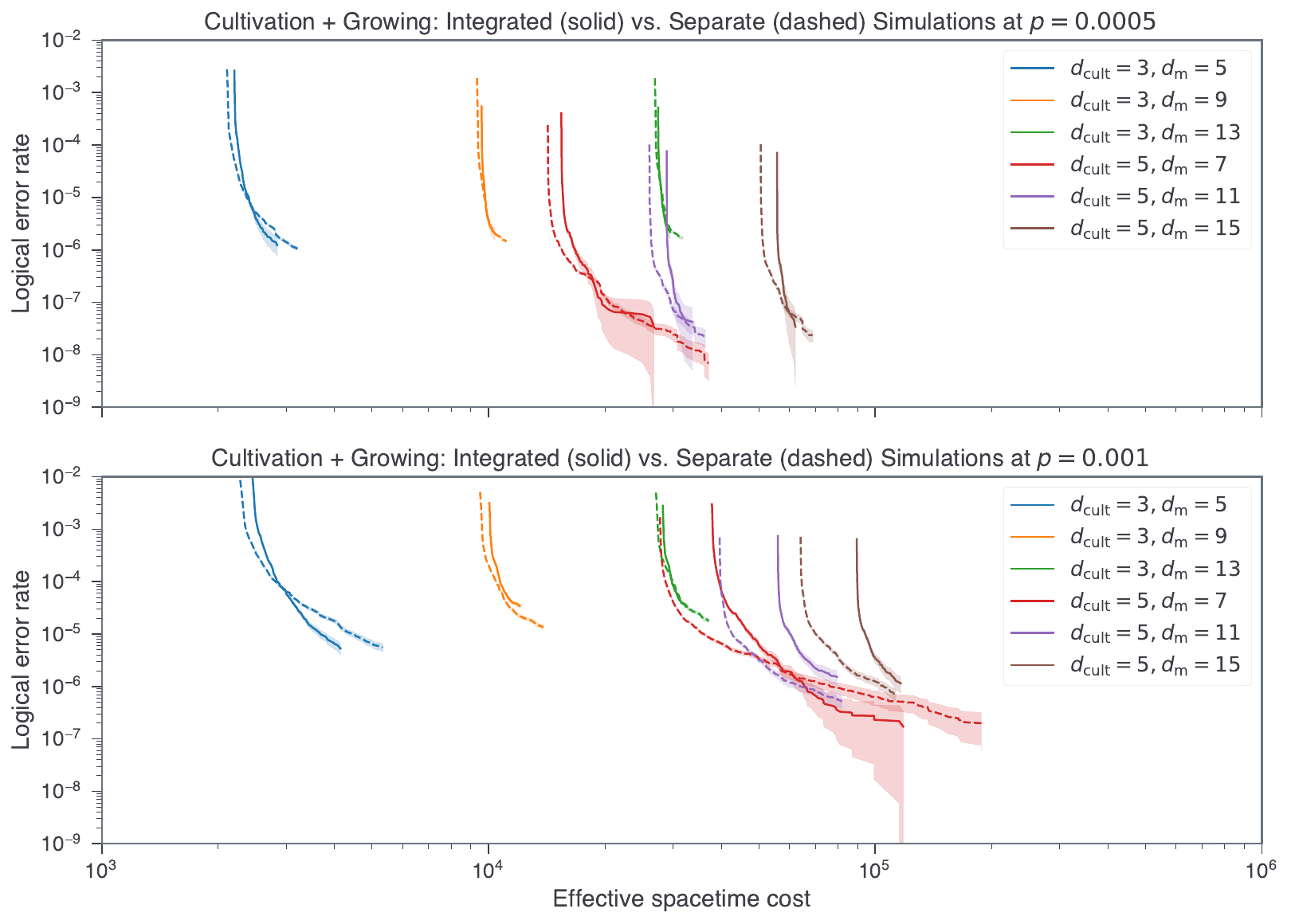}
	\caption{Logical error rate of the cultivation + growing circuit, plotted against the effective spacetime cost for $p \in \qty{5 \times 10^{-4}, 10^{-3}}$ and various combinations of $\dcult$ and $\dm$. For comparison, integrated and separate scenarios are represented by solid and dashed lines.}
	\label{fig:integrated_cult_growing_spacetime_cost}
\end{figure*}

The time cost analyzed in Fig.~\ref{fig:integrated_cult_growing_time_cost} is the one that directly affects the cost of our cultivation-MSD scheme, as it affects $T_\mr{m}$ (average number of rounds between successive stages).
The discrepancy in the time costs between the integrated and separate scenarios is particularly notable for $(p, \dcult) = (10^{-3}, 5)$, where they differ by at most about 1.7 times, implying that the spacetime cost of the cultivation-MSD scheme can be increased by a similar factor.
Additionally, the minimum achievable logical infidelity of the cultivation-MSD scheme also can be affected.
Although the large confidence intervals for the logical error rates in Fig.~\ref{fig:integrated_cult_growing_time_cost} make precise estimation difficult, we conjecture that the minimum infidelity could increase by approximately 3–30 times for both $p = 5 \times 10^{-4}$ and $10^{-3}$. This conjecture is based on the observation that the minimum logical error rates for the integrated scenario are roughly 1.5–3 times higher than those for the separated scenario when $\dcult = 5$ and $\dm \in \qty{11, 15}$.

The spacetime cost analyzed in Fig.~\ref{fig:integrated_cult_growing_spacetime_cost} is not directly related to our resource analysis for MSD, but we place the figure here for readers who are interested in a `color-code-only' cultivation scheme that does not involve conversion to a surface code unlike the scheme in Ref.~\cite{gidney2024magic}.

The codes we have used for the above simulations are available in the \textit{color-code-stim} package \cite{github:colorcodestim}.

\bibliography{bibliography}

\end{document}